\documentclass[10pt,reqno,oneside]{amsart}

\usepackage[T1]{fontenc}
\usepackage{amsmath}
\usepackage{amsfonts}
\usepackage{amsthm}
\usepackage{amssymb}
\usepackage{graphicx}
\usepackage{hyperref}
\hypersetup{hidelinks}
\usepackage{cite}
\usepackage{calc}
\usepackage{geometry}

\usepackage{tikz}
\usepackage{tikz-cd}
\usetikzlibrary{%
  matrix,%
  calc,%
  arrows%
}

	\newcommand{\ncd}{\newcommand}
	\ncd{\mrm}    {\mathrm}
	\ncd{\beq} {\begin{equation}}
	\ncd{\eeq} {\end{equation}}
	\ncd{\nn}{\nonumber}

	\def\d{{\rm d}}
	\def\D{{\rm D}}

	\def\tpm{T_p M}
	
	\def\Reals{ \mathbb{R}}
	\def\basis[#1]{\frac{\partial}{\partial #1}}
	\def\dt[#1]{\frac{\d}{\d #1}}
	\def\propagator[#1,#2,#3]{P_{#1}\left(#2,#3 \right)}
	\def\ipropagator[#1,#2,#3]{P^{-1}_{#1}\left(#2,#3 \right)}
	\def\GTM{\Gamma\left(TM \right)}
	\def\GTsM{\Gamma\left(T^*M \right)}

	\newtheorem{prop}{Proposition}
	\newtheorem{theo}{Theorem}
	\newtheorem{corollary}{Corollary}
	\newtheorem{lemma}{Lemma}
	\newtheorem{mydef}{Definition}

\begin{document}

\title[How Much Geometry Does Newtonian Dynamics Fix?]{How Much Geometry Does Newtonian Dynamics Fix?}

\author{C\'esar S. L\'opez-Monsalvo}
\address{Departamento de Ciencias B\'asicas, Universidad Aut\'onoma
Metropolitana -- Azcapotzalco, Avenida San Pablo 420, Colonia Nueva El
Rosario, Azcapotzalco 02128, Ciudad de M\'exico, M\'exico}
\email{cslm@azc.uam.mx}
\urladdr{https://orcid.org/0000-0002-0378-0415}



\begin{abstract}
Newtonian mechanics is usually written on a spacetime of fixed geometry. Here we build that geometry, introducing a propagator and a bilinear form only as the description demands. Within this reconstruction the Principle of Inertia emerges as a theorem rather than an independent postulate. Inertial motion is blind to the scale of the form and to the torsion the propagator defines, and that blindness is a geometric precursor of the \emph{universality of free fall}, before gravity enters. Newton's Second Law is a compatibility between two derivatives along a curve. The force exceeds the mass times the lowered acceleration by an exact term, and the law survives arbitrary torsion. Requiring force-free motion to be inertial along curves parametrised by arc length restricts the torsion to be totally antisymmetric. The remaining torsion is a three-form which we postulate the particle carries along its worldline. Its contraction with the velocity annihilates that velocity, so the supplementary condition of spinning-particle mechanics holds identically, and we propose that contraction as a classical spin. The algebra is familiar from relativistic spin hydrodynamics; the assignment, to the best of our knowledge, has not been proposed elsewhere. Since inertial motion resolves neither the mass nor the spin, the predictive content lies in the transport law we postulate, which turns the spin against a parallel-transported frame at a rate fixed by the acceleration alone and by no coupling constant. That rate becomes a laboratory number only once the arc-length parameter is identified with a clock, which we leave open. The dynamics asks non-degeneracy of the form, so this spacetime is not Newton--Cartan's Galilean one, and gravitation lies outside the construction. We prove every mathematical result, make three constitutive postulates and one physical conjecture.
\end{abstract}

\maketitle

\tableofcontents

\section{Introduction}

\emph{Where does the spacetime of Newtonian dynamics come from?} The question is seldom put, for the theory is ordinarily set down upon one already equipped -- a manifold, an affine connection, a bilinear form, a class of inertial trajectories -- and the physical warrant for each of these is taken for granted. Which of those structures do the dynamical principles fix exactly, which do they select only with the help of a constitutive assumption, and which do they leave open? We reconstruct that geometry from physical principles, and what the dynamics leaves undetermined is then a result of the construction.

Let us begin by depriving the space on which the motion of the particles takes place of every mathematical structure other than the one required to make sense of the notion of smoothness, i.e. differentiability. This strategy of building geometry upward from the barest primitives has been pursued in its own right~\cite{maudlin2014new}. We thus consider an $n$-dimensional differentiable manifold $M$. A point on a curve denotes an event in the history of a particle, so that $M$ accounts for space and time together. In this sense, we will refer to $M$ as the \emph{spacetime} manifold. For the moment, this is merely a name, with no reference to the metric theory given by relativity. 

We will proceed through independent explorations of the various concepts defining Newtonian dynamics. We shall find that the link between the notions of force and acceleration is nothing but the relation between the Lie and the covariant derivative, taken with respect to an affine connection compatible with a metric for the spacetime.%

In this work we provide that construction and develop its consequences in full mathematical rigour. To the manifold we add only two further structures, a law of transport along its curves and a bilinear form, and we require nothing more of either, neither symmetry nor non-degeneracy nor signature for the form, and no torsion condition for the connection. The kinematics asks nothing further. The dynamics later asks non-degeneracy of the form, since the musical maps make the momentum the dual of the velocity. We derive the affine connection from the transport, and within this reconstruction the Principle of Inertia emerges as a theorem rather than being introduced as an independent postulate, for the symmetries of inertial motion are exactly the affine reparametrisations of the worldline, and these alone (Theorem~\ref{theo.inertia}). In particular, inertial motions through every event and direction do not require the connection to be compatible with the bilinear form; they require only the weaker condition that the form be a Killing tensor of it (Proposition~\ref{prop.killing}). Full compatibility is a separate demand -- that of transporting standards of magnitude unchanged -- and the dynamics supplies a third, the identification of force with the change of momentum. The three together fix exactly the symmetric part of the connection (Theorem~\ref{theo.nsl}), and the Levi-Civita connection is the canonical torsion-free representative, selected when we retain no additional constitutive torsion. A fourth demand, that force-free motion and inertial motion agree, then narrows what is left over without extinguishing it.

What inertia determines it determines exactly. However, what it leaves undetermined are two constitutive ambiguities, precisely the two data that the structures of spacetime carry beyond the inertial structure. The geometry of inertial motion depends on the bilinear form only up to a global scale, since multiplying the form by a positive constant leaves every inertial motion unchanged (Proposition~\ref{prop.homothety}). It depends on the connection only through its symmetric part, and the antisymmetric part is invisible to every inertial trajectory (Proposition~\ref{prop.torsion}). These two blindnesses are one principle, a geometric precursor of the \emph{universality of free fall} obtained before any notion of gravity has been introduced (Theorem~\ref{theo.ise}). We resolve each ambiguity afterwards by a constitutive choice. The scale discloses the inertial mass of the particle (Definition~\ref{def.mass}), and the antisymmetric part, which the dynamics restricts to a $3$-form, discloses a candidate for its spin (Definition~\ref{def.spin}). 

Dynamics is the metric dual of kinematics, and we generate it from a single constitutive relation. The momentum is the constitutive dual of the velocity under the bilinear form (Definition~\ref{def.mom}), and force descends from it. The whole dynamical content therefore springs from the one constitutive datum that dynamics adds to kinematics.

Newton's Second Law is the compatibility of the two derivatives the construction affords. The force is the rate of change of the momentum measured by the Lie derivative. That derivative acts upon a velocity field, so the motions fill a region and a particle here is one worldline of a congruence. It differs from the mass times the lowered acceleration, measured by the covariant derivative, by an exact term which a suitable choice of clock removes (Theorem~\ref{theo.nsl}). For a uniform motion the exact term vanishes, and we find that the dynamical notion of a force-free particle coincides with the kinematic notion of an inertial one (Corollary~\ref{cor.forcefree}). The identity survives an arbitrary torsion at the cost of one further term quadratic in the velocity (Theorem~\ref{theo.nslgen}), and the vanishing of that term is exactly the condition under which force-free and inertial motion continue to agree (Corollary~\ref{cor.forcefreegen}). That fourth demand restricts the second ambiguity instead of annihilating it, and it spares exactly the torsions whose lowered form is totally antisymmetric (Corollary~\ref{cor.antisym}). That condition is already familiar from the geometry of metric-affine gravity~\cite{hehl1976general,trautman2006einstein}, and we reach it here from the Second Law together with that demand and with the arc-length parametrisation it presupposes. We thus exhibit torsion-freeness, customarily imposed for convenience, as the canonical representative of a physical demand. Of the choices within the residual freedom, it alone requires no further tensor field upon the spacetime manifold. The torsions the demand spares are those from which we build the second constitutive relation.

That residual freedom is a $3$-form. Interestingly, a counting of degrees of freedom rules out one of the two places it might live. A form fixed once and for all upon $M$ would afford a particle a spin determined by its velocity and a single number, where three are needed, so we carry the constitutive torsion along the worldline instead (Definition~\ref{def.consttorsion}). Its contraction with the velocity is a $2$-form which annihilates that velocity identically (Definition~\ref{def.spin}), so the supplementary condition that classical spinning-particle mechanics must postulate is here a theorem (Proposition~\ref{prop.frenkel}). We propose Fermi--Walker transport as the constitutive transport law the carried $3$-form obeys, and it is a postulate. Parallel transport does not preserve the magnitude of the spin under acceleration. The Fermi--Walker propagator preserves it, preserves the supplementary condition, and reduces to the parallel propagator along inertial motions. Along an accelerated motion the two turn the spin differently, in the plane of the velocity and the acceleration, at a rate fixed by the acceleration alone and by no coupling constant. We assemble the generator of that transport from the bilinear form, the velocity, and the acceleration the connection assigns to that velocity. The minimality of the construction therefore survives the second constitutive relation (Section~\ref{sec.spin}). Inertial motion is then independent of both the mass and the spin of the body that executes it. The first independence becomes the universality of free fall once a gravitational coupling exists. The second says only what Theorem~\ref{theo.ise} guarantees, that inertial motion is blind to the spin. The predictive content lies in the law by which the carried $3$-form is transported.

The question of what an observation of motion can fix, and what it must leave free, organises the axiomatic programme of Ehlers, Pirani and Schild~\cite{ehlers1972geometry}, where light propagation fixes the conformal class of the metric and free fall its projective structure. Malament has since given that recovery of structure from observation a rigorous treatment~\cite{malament2012topics}, and Wheeler has lately examined the torsion-free connection reached there instead of assuming it, showing that torsion-freeness follows from restricting the reparametrisations admitted, as opposed to the geodesic data alone~\cite{wheeler2025geometry}. The torsion his enlarged class generates is of the trace type and the residue retained here is totally antisymmetric, so the two occupy complementary parts of one freedom. Theorem~\ref{theo.ise} puts the same question to a smaller structure. Inertial motion depends on the pair $(g,\nabla)$ only through the homothety class of $g$ and the symmetric part of $\nabla$. No observation of a particle in inertial motion therefore resolves either the scale of the bilinear form or the torsion of the connection, and we reach that result before any notion of gravity enters the construction. The equivalence principle, in its familiar non-relativistic form, is a statement about gravitation, and a geometric reading in which the Newtonian potential is traded for a curved, time-dependent spatial metric forces the equality of gravitational and inertial mass by that very requirement~\cite{kapustin2021nonrelativistic}. That statement cannot be posed until a gravitational coupling exists. Theorem~\ref{theo.ise} is a precursor obtained before one exists, and it asserts only that the inertial structure is blind to scale and to torsion.

Nevertheless, torsion more often enters physics by a different door. Einstein--Cartan theory sources it in the spin of matter through a field equation~\cite{sciama1964physical,kibble1961lorentz,hehl1976general}. That equation is algebraic rather than differential, and the torsion stands pointwise in a fixed proportion to the spin density, with a coupling constant of gravitational strength. The field therefore neither propagates nor survives outside the matter producing it. Metric-affine gravity treats the same object on a broader footing~\cite{hehl1995metric,hehl2007cartan,trautman2006einstein}.

We propose nothing of that kind since the two constructions can be mistaken for one another at a glance. This paper writes down no field equation relating a torsion to a spin density, whether algebraic or differential, and no coupling constant by which such a relation might be scaled. Wheeler has catalogued in detail what a field theory must supply in order to source a torsion~\cite{wheeler2023sources}, and we supply none of it. The logical direction is the reverse of Einstein--Cartan's. There the spin of matter is given and the torsion is what that spin produces. That torsion is ordinarily approached through its couplings to matter instead of through the kinematics of a free particle is itself a feature of the subject~\cite{shapiro2002physical}. Here the torsion comes first, and nothing produces it. It is the residual freedom in the connection which inertial motion cannot resolve (Theorem~\ref{theo.ise}). We read a geometric freedom physically, and we propose a constitutive relation of the same kind as the one that discloses the inertial mass. The differential geometry of skew torsion is older than any physical use of it, and the algebraic fact we build upon already belongs there. That a metric connection preserves the geodesics of its metric precisely when its torsion is totally antisymmetric is Corollary 2.1 of Agricola and Friedrich~\cite{agricola2004holonomy}. When that torsion is further required to be parallel, the manifolds admitting it form a short, classified list, the naturally reductive homogeneous spaces and the nearly K\"ahler, nearly parallel $G_2$, Sasakian and $3$-Sasakian geometries among them~\cite{cleyton2021metric}. Both facts are theirs. What we add to the first is where it comes from, for we reach it from the Second Law together with the demand that force-free and inertial motion agree. The second we use negatively, to weigh what a torsion fixed upon the manifold would cost against one carried along the worldline (Definition~\ref{def.consttorsion}).

The paper's central contribution is best stated as a chain of results and a subsequent physical proposal. Let us set out the chain first. Newton's Second Law survives an arbitrary torsion, and demanding that force-free motion coincide with inertial motion restricts the lowered torsion to be totally antisymmetric (Corollary~\ref{cor.forcefreegen}, Corollary~\ref{cor.antisym}). Such a torsion is a $3$-form, and the degrees-of-freedom count, together with the classification above~\cite{cleyton2021metric}, rules out a form fixed upon the manifold, so we take it as a datum carried along the worldline (Definition~\ref{def.consttorsion}). Its contraction with the velocity annihilates that velocity identically, and does so as an identity of the interior product, as opposed to a further demand (Proposition~\ref{prop.frenkel}). The proposal comes after the chain, and it is physical. We propose that contraction as the spin of the particle (Definition~\ref{def.spin}). How much of this is new? The algebraic characterisation in the middle is in print, as recorded above, and so is the pairing of the velocity condition with Fermi--Walker transport, which relativistic mechanics derives for a torque-free gyroscope~\cite{costa2015center}. Nor is the algebra of the last link new. Relativistic spin hydrodynamics constructs a spin density by contracting a totally antisymmetric spin tensor with the fluid velocity, and obtains the supplementary condition from that antisymmetry alone~\cite{fang2026relativistic,hongo2021relativistic,cao2022gyrohydrodynamics}. What has not been proposed elsewhere, to the best of our knowledge, is the assignment rather than the algebra. In those treatments the totally antisymmetric object is a spin current belonging to matter, whereas here it is the residual freedom in the connection which inertial motion cannot resolve. The condition itself is the Frenkel--Mathisson--Pirani condition of relativistic spinning-particle mechanics, one of a family among which that mechanics must choose~\cite{costa2015center}, and one which the constitutive theory of a spinning fluid lists among its postulates~\cite{herrmann2000constitutive}. Here we have nothing to choose and nothing to postulate, since the condition follows identically from the definition of the spin. Moreover, that condition is infinitely degenerate in its relativistic setting, where an entire disc of centroids satisfies it and generates helical motions, for a centroid is a moment of an extended matter distribution and every observer computes a different one~\cite{costa2015center}. We meet no such disc, since the worldline is primitive instead of a centroid, and the spin is a constitutive datum carried along it.

Throughout, we keep three registers apart, and we say at each point which of them is in force. We prove every mathematical result from the structures declared, and we smuggle nothing physical into a proof. We make three constitutive postulates, and declare each one where it is made. There is one physical conjecture. We label it as such, and it carries none of the weight we ask the geometry to bear.

Let us say what we do not attempt. The musical maps make the momentum the dual of the velocity, and they require the bilinear form to be non-degenerate. Our spacetime is therefore not the Galilean spacetime of the Newton--Cartan formulation, which replaces the metric by a degenerate pair and geometrises gravitation in the curvature of a connection that the pair alone does not determine~\cite{cartan1923varietes,kunzle1972galilei,malament2012topics}. The philosophical literature has examined which parts of such a structure are essential and what is at stake in choosing between formulations of the same Newtonian content~\cite{earman1989world,knox2011newton}. Torsional Newton--Cartan geometry has been built at length~\cite{bergshoeff2015torsional,banerjee2016torsional,bergshoeff2017newton}, and the torsion it carries is not ours. There it denotes the failure of the clock $1$-form to be closed, the absence of an absolute time. The two share a word and little else, a distinction the literature has lately drawn for itself~\cite{march2024remarks}. Indeed, the two constructions answer different questions. Newton--Cartan theory asks how a gravitational interaction already given is to be written upon a spacetime Galilean from the outset. We ask what structure the notion of force requires before any interaction has been named, and we take the forces as data throughout. It is worth mentioning a point of contact. The connections compatible with a Galilean structure differ from one another by an antisymmetric form, as the admissible connections of Section~\ref{sec.spin} do, although the Galilean freedom is visible to the trajectories and ours is invisible to them. Whether the construction survives the degeneration of the bilinear form to a Galilean pair is open, and it is the first question we would put to it. The second is what the arc-length parameter measures, for the rate at which the spin turns against a parallel-transported frame becomes a laboratory number only once that parameter is identified with the reading of a clock. We withhold more than the Galilean structure. There is no field equation coupling the torsion to a source, no relativity, and no extended body whose centroid a supplementary condition must fix~\cite{puetzfeld2014equations}. In Section~\ref{sec.spin} we declare the identification of the velocity contraction of the constitutive torsion with a physical spin a conjecture rather than a proved fact. Dimensional analysis fixes its dimension and we supply no scale, as opposed to the couplings that make the relativistic and Einstein--Cartan literatures testable. Closing that gap, so the constitutive torsion of a Newtonian particle earns the physical content the word spin borrows from elsewhere, is a task for another work.

The manuscript is structured as follows. In Section~\ref{sec.kin} we build the kinematics, introducing each geometric notion at the point where the description of motion demands it. That section closes with the Principle of Inertia and the two freedoms it leaves. In Section~\ref{sec.dyn} we turn to the dual objects and obtain Newton's Second Law as a compatibility between the two derivatives. We then settle the two freedoms by a constitutive choice apiece. The first discloses the inertial mass, and the second the constitutive torsion whose contraction with the velocity we propose as a spin. In Section~\ref{sec.closing} we take stock of what is proved, postulated and conjectured.

\section{Kinematics}
\label{sec.kin}

Kinematics deals with the \emph{description} of the motion of a single particle, independently of its causes. This section aims at building the minimal mathematical structures required to realise such a description. In this manner, each structure is accompanied by a physical hypothesis so that the geometry of motion, encoded in the \emph{Principle of Inertia}, emerges as a theorem rather than being introduced as an independent postulate. We build these structures in the order the description demands. First an event and its coordinates, for without them nothing can be said to change; next a notion of velocity, constructed without borrowing the vector-space structure of $\mathbb{R}^n$ that $M$ does not possess; then a means of comparing vectors at different events -- the parallel propagator -- before an acceleration can be defined at all; and, only when the question of uniform motion forces it, a bilinear structure with which to measure magnitudes. Many of the definitions and results in this section are well known from a differential geometry point of view. Our goal is to present such ideas as the necessary and sufficient mathematical background to encode the \emph{physical realism} of the Galilean programme~\cite{mcmullin1985galilean}.

\noindent\textbf{Note for the reader.} The differential-geometric constructions developed in subsections~\ref{subsec.pos} to~\ref{subsec.td} are standard. Nevertheless, our reasons for presenting them explicitly are two-fold. On the one hand, the purpose of the paper is to construct the geometric framework of Newtonian dynamics, where a physical necessity gives rise to every mathematical structure, instead of our assuming the whole machinery from the outset. On the other hand, presenting them explicitly serves as a notational and referential guide, and makes the entire manuscript self-contained. Readers familiar with affine connections, torsion, Lie derivatives and the covariant differentiation of multilinear forms may proceed directly to Section~\ref{subsec.IM}.

\subsection{Position}
\label{subsec.pos}

Let $M$ be an $n$-dimensional differentiable manifold whose points represent the totality of places in space and time where a particle can be. Thus we shall refer to the points on $M$ as \emph{events}. The precise location of an event is in no way related to any structure of the manifold other than the \emph{topological}. Each event lies in the interior of at least one region of $M$ and it is with respect to this region that the notion of \emph{coordinates} is well defined. 

\begin{mydef}[Coordinate chart]
\label{def.chart}
A \emph{coordinate chart} around an event $p\in M$ is a pair $(U,\psi)$ where $p\in U \subset M$ and $\psi:U\subset M \longrightarrow \mathbb{R}^n$ is a \emph{homeomorphism}. 
\end{mydef}
Thus, the coordinate location of the event $p\in M$ in the chart $(U,\psi)$ is given by the ordered $n$-tuple
	\beq
	\psi(p) = \left(x^1_p,\cdots, x^n_p\right) \in \mathbb{R}^n.
	\eeq

Coordinates themselves do not bear any particular meaning. Indeed, the same event can be located with respect to another chart $(V,\phi)$, where $p\in V$, so that now
	\beq
	\phi(p) = \left(y^1_p,\cdots,y^n_p \right)\in \mathbb{R}^n,
	\eeq
is a different ordered $n$-tuple, yet it does represent the same point. This is achieved by demanding that the composite map $\phi \circ \psi^{-1}: \psi(U \cap V) \subset \mathbb{R}^n \longrightarrow \phi(U \cap V) \subset \mathbb{R}^n$ be a \emph{homeomorphism}, so that the two coordinate descriptions are topologically consistent, that is, the diagram
\beq
	\label{diag.charts}
	\begin{tikzpicture}[]
	\matrix[matrix of math nodes,column sep={100pt,between origins},row
	sep={70pt,between origins},nodes={asymmetrical rectangle}] (s)
	{
	|[name=a1]| 									& |[name=a2]| (U \cap V)\subset M		& |[name=a3]|									\\
	|[name=b1]| \psi(U \cap V)\subset\mathbb{R}^n 	& |[name=b2]| 						& |[name=b3]| \phi(U \cap V)\subset\mathbb{R}^n	 \\
	}
	;
	\draw[->] 	
			(a2) edge node[above] {\(\psi \)} (b1)
			(a2) edge node[above] {\(\phi\)} (b3)
			(b1) edge[bend right=0] node[below] {\(\phi\circ \psi^{-1}\)} (b3)
	;
	\end{tikzpicture}
	\eeq
commutes. %

 The explicit form of the change of coordinate maps is given by
	\beq
	\label{eq.coordch1}
	\left(\phi \circ\psi^{-1}\right)\left[x^1,\dots,x^n \right] = \left[y^1\left(x^1,\dots,x^n \right),\dots,y^n\left(x^1,\dots,x^n \right) \right],
	\eeq
and
	\beq
	\label{eq.coordch2}
	\left(\psi \circ\phi^{-1}\right)\left[y^1,\dots,y^n \right] = \left[x^1\left(y^1,\dots,y^n \right),\dots,x^n\left(y^1,\dots,y^n \right) \right].
	\eeq

In kinematics, the image of a map $\gamma:I\subset \mathbb{R}\longrightarrow U\subset M$ is called a \emph{path} and, if we consider its parametrisation, it corresponds to the curve followed by a particle within the open set $U$. As we will shortly see, each parametrisation yields a different curve tracing the same path. In this way each point in the curve corresponds to an event in the history of a particle. Note that a particle cannot be represented by a single point, but by the whole curve it follows in $M$.

In order to express our geometric constructions in terms of local coordinates, consider now a parametrised curve $\gamma$, a chart $(U,\psi)$ and an event in the history of the particle $p=\gamma(t_0) \in U$ for some $t_0 \in I$. 
 The map $\gamma^U:I\subset\mathbb{R}\longrightarrow \psi(U)\subset \mathbb{R}^n$ is a curve defined by the composition $\gamma^U = \psi \circ \gamma$ and provides us with the computational means to exploit the results from multivariable calculus and \emph{lift them up} to the differentiable manifold through the diagram
	\beq
	\label{diag.curve}
	\begin{tikzpicture}[]
	\matrix[matrix of math nodes,column sep={100pt,between origins},row
	sep={70pt,between origins},nodes={asymmetrical rectangle}] (s)
	{
	|[name=a1]| I\subset\mathbb{R}	& |[name=a2]| U\subset M 		\\
	|[name=b1]| 				 	& |[name=b2]| \psi(U) \subset \mathbb{R}^n,	 \\
	}
	;
	\draw[->] 	
			(a1) edge node[left] {\(\gamma^U=\psi \circ \gamma\)} (b2)
			(a2) edge node[auto] {\(\psi\)} (b2)
			(a1) edge node[auto] {\(\gamma\)} (a2)
	;
	\end{tikzpicture}
	\eeq
where $(\psi \circ \gamma)(t)= \gamma^U(t) = \left[x^1(t),\cdots, x^n(t) \right] $ represents the coordinates traced by the curve $\gamma$ in the local chart $(U,\psi)$. In this setting, we have enough structure to have a well-defined notion of velocity on $\mathbb{R}^n$. Extending such a notion to the manifold itself is of fundamental relevance.

  To fully appreciate this, first let us consider the more familiar tangent vector to the curve $\gamma^U\subset\mathbb{R}^n$ at $\psi(p)$, which can be obtained from the limit process
	\beq
	\label{eq.vectorvelocity}
	\vec{v}\vert_{\psi(p)} \equiv \lim_{t\rightarrow t_0} \frac{\vec r(t) - \vec{r}(t_0)}{t-t_0} = \lim_{t\rightarrow t_0} \sum_{i=1}^n\left[\frac{x^i(t) - x^i(t_0)}{t-t_0}\right]\hat e_{(i)} = \sum_{i=1}^n\left.\frac{\d }{\d t}x^i(t)\right\vert_{t=t_0}\hat e_{(i)},
	\eeq
where $\vec r(t) = x^1(t) \hat e_{(1)} + \cdots+ x^n(t)\hat e_{(n)}$ is the \emph{position vector} of the points along the curve $\gamma^U\subset\mathbb{R}^n$ whose coordinates are $\psi\left(\gamma\right)= \gamma^U(t) =\left[x^1(t),\cdots, x^n(t) \right]$ and the ordered set of vectors $\{\hat e_{(1)},\cdots,\hat e_{(n)}\}$ is a basis for $\mathbb{R}^n$. Although some literature interchanges them, one shall never confuse the \emph{auxiliary} notion of the position vector with the coordinates of a point.

\subsection{Velocity}

\emph{What kind of object is the velocity at an event $p=\gamma(t_0)\in M$?} Now we cannot resort to the same strategy of considering the difference of nearby position vectors since, in general, $M$ is not a vector space. Instead, we adopt a more constructive approach that allows us to take the limit of a well-defined object. To this end, let us consider a real valued function $f:U\subset M \longrightarrow \mathbb{R}$ so that the composite map $f\circ \gamma:I\subset \mathbb{R} \longrightarrow \mathbb{R}$ is a real function of the parameter $t\in I\subset \mathbb{R}$. Thus, its derivative at $t=t_0$ is
	\beq
	\label{eq.derivativef}
	\left.\frac{\d}{\d t} \left[f\circ \gamma\right](t)\right\vert_{t=t_0} \equiv \lim_{t\rightarrow t_0} \frac{f\left[\gamma(t)\right] - f\left[\gamma(t_0) \right]}{t-t_0} \in \mathbb{R}.
	\eeq
The definition itself is not very useful for explicit calculations. We can use a local chart $(U,\psi)$ to provide an operational definition for expression \eqref{eq.derivativef} in terms of the corresponding function $f_U \equiv f \circ \psi^{-1}:\psi(U) \subset \mathbb{R}^n \longrightarrow \mathbb{R}$ by noting that $f_U \circ \gamma^U: I \subset \mathbb{R} \longrightarrow \mathbb{R}$ satisfies
	\beq
	\label{eq.equivmaps}
	f_U \circ \gamma^U = \left(f \circ \psi^{-1}\right) \circ \left(\psi \circ \gamma \right) = f \circ \left(\psi^{-1} \circ \psi \right)\circ \gamma = f \circ \gamma,
	\eeq
which can be expressed by
	\beq
	\label{diag.comp}
	\begin{tikzpicture}[]
	\matrix[matrix of math nodes,column sep={100pt,between origins},row
	sep={70pt,between origins},nodes={asymmetrical rectangle}] (s)
	{
	|[name=a1]| 					& |[name=a2]| U\subset M 		& 					\\
	|[name=b1]| I\subset\mathbb{R} 	& |[name=b2]| \psi(U) \subset \mathbb{R}^n	&	  |[name=b3]| \mathbb{R}.		\\
	}
	;
	\draw[->] 	
			(b1) edge node[auto] {\(\gamma^U\)} (b2)
			(a2) edge node[right] {\(\psi\)} (b2)
			(b1) edge node[auto] {\(\gamma\)} (a2)
			(a2) edge node[auto] {\(f\)} (b3)
			(b2) edge node[auto] {\(f_U \)} (b3)
			(b1) edge[bend right=30] node[below] {\(f_U \circ \gamma^U \)} (b3)
	;
	\end{tikzpicture}
	\eeq
It follows that
	\beq
	\label{eq.dtf}
	\left.\dt[t]\left[f_U\circ \gamma^U\right](t)\right\vert_{t=t_0} 	 = \left. \dt[t] f_U \left[x^1(t),\cdots,x^n(t) \right]\right\vert_{t=t_0}
									 = \left. \sum_{i=1}^n\dt[t]x^i(t)\right\vert_{t=t_0}\left. \basis[x^i] f_U\right\vert_{\psi(p)}.
	\eeq
Therefore, the coordinate expression for \eqref{eq.derivativef} is given by
	\beq
	\label{eq.vellocal}
	\left.\frac{\d}{\d t} \left[f \circ \gamma \right](t)\right\vert_{t =t_0} = \sum_{i=1}^n \left.\frac{\d}{\d t} x^i(t)\right\vert_{t=t_0} \left.\frac{\partial}{\partial x^i} \left(f \circ \psi^{-1}\right)\right\vert_{\psi(p)}.
	\eeq
This is reminiscent of the usual directional derivative. However, note that we have arrived at this point using only the differential structure of $M$. Thus, the velocity of a particle following a trajectory defined by $\gamma \subset U\subset M$ at the point $p=\gamma(t_0)\in U \subset M$ can be understood as a differential operator acting on functions.
	
	\begin{mydef}[Particle velocity]
	\label{def.velocity}
	Let $M$ be a differentiable manifold, $(U,\psi)$ a local chart around $p=\gamma(t_0)\in U$ where $\gamma:I\subset \mathbb{R} \longrightarrow M$ is a smooth curve. Let $f:U\subset M \longrightarrow \mathbb{R}$ be a smooth real valued function. The velocity of the particle at $p=\gamma(t_0)$ is the directional derivative differential operator whose action on $f$ is given by 
	\beq
	\label{eq.velocity}
	 v\vert_p \left[f \right] \equiv \left.\dt[t] \left[f\circ\gamma \right](t)\right\vert_{t=t_0}.
	\eeq
	\end{mydef}
	
The reader may ponder on the significance of the function $f$ we require in Definition~\ref{def.velocity}. In our construction, we do not have any structure on the manifold allowing us to subtract the positions of two neighbouring events -- as we did in $\mathbb{R}^n$ through the position vector \eqref{eq.vectorvelocity} -- for $M$ carries no affine, let alone vector, structure. What the differentiable structure does place at our disposal is the collection of smooth scalar fields and, with them, the unambiguous operation of differencing the \emph{values} a field takes along the motion. The function $f$ is therefore not a constituent of the velocity but a \emph{probe}. The velocity is the object that reports, for every scalar field, its instantaneous rate of change along the curve, and it is characterised in full by this action. This is the precise sense in which we have traded the ill-defined difference of nearby points for the limit of a well-defined real number.

The velocity \eqref{eq.velocity} was framed without reference to any chart, since the composite $f\circ\gamma$ is a real function of a real variable whose derivative is a number. Its independence of the chart is therefore a property of the construction rather than a hypothesis about the world, and to postulate it would be to postulate a tautology. What does carry content is the weaker demand that this chart-free number be \emph{computable}, that is, that it admit a coordinate representation of the form \eqref{eq.vellocal} in every chart of the atlas. Coordinates are not themselves relevant but a mere computational tool, and the demand is that the tool be available everywhere. To see what structure that demand entails, let us consider another chart $(V,\phi)$ such that $p=\gamma(t_0) \in U\cap V$ together with the diagram
	\beq
	\label{diag.charind}
	\begin{tikzpicture}[]
	\matrix[matrix of math nodes,column sep={100pt,between origins},row
	sep={70pt,between origins},nodes={asymmetrical rectangle}] (s)
	{
	|[name=a1]| 					& |[name=a2]| \psi(U\cap V) \subset \mathbb{R}^n 		& 									\\
	|[name=b1]| I\subset\mathbb{R} 	& |[name=b2]| U\cap V \subset M							&	  |[name=b3]| \mathbb{R}.		\\
	|[name=c1]| 					& |[name=c2]| \phi(U\cap V) \subset \mathbb{R}^n 		& 									\\
	}
	;
	\draw[->] 	
			(b1) edge node[auto] {\(\gamma^U\)} (a2)
			(b1) edge node[auto] {\(\gamma^V\)} (c2)
			(b1) edge node[auto] {\(\gamma\)} (b2)
			(b2) edge node[right] {\(\psi\)} (a2)
			(b2) edge node[right] {\(\phi\)} (c2)
			(b2) edge node[auto] {\(f\)} (b3)
			(a2) edge node[auto] {\(f_U\)} (b3)
			(c2) edge node[auto] {\(f_V\)} (b3)
	;
	\end{tikzpicture}
	\eeq
The coordinate expressions must satisfy
	\beq
	\left.\frac{\d}{\d t}\left[f_U \circ \gamma^U \right](t)\right\vert_{t=t_0} = \left.\frac{\d}{\d t}\left[f_V \circ \gamma^V \right](t)\right\vert_{t=t_0},
	\eeq
where, as before -- cf. equation~\eqref{eq.dtf} -- we have
	\beq
	\left.\frac{\d}{\d t}\left[f_V \circ \gamma^V \right](t)\right\vert_{t=t_0} = \sum_{i=1}^n \left.\frac{\d}{\d t} y^i(t)\right\vert_{t=t_0} \left.\frac{\partial}{\partial y^i} \left(f \circ \phi^{-1}\right)\right\vert_{\phi(p)}.
	\eeq
Indeed,
	\beq
	f_V \circ \gamma^V = \left(f_U \circ\psi\circ \phi^{-1} \right) \circ \left(\phi \circ\psi^{-1}\circ\gamma^U\right) = f_U \circ \gamma^U,
	\eeq
and using the change of coordinates \eqref{eq.coordch1} it follows that
	\begin{align}
	 & \sum_{i=1}^n \left.\frac{\d}{\d t} y^i(t)\right\vert_{t=t_0} \left.\frac{\partial}{\partial y^i} \left(f \circ \phi^{-1}\right)\right\vert_{\phi(p)}\nn\\
	 & \quad = \sum_{i=1}^n \left[\sum_{j=1}^n \left.\left.\left(\frac{\partial y^i}{\partial x^j}\right)\right\vert_{\left(\phi \circ\psi^{-1}\right)\left[x^k_p\right]}\frac{\d}{\d t} x^j(t)\right\vert_{t=t_0}\right] \left.\frac{\partial}{\partial y^i} \left(f \circ \phi^{-1}\right)\right\vert_{\phi(p)}\nn\\
	 	& \quad = \sum_{j=1}^n \left.\frac{\d}{\d t} x^j(t)\right\vert_{t=t_0} \left[\sum_{i=1}^n\left.\left(\frac{\partial y^i}{\partial x^j}\right)\right\vert_{\left(\phi \circ\psi^{-1}\right)\left[x^k_p\right]}\left.\frac{\partial}{\partial y^i}\left(f \circ \phi^{-1} \right)\right\vert_{\phi(p)}\right]\nn\\
		& \quad = \sum_{j=1}^n \left.\frac{\d}{\d t} x^j(t)\right\vert_{t=t_0} \left.\frac{\partial}{\partial x^j}\left(f \circ \psi^{-1} \right)\right\vert_{\psi(p)},
	\end{align}
from which we obtain
	\beq
	\left.\frac{\partial}{\partial x^j}\left(f \circ \psi^{-1} \right)\right\vert_{\psi(p)} = \sum_{i=1}^n\left.\left(\frac{\partial y^i}{\partial x^j}\right)\right\vert_{\left(\phi \circ\psi^{-1}\right)\left[x^k_p\right]}\left.\frac{\partial}{\partial y^i} \left(f \circ \phi^{-1} \right)\right\vert_{\phi(p)},
	\eeq
or, equivalently
	\beq
	\label{eq.jacobian}
		\left.\frac{\partial}{\partial y^j}\left(f \circ \phi^{-1} \right)\right\vert_{\phi(p)} = \sum_{i=1}^n\left.\left(\frac{\partial x^i}{\partial y^j}\right)\right\vert_{\left(\psi \circ\phi^{-1}\right)\left[y^k_p\right]}\left.\frac{\partial}{\partial x^i} \left(f \circ \psi^{-1} \right)\right\vert_{\psi(p)}.
	\eeq
The appearance of the components of the Jacobian matrix in \eqref{eq.jacobian} allows us to see that the change of coordinate maps -- equations \eqref{eq.coordch1} and \eqref{eq.coordch2} -- must be, in fact, diffeomorphisms. That is, the transition maps must be not only homeomorphisms but diffeomorphisms. This is the sense in which the differentiability of the transition maps is not a gratuitous convention but a necessary condition for our definition of the particle velocity to be coordinate independent.

The lesson concerns the status of coordinates rather than the motion of particles. A choice of chart is a matter of what Reichenbach called \emph{descriptive} rather than inductive simplicity~\cite{reichenbach1958philosophy}, a choice among equivalent descriptions which no observation can adjudicate, so that a demand of invariance under it restricts no physics -- the objection Kretschmann pressed against Einstein's demand for general covariance~\cite{kretschmann1917physikalischen,norton1993general}. What such a demand restricts is the formalism, and here it restricts it productively. The differentiable structure of $M$, assumed at the outset as the least that would give the word smoothness a meaning, is thus disclosed not as a convenience adopted in advance but as the exact price of insisting that a chart-free rate of change be expressible in charts at all. It is the one attribute of the spacetime manifold in this paper that the description forces upon us while forcing no physics, and every structure introduced hereafter -- the parallel propagator and the bilinear form alike -- will carry a physical hypothesis of its own.

\emph{How is this notion of velocity related to that of a tangent vector?} In $\mathbb{R}^n$ it is clear that $\vec{v}\vert_{\psi(p)}$ is indeed a vector [cf. equation~\eqref{eq.vectorvelocity}]. Now we will show that the differential operator \eqref{eq.velocity} itself lies in a vector space. 

	\begin{prop}
	\label{prop.vecspace}
	Let $T_p M$ be the set of directional derivatives at a point $p\in M$ whose action on functions is given by equation~\eqref{eq.velocity}. $T_p M$ is a vector space.
	\end{prop}
	\begin{proof}
	It is sufficient to show that $w\vert_p = v\vert_p + u\vert_p \in T_p M$ for any two directional derivatives $v\vert_p$ and $u\vert_p \in T_p M$ and that $\lambda v\vert_p \in T_p M$ for any value $\lambda \in \mathbb{R}$.
	
	Let $\gamma, \sigma:I\subset \mathbb{R}\longrightarrow U\subset M$ be two curves passing through the point $p\in U\subset M$ at the same time, i.e. $p = \gamma(t_0) = \sigma(t_0)$ for some value $t_0 \in I$. Note that the two curves might have been independently parametrised, however such a situation would only require an additional function mapping the corresponding intervals. Also, note that their coincidence at the point $p$ serves to \emph{synchronise} their parameters. Let $(U,\psi)$ be a coordinate chart such that the directional derivatives acting on a function $f:U\subset M\longrightarrow \mathbb{R}$ are expressed as
		\beq
		v\vert_{p}[f]=%
						\sum_{i=1}^n \left.\frac{\d}{\d t} x^i_\gamma (t)\right\vert_{t=t_0} \left.\frac{\partial}{\partial x^i}\left(f \circ \psi^{-1} \right)\right\vert_{\psi(p)},
		\eeq
and
	\beq
	u\vert_{p}[f]=%
						\sum_{i=1}^n \left.\frac{\d}{\d t} x^i_\sigma (t)\right\vert_{t=t_0} \left.\frac{\partial}{\partial x^i}\left(f \circ \psi^{-1} \right)\right\vert_{\psi(p)},
	\eeq
respectively. Clearly $v\vert_p$ and $u\vert_p \in T_p M$. Here, as before [cf. diagram~\eqref{diag.curve}], we have the corresponding curves on $\psi(U)\subset \mathbb{R}^n$
	\beq
	\gamma^U \equiv \left[ \psi \circ \gamma\right] (t) = \left[x^1_\gamma(t),\cdots,x^n_\gamma(t) \right] \quad \text{and} \quad \sigma^U \equiv \left[ \psi \circ \sigma\right] (t) = \left[x^1_\sigma(t),\cdots,x^n_\sigma(t) \right],
	\eeq
together with their associated position vectors
	\beq
	\label{pro.posvects}
	\vec r_\gamma(t) = \sum_{i=1}^n x^i_\gamma (t)\ \hat e_{(i)} \quad \text{and} \quad \vec r_\sigma(t) = \sum_{i=1}^n x^i_\sigma (t)\ \hat e_{(i)}.
	\eeq
Let us adjust the coordinates of the point $p\in U\subset M$ so that they coincide with the origin, i.e.
	\beq
	\psi(p) = \psi\left[\gamma(t_0) \right] = \psi\left[\sigma(t_0) \right] = \vec 0 \in \mathbb{R}^n.
	\eeq
Thus, the position vector given by the sum of \eqref{pro.posvects}
	\beq
	\vec r_\lambda \equiv \vec r_\gamma + \vec r_\sigma = \sum_{i=1}^n \left[x^i_\gamma(t) + x^i_\sigma (t)\right] \ \hat e_{(i)}
	\eeq
satisfies that $\vec r_\lambda (t_0) = \vec 0$. That is, we have the curve $\lambda^U:I\subset \mathbb{R}\longrightarrow \psi(U)\subset\mathbb{R}^n$ given by
	\beq
	\lambda^U = \left[x^1_\lambda(t),\cdots,x^n_\lambda(t) \right] = \left[x^1_\gamma(t) + x^1_\sigma(t),\cdots,x^n_\gamma(t)+x^n_\sigma(t) \right]
	\eeq
	with $\lambda^U(t_0)=\psi(p)$ and whose velocity at the point $\psi(p)$ is given by
		\beq
		\vec w\vert_{\psi(p)} = \sum_{i=1}^n \left.\frac{\d}{\d t} x^i_\lambda(t)\right\vert_{t=t_0} \hat e_{(i)} = \sum_{i=1}^n \left.\frac{\d}{\d t}\left[x^i_{\gamma}(t) + x^i_\sigma(t) \right]\right\vert_{t=t_0} \hat e_{(i)}.
		\eeq
	Therefore,
		\beq
		\begin{split}
		w\vert_p[f] = v\vert_p[f] + u\vert_p[f] & = \sum_{i=1}^n \left.\frac{\d}{\d t}\left[x^i_{\gamma}(t) + x^i_\sigma(t) \right]\right\vert_{t=t_0} \left.\frac{\partial}{\partial x^i}\left(f \circ \psi^{-1} \right)\right\vert_{\psi(p)}\\
		& = \left.\frac{\d}{\d t} \left[f \circ \lambda \right](t)\right\vert_{t=t_0}
		\end{split}
		\eeq
is the directional derivative along the curve $\lambda:I\subset\mathbb{R}\longrightarrow U \subset M$ given by $\lambda \equiv \psi^{-1} \circ \lambda^U$ at the point $p=\lambda(t_0)$. Hence, $w\vert_p = v\vert_p + u\vert_p\in T_p M$.

Now, let us consider a reparametrisation of the curve $\gamma$, namely $\tilde\gamma:I'\subset \mathbb{R} \longrightarrow U\subset M$ together with a bijective and differentiable function $\tau:I\subset \mathbb{R} \longrightarrow I'\subset \mathbb{R}$ so that the diagram
	\beq
		\label{diag.reparcurve}
	\begin{tikzpicture}[]
	\matrix[matrix of math nodes,column sep={100pt,between origins},row
	sep={70pt,between origins},nodes={asymmetrical rectangle}] (s)
	{
	|[name=a1]| I\subset\mathbb{R}		& |[name=a2]| 					U \subset M		& 	|[name=a3]|	 I'\subset\mathbb{R}			\\
	|[name=b1]|						 	& |[name=b2]| \psi(U) \subset \mathbb{R}^n		&	|[name=b3]| 								\\
	}
	;
	\draw[->] 	
			(a1) edge node[auto] {\(\gamma\)} (a2)
			(a3) edge node[above] {\(\tilde\gamma\)} (a2)
			(a3) edge[bend right=30] node[above] {\(s\)} (a1)
			(a2) edge node[auto] {\(\psi\)} (b2)
			(a1) edge node[below] {\(\gamma^U \)} (b2)
			(a3) edge node[auto] {\(\tilde \gamma^U \)} (b2)
	;
	\end{tikzpicture}
	\eeq
commutes. Here, $\tilde \gamma = \gamma \circ s$ with $t=s(\tau)$ where $\tau\in I'\subset \mathbb{R}$ is the parameter of $\tilde \gamma$ and we set the function $s$ so that $p=\gamma(t_0)=\tilde\gamma(\tau_0)$, that is, $t_0 = s(\tau_0)$. Note that although the image of the curves is the same -- $\gamma(I) = \tilde\gamma(I')\subset U$ -- they are indeed different curves. %
Indeed
	\beq
	\label{eq.tildeu}
	\tilde v\vert_p[f] = \left.\frac{\d}{\d \tau} \left[f \circ \tilde\gamma \right](\tau) \right\vert_{\tau=\tau_0} 
					 	=	\left.\frac{\d }{\d \tau}\left[f \circ \gamma \circ s \right](\tau) \right\vert_{\tau=\tau_0} 
					 	= \left.\frac{\d}{\d \tau} s(\tau) \right\vert_{\tau=\tau_0}\left.\frac{\d}{\d t}\left[f \circ \gamma \right](t) \right\vert_{t=t_0},
	\eeq
from which it follows that $\tilde v\vert_p[f] = s'(\tau_0) v\vert_p[f]$.
Clearly, $\tilde v\vert_p \in T_p M$ since it represents the directional derivative along the curve $\tilde\gamma$ at the point $p\in U$ and, due to the arbitrariness of the reparametrisation, $s'(\tau_0) v\vert_p \in T_p M$ where the prime denotes the derivative of $s$ with respect to the parameter $\tau$. A reparametrisation is invertible, so $s'(\tau_0)$ ranges over the reals other than zero, and the remaining value is furnished by the constant curve $\gamma(t) \equiv p$, whose directional derivative annihilates every function and is the null element of $T_p M$. Hence $\lambda\, v\vert_p \in T_p M$ for every $\lambda \in \mathbb{R}$.

	\end{proof}

	\begin{corollary}
	\label{cor.dim}
	Let $M$ be an $n$-dimensional differentiable manifold and $p\in M$. The tangent space $T_p M$ has dimension $n$. In any chart $(U,\psi)$ around $p$, the coordinate directional derivatives $\left\{\left.\basis[x^i]\right\vert_p\right\}_{i=1}^n$ form a basis, and the assignment $v\vert_p \longmapsto \left(\left.\dt[t] x^1\right\vert_{t_0},\dots,\left.\dt[t] x^n\right\vert_{t_0}\right)$ is a linear isomorphism $T_p M \cong \mathbb{R}^n$.
	\end{corollary}
	\begin{proof}
	By the coordinate expression \eqref{eq.vellocal}, every $v\vert_p \in T_p M$ acts on a function as
	\beq
	v\vert_p[f] = \sum_{i=1}^n \left.\dt[t] x^i\right\vert_{t_0} \left.\basis[x^i]\left(f\circ\psi^{-1}\right)\right\vert_{\psi(p)},
	\eeq
	so the operators $\basis[x^i]\vert_p$ span $T_p M$. They are also linearly independent, for the coordinate functions serve as their own probes. Since $\basis[x^i]\vert_p\left[x^j\right] = \delta^j_i$, a vanishing combination applied to $x^j$ returns its own coefficient,
	\beq
	\sum_{i=1}^n \left. c^i \basis[x^i]\right\vert_p = 0 \quad \Longrightarrow \quad 0 = \left(\sum_{i=1}^n \left.c^i \basis[x^i]\right\vert_p\right)\left[x^j\right] = \sum_{i=1}^n c^i \delta^j_i = c^j, \qquad j=1,\dots,n.
	\eeq
	The stated assignment is therefore a linear bijection.
	\end{proof}

Two remarks are in order. On the one hand, from the geometric point of view, we see that the directional derivative operator $v\vert_p\in T_p M$ along a curve $\gamma$ at a given point $p\in U\subset M$ is a vector which can be considered as \emph{tangent} to the curve in the same way as $\vec v_{\psi(p)}$ is tangent to $\gamma^U$ at $\psi(p)\in \psi(U)\subset \mathbb{R}^n$ in the corresponding image in the local chart $(U,\psi)$. On the other hand, from the kinematic point of view, the intuitive notion of velocity is linked to the idea of how fast the position of a particle is changing. However, we do not yet have an inner product or a norm defined on $T_p M$ that allows us to compare the magnitude of tangent vectors. Nonetheless, in proving Proposition~\ref{prop.vecspace} -- cf. equation~\eqref{eq.tildeu} -- we see that a reparametrisation of a curve is morally tantamount to traversing the same \emph{path} at a different \emph{speed}. Here, the path is merely the image of the curve. Although we have not yet properly introduced the notion of speed, we can unambiguously tell whether a function changes \emph{faster} or \emph{slower} at a given point while travelling along the given path, depending on the parametrisation we use.

\subsection{Acceleration}

Now, our inquiry can go a step further to consider the rate of change of the velocity of a particle traversing a curve at a given point. This turns out to be a harder problem since we will have to compare vectors at neighbouring points. Indeed, if $p,q \in M$ are different, $q \neq p$, then $T_p M$ and $T_q M$ are two different vector spaces and we have no means of adding a tangent vector at $p$ with one at $q$. In order to proceed we need to prescribe a procedure to \emph{move} vectors from one point in the manifold to another.

To begin our new task, we set up a procedure to relate tangent vectors not only at different points but in a more general setting, namely, that between different manifolds.

	\begin{mydef}[The push-forward]
	\label{def.push}
	Let $M$ and $N$ be differentiable manifolds, $U\subset M$ an open subset and $\gamma:I\subset\mathbb{R}\longrightarrow U\subset M$ a smooth curve such that $\gamma(t_0) = p\in U$. Let $F:M \longrightarrow N$ be a smooth map and $f:F(U)\subset N \longrightarrow \mathbb{R}$ a real valued differentiable function. The induced linear map $\D_p (F): T_p M \longrightarrow T_{F(p)} N$ whose action on a tangent vector $v\vert_p \in T_p M$,
	\beq
	\D_{p}(F) [v] \equiv F_*\left[v\vert_p\right] = u\vert_{F(p)} \in T_{F(p)} N,
	\eeq
is defined by 
	\beq
	\label{eq.pushf}
	u\vert_{F(p)}\left[f \right] = \left.\frac{\d}{\d t}\left[f \circ (F \circ \gamma) \right]\right\vert_{t=t_0} =\left.\frac{\d}{\d t}\left[(f \circ F) \circ \gamma \right]\right\vert_{t=t_0} = v\vert_p [f \circ F],
	\eeq
 and is called the push-forward of $v$ by $F$ at $p$.
\end{mydef}

The schematic picture of Definition~\ref{def.push} can be summarised by the diagram
		\beq
		\label{diag.compmaps2}
	\begin{tikzpicture}[]
	\matrix[matrix of math nodes,column sep={100pt,between origins},row
	sep={70pt,between origins},nodes={asymmetrical rectangle}] (s)
	{
	|[name=a1]| I \subset \mathbb{R}					& |[name=a2]| 	U\subset M		& |[name=a3]| 				\\
	|[name=b1]|											& |[name=b2]| F(U) \subset N	& |[name=b3]| \mathbb{R}.	\\
	}
	;
	\draw[->] 	
			(a1) edge node[auto] {\(\gamma\)} (a2)
			(a2) edge node[auto] {\(F\)} (b2)
			(a1) edge node[left] {\(F \circ \gamma \)} (b2)
			(b2) edge node[auto] {\(f \)} (b3)
			(a2) edge node[auto] {\(f \circ F \)} (b3)
	;
	\end{tikzpicture}
	\eeq
 Thus we see that the action that $u\vert_{F(p)}\in T_{F(p)} N$ has on $f$ is equivalent to that of $v\vert_p\in T_p M$ on $(f\circ F):U\subset M \longrightarrow \mathbb{R}$.

	In this sense, the vector $u\vert_{F(p)}$ is tangent to the induced curve $(F\circ \gamma)$ at the point $F(p)$ and, as a differential operator, its action on any function $f:F(U)\subset N \longrightarrow \mathbb{R}$ is \emph{pushed forward} from that of $v\vert_p$ on $(f\circ F):U\subset M \longrightarrow \mathbb{R}$. 
	
	Now, let us express the push-forward~\eqref{eq.pushf} in terms of local coordinates. To this end, consider the diagram
	\beq
		\label{diag.compmaps}
	\begin{tikzpicture}[]
	\matrix[matrix of math nodes,column sep={100pt,between origins},row
	sep={70pt,between origins},nodes={asymmetrical rectangle}] (s)
	{
	|[name=a0]|						& |[name=a1]| U \subset M						& |[name=a2]| 		F(U) \subset N							& |[name=a3]|					\\
	|[name=b0]|	I\subset\mathbb{R}	& |[name=b1]| \psi(U) \subset\mathbb{R}^m 		& |[name=b2]| \phi\left[F(U) \right] \subset \mathbb{R}^n	& |[name=b3]| \mathbb{R}.		\\
	}
	;
	\draw[->] 	
			(a1) edge node[auto] {\(F\)} (a2)
			(a2) edge node[auto] {\(\phi\)} (b2)
			(a1) edge node[left] {\(\psi \)} (b1)
			(b1) edge node[auto] {\(\tilde F \)} (b2)
			(b0) edge node[auto] {\(\gamma\)} (a1)
			(b0) edge node[auto] {\(\gamma^U\)} (b1)
			(a2) edge node[auto] {\(f\)} (b3)
			(b2) edge node[auto] {\(f_{F(U)} \)} (b3)
			(b0) edge[bend right=20] node[below] {\(f_{F(U)} \circ \tilde F \circ \gamma^U \)} (b3)
	;
	\end{tikzpicture}
	\eeq
Here, as before [cf. equation~\eqref{eq.equivmaps}] we have
	\beq
	f_{F(U)} \circ \tilde F \circ \gamma^U = \left(f \circ \phi^{-1}\right) \circ \left(\phi \circ F \circ \psi^{-1} \right) \circ \left(\psi \circ \gamma\right) = f \circ F \circ \gamma.
	\eeq
In this case, we can compute
	\begin{align}
	& \left.\frac{\d}{\d t} \left[f_{F(U)} \circ \tilde F \circ \gamma^U \right]\right\vert_{t=t_0}\nonumber\\
	& \quad = \left.\frac{\d}{\d t} f_{F(U)}\left[y^1(t),\cdots,y^n(t) \right]\right\vert_{t=t_0}\nonumber\\
				& \quad = \sum_{i=1}^n \left.\frac{\partial}{\partial y^i} \left(f\circ \phi^{-1}\right)\right\vert_{\phi\left[F(p)\right]} \left.\frac{\d}{\d t} y^i(t)\right\vert_{t=t_0}\nonumber\\
				& \quad = \sum_{i=1}^n \left.\frac{\partial}{\partial y^i} \left(f\circ \phi^{-1}\right)\right\vert_{\phi\left[F(p)\right]} \left[\sum_{j=1}^m \left.\left(\frac{\partial \tilde F^i}{\partial x^j}\right)\right\vert_{(\phi\circ F\circ\psi^{-1}) [x^k_p]} \left.\frac{\d}{\d t}x^j(t) \right\vert_{t=t_0}\right]\nonumber\\
				& \quad = \sum_{i=1}^n \sum_{j=1}^m \left.\left(\frac{\partial \tilde F^i}{\partial x^j}\right)\right\vert_{(\phi\circ F\circ\psi^{-1}) [x^k_p]} \left.\frac{\d}{\d t}x^j(t) \right\vert_{t=t_0} \left.\frac{\partial}{\partial y^i} \left(f\circ \phi^{-1}\right)\right\vert_{\phi\left[F(p)\right]},
	\end{align}
	and it follows that the components of the tangent vector to the induced curve $(F\circ \gamma)$ at $F(p)$ are
	\beq
	 \left.\frac{\d}{\d t} y^i(t)\right\vert_{t=t_0}= \sum_{j=1}^m \left.\left(\frac{\partial \tilde F^i}{\partial x^j}\right)\right\vert_{(\phi\circ F\circ\psi^{-1}) [x^k_p]} \left.\frac{\d}{\d t}x^j(t) \right\vert_{t=t_0}.
	\eeq

 Moreover, let $(U,\psi)$ and $\left(F(U),\phi\right)$ be local coordinate charts for $M$ and $N$ around $p$ and $F(p)$, respectively, so that
	\beq
	v\vert_p[f\circ F] = \sum_{j=1}^m \left.\frac{\d}{\d t} x^j(t) \right\vert_{t=t_0} \left.\frac{\partial}{\partial x^j}\left[(f \circ F) \circ \psi^{-1}\right]\right\vert_{\psi(p)},
	\eeq
	and
	\beq
	u\vert_{F(p)}\left[f \right] = \sum_{i=1}^n \left.\frac{\d}{\d t} y^i(t) \right\vert_{t=t_0} \left.\frac{\partial}{\partial y^i}\left(f \circ \phi^{-1}\right)\right\vert_{\phi[F(p)]}.
	\eeq
 Then
	\beq
	\begin{split}
	\D_p(F)[v][f] & = \sum_{i=1}^n \sum_{j=1}^m \left.\left(\frac{\partial \tilde F^i}{\partial x^j}\right)\right\vert_{(\phi\circ F\circ\psi^{-1}) [x^k_p]} \left.\frac{\d}{\d t}x^j(t) \right\vert_{t=t_0} \left.\frac{\partial}{\partial y^i} \left(f\circ \phi^{-1}\right)\right\vert_{\phi\left[F(p)\right]}\\
	& = u\vert_{F(p)}[f].
	\end{split}
	\eeq

Up to this point every tangent vector has been attached to a single event. The constructions that follow require them assembled along whole curves, and the object that carries them deserves a name. The \emph{tangent bundle} of $M$ is the union $TM \equiv \bigcup_{p\in M} T_p M$ of all its tangent spaces, and a \emph{vector field} is a smooth assignment of a tangent vector $v\vert_p \in T_p M$ to each event $p\in M$, that is, a smooth section of $TM$. We write $\GTM$ for the set of such sections.

	\begin{mydef}[Flow of a vector field]
	\label{def.flow}
	Let $v\in \GTM$, $\gamma:I\subset \mathbb{R}\longrightarrow M$ an integral curve of $v$ passing through $p=\gamma(t_0) \in M$ for some $t_0 \in I$. The flow of the vector field $v$ through the point $p$ is the one-parameter family of mappings $\varphi:(-\epsilon,\epsilon)\times M \longrightarrow M$,
		\beq
		(t,p)\mapsto \varphi_t(p) = \gamma(t_0 + t) \in M, 
		\eeq
	such that
		\begin{enumerate}
			\item $\varphi_{0}(p) = \gamma(t_0) =p$,
			\item $\varphi_{t+s}(p) = \varphi_t\left[\varphi_s (p)\right]$ for every $t,s \in (-\epsilon,\epsilon)$, 
			\item $\varphi^{-1}_t \left[\varphi_t (p) \right] = \varphi_{(-t)}\left[\gamma(t_0 + t)\right]= \gamma(t_0)= p$.
		\end{enumerate}
	\end{mydef}
The existence of the flow is itself an assumption, and a mild one. The smoothness of the vector field $v$ guarantees, through the standard theory of ordinary differential equations, that a unique integral curve issues from each event. With it the flow $\varphi_t$ is defined for $t$ in some interval $(-\epsilon,\epsilon)$ about the origin, in general a proper one. Whether $\epsilon$ may be taken arbitrarily large, so that $\varphi_t$ is defined for every $t\in\mathbb{R}$, is the question of the \emph{completeness} of $v$; we shall assume it whenever we follow a motion for an unbounded lapse of its parameter, and shall signal the assumption at the point of use.

For each $t\in(-\epsilon,\epsilon)$, $\varphi$ induces a map between tangent spaces ${\varphi_t}_* : T_{\gamma(t_0)}M \longrightarrow T_{\gamma(t_0 + t)}M$ such that for every differentiable function $f: M\longrightarrow \mathbb{R}$
we have
	\beq
	\label{diag.push}
	\begin{tikzpicture}[]
	\matrix[matrix of math nodes,column sep={100pt,between origins},row
	sep={70pt,between origins},nodes={asymmetrical rectangle}] (s)
	{
	|[name=a1]| 	M				& |[name=a2]| M 		\\
	|[name=b1]| 				 	& |[name=b2]| \mathbb{R}				 \\
	}
	;
	\draw[->] 	
			(a1) edge node[above] {\(\varphi_t\)} (a2)
			(a2) edge node[right] {\(f\)} (b2)
			(a1) edge node[left] {\(\tilde f = f \circ \varphi_t\)} (b2)
	;
	\end{tikzpicture}
	\eeq
and
	\beq
	\label{eq.flowpush}
	 v\vert_p \left[f \circ \varphi_t \right] = {\varphi_t}_* \left(v\vert_p \right)\left[ f\right] .
	\eeq

This allows us to define the derivative of a vector field $w$ along the flow generated by $v$ as
	\beq
	\label{eq.lieflow}
	\dot w\vert_p \equiv \lim_{t\rightarrow 0} \frac{{\varphi_{t}}^{-1}_*\left(w\vert_{\varphi_t(p)}\right) - w\vert_p}{t} = \left.\dt[t] {\varphi_t}^{-1}_*\left(w\vert_{\varphi_t(p)}\right) \right\vert_{t=0} \equiv L_v w\vert_p \in \tpm.
	\eeq
Here, the map ${\varphi_t}^{-1}_* = {\varphi_{-t}}_*$ carries the vector $w\vert_{\varphi_t(p)}$ back from the advanced event $\varphi_t(p)$ to $p$, so that the difference in the numerator is an element of the single tangent space $T_p M$ and the limit is well defined. A function needs no transport back, since the value $f\left(\varphi_t(p)\right)$ is already a real number, and the flow derivative reduces to the rate of change of $f$ along the motion, that is, $L_v f\vert_p \equiv v\vert_p[f]$ in the sense of Definition~\ref{def.velocity}.

The abstract definition \eqref{eq.lieflow} acquires a concrete and computable form once we let it act on a function. %

	\begin{lemma}[Lie bracket]
	\label{lem.liebracket}
	Let $v,u\in \GTM$ be a pair of vector fields and $\varphi$ the flow generated by $v$. The flow derivative \eqref{eq.lieflow} of $u$ along $v$ is the vector field whose action on any smooth function $f:M\longrightarrow \mathbb{R}$ is
		\beq
		\label{eq.Xdot}
		L_v u\,[f] = v\left[u[f]\right] - u\left[v[f]\right].
		\eeq
	This field is the \emph{Lie bracket} of $v$ and $u$, written $L_v u = [v,u]$~\cite{nakahara2018geometry}.
	\end{lemma}
	\begin{proof}
	The map ${\varphi_t}^{-1}_* = {\varphi_{-t}}_*$ transports $u\vert_{\varphi_t(p)}$ back to $p$ along the diffeomorphism $\varphi_{-t}$, which sends $\varphi_t(p)$ to $p$; by the push-forward definition \eqref{eq.pushf}, its action on $f$ is
		\beq
		\label{eq.liebproof1}
		{\varphi_t}^{-1}_*\left(u\vert_{\varphi_t(p)}\right)[f] = u\vert_{\varphi_t(p)}\left[f\circ\varphi_{-t}\right].
		\eeq
	Expanding the right hand side to first order in $t$ -- using $\left.\dt[t]\left(f\circ\varphi_t\right)\right\vert_{t=0} = v[f]$ for the velocity of the flow [cf. definition \eqref{eq.velocity}] -- gives
		\beq
		\label{eq.liebproof2}
		f\circ\varphi_{-t} = f - t\, v[f] + O(t^2).
		\eeq
	Substituting the expansion~\eqref{eq.liebproof2} into the pulled-back action~\eqref{eq.liebproof1} and applying the companion expansion $h\left(\varphi_t(p)\right) = h(p) + t\, v[h](p) + O(t^2)$ of any smooth function $h$ at the advanced event -- to $h = u[f]$ in the leading term and, within the term already of order $t$, to $h = u\left[v[f]\right]$ -- yields
		\begin{align}
		\label{eq.liebproof3}
		{\varphi_t}^{-1}_*\left(u\vert_{\varphi_t(p)}\right)[f] 	& = u\vert_{\varphi_t(p)}[f] - t\, u\vert_{\varphi_t(p)}\left[v[f]\right] + O(t^2)\nn\\
			& = u[f](p) + t\left(v\left[u[f]\right] - u\left[v[f]\right]\right)(p) + O(t^2).
		\end{align}
	Subtracting $u\vert_p[f] = u[f](p)$, dividing by $t$ and letting $t\rightarrow 0$ in the definition \eqref{eq.lieflow} yields the expression \eqref{eq.Xdot}. The right hand side of \eqref{eq.Xdot} is again a derivation -- hence a genuine vector field -- since expanding $[v,u][fh]$ by the Leibniz rule cancels the terms in products of first derivatives of $f$ and $h$, leaving $f\,[v,u][h] + h\,[v,u][f]$; the antisymmetry $[v,u] = -[u,v]$ and the Jacobi identity are then immediate, standard properties for which we refer the reader to~\cite{nakahara2018geometry}.
	\end{proof}
In this case -- intuitively -- the acceleration of the curve should correspond to the time derivative of its velocity. However, in view of Lemma~\ref{lem.liebracket}, this derivative is the bracket of the velocity field with itself, $L_v v = [v,v] = 0$, so the condition it would impose is trivial and identically satisfied.

Recalling Newtonian dynamics, the acceleration of a trajectory measures the departure of its tangent velocity vector field from an \emph{a priori} notion of straightness when \emph{propagated} along its own integral curves. That is, the manifold $M$ ought to be equipped with an attribute allowing us to uniquely propagate vectors from one point in the manifold to another. This is precisely the notion of a \emph{propagator}. Thus, let us consider the following
	\begin{mydef}[Parallel propagator]
	\label{def.propagator}
	Let $v,w \in \GTM$ be a pair of vector fields, $\gamma:I\subset\mathbb{R}\longrightarrow M$ the integral curve of $v$ passing through $p=\gamma(t_0)$, $t_0\in I$. A propagator along $\gamma$ is a one-parameter family of bijective maps between tangent spaces
	\beq
	P_\gamma\left(t_0,t\right): T_{\gamma(t_0)}M \longrightarrow T_{\gamma(t)}M
	\eeq 
such that
	\begin{enumerate}
	\item $\propagator[\gamma,t_0,t_0]\left[w\vert_{\gamma(t_0)} \right]= w\vert_{\gamma(t_0)}$,
	\item $\propagator[\gamma,t_1,t_2]\left[ \propagator[\gamma,t_0,t_1]\left(w\vert_{\gamma(t_0)} \right)\right] = \propagator[\gamma,t_0,t_2]\left(w\vert_{\gamma(t_0)} \right)$,
	\item $\ipropagator[\gamma,t,t_0]\left[ \propagator[\gamma,t_0,t]\left(w\vert_{\gamma(t_0)} \right)\right] = w\vert_{\gamma(t_0)}$ and, finally,
	\item when acting on functions
	\beq
	\propagator[\gamma,t_0,t]\left[f(p)\right] = f(p) \quad \text{where} \quad p=\gamma(t_0).
	\eeq	
	\end{enumerate}
A propagator should not distort the linear structure of the tangent spaces it connects, that is, transporting a sum amounts to transporting the summands. Thus, we say $P_\gamma\left(t_0,t\right)$ is \emph{linear} if each of its maps is a linear isomorphism,
	\beq
	\label{eq.linprop}
	P_\gamma\left(t_0,t \right)\left[\alpha\, w\vert_{\gamma(t_0)} + \beta\, u\vert_{\gamma(t_0)} \right] = \alpha\, P_\gamma\left(t_0,t \right) \left[w\vert_{\gamma(t_0)} \right] + \beta\, P_\gamma\left(t_0,t \right) \left[u\vert_{\gamma(t_0)} \right],
	\eeq
for every $\alpha,\beta\in \Reals$ and every pair $w\vert_{\gamma(t_0)}, u\vert_{\gamma(t_0)} \in T_{\gamma(t_0)}M$. Moreover, transport should be an attribute of the \emph{path} -- the oriented image of $\gamma$ -- and not of the clock used to traverse it. We say the propagator is \emph{parametrisation independent} if, for every reparametrisation $\sigma$ with $\sigma' \neq 0$,
	\beq
	\label{eq.reparinv}
	P_{\gamma\circ\sigma^{-1}}\left(\sigma(t_0),\sigma(t) \right) = P_\gamma\left(t_0,t\right).
	\eeq
A linear and parametrisation-independent propagator is called \emph{parallel}.
\end{mydef}

Given a parallel propagator, we can introduce a derivative operator acting on vector fields as follows
	\begin{mydef}[Covariant derivative]
	\label{def.covd}
	Let $v,w\in \GTM$ be a pair of vector fields, $P_{\gamma}\left(t_0,t \right)$ a parallel propagator along the integral curve $\gamma:I\subset\mathbb{R} \longrightarrow M$ of $v$ passing through $p=\gamma(t_0)$, $t_0\in I$. The covariant derivative of $w$ along the flow of $v$ at the point $p\in M$  is defined as
		\beq
		 \left.\frac{\D}{\d t} w\right\vert_{t=t_0} \equiv \lim_{t\rightarrow t_0} \frac{P^{-1}_\gamma \left(t,t_0 \right)\left[w\vert_{\gamma(t)} \right] - w\vert_{\gamma(t_0)}}{t-t_0} = \left.\frac{\d}{\d t} P^{-1}_\gamma\left(t,t_0 \right)\left[w\vert_{\gamma(t)} \right]\right\vert_{t=t_0} \equiv \nabla_v w\vert_p.
		\eeq
	\end{mydef}
Here, we have also introduced the notation $\nabla_v w\vert_p$ so that the reader can relate the content of this section to the standard literature. The Leibniz notation, however, serves us to identify the infinitesimal nature of the covariant derivative.

The difference in the numerator of Definition~\ref{def.covd} is taken between vectors belonging to a single tangent space, and it is the propagator that puts them there, as the diagram
	\beq
	\label{diag.covd}
	\begin{tikzpicture}[]
	\matrix[matrix of math nodes,column sep={110pt,between origins},row
	sep={70pt,between origins},nodes={asymmetrical rectangle}] (s)
	{
	|[name=a1]| I\subset\mathbb{R}	& |[name=a2]| T_{\gamma(t)}M 	\\
	|[name=b1]| 				 	& |[name=b2]| T_p M	 			\\
	}
	;
	\draw[->]
			(a1) edge node[left] {\(\ipropagator[\gamma,t,t_0]\circ\, w \circ \gamma\)} (b2)
			(a2) edge node[auto] {\(\ipropagator[\gamma,t,t_0]\)} (b2)
			(a1) edge node[auto] {\(w \circ \gamma\)} (a2)
	;
	\end{tikzpicture}
	\eeq
records. Its upper right node is not one vector space but a different one for each value of $t$, and that is precisely why the vectors $w\vert_{\gamma(t)}$ and $w\vert_{\gamma(t_0)}$ cannot be subtracted as they stand. The diagonal composite lands, for every $t$, in the single space $T_p M$ fixed by the event at which we take the derivative, and there the quotient of Definition~\ref{def.covd} is an ordinary limit of real vectors in a fixed vector space, of the kind we already performed upon functions in the definition \eqref{eq.velocity} of the velocity. The covariant derivative is the derivative of the curve traced along the diagonal, and the propagator is what makes that curve exist.

The covariant derivative satisfies a number of properties which follow from the definition of the parallel propagator. These are usually presented as the defining properties of an affine connection, in terms of which the parallel transport is itself defined; here the order is reversed. We have adopted a posture in which the notion of transport -- not necessarily parallel -- is the more primitive of the two, and an affine connection will emerge as the object satisfying the properties we derive. In formulating Newtonian kinematics, then, we take as given an \emph{a priori} notion of parallel transport, and explore the basic properties of the covariant derivative it induces.
	
	\begin{lemma}
	\label{lem.covfunc}
	Let $v\in \GTM$ be a vector field on $M$ such that $v\vert_p$ is the tangent vector to a curve $\gamma:I\subset\mathbb{R} \longrightarrow M$ at $p=\gamma(t_0)\in M$, $f:M\longrightarrow \mathbb{R}$ a differentiable function. The covariant derivative of $f$ along the flow of $v$ at the point $p\in M$ satisfies
	\beq
	\label{eq.funcD}
	\left.\frac{\D}{\d t} f\,\right\vert_{t=t_0}  = \left.\frac{\d}{\d t} [f \circ \gamma] \right\vert_{t=t_0},
	\eeq
	or, equivalently,
	\beq
	\label{eq.func}
	\nabla_v f\vert_p = v[f]\vert_p.
	\eeq
	\end{lemma}
	\begin{proof}
	From the definition of the action of the parallel propagator on functions we have
		\begin{align}
		\left.\frac{\D}{\d t} f\,\right\vert_{t=t_0}	& = \lim_{t\rightarrow t_0}\frac{\ipropagator[\gamma,t,t_0][f\left(\gamma(t)\right)] - f\left(\gamma(t_0)\right)}{t - t_0}\nonumber\\
													& = \lim_{t\rightarrow t_0}\frac{f\left(\gamma(t)\right) - f\left(\gamma(t_0)\right)}{t- t_0}\nonumber\\
													& = \left.\frac{\d}{\d t} [f \circ \gamma] \right\vert_{t=t_0} = v[f]\vert_p,
		\end{align}
	the last equality by the very definition \eqref{eq.velocity} of the particle velocity.
	\end{proof}
On functions, then, the covariant derivative retains nothing of the propagator and coincides with the flow derivative, $\nabla_v f\vert_p = L_v f\vert_p$.
	
	\begin{lemma}[Leibniz rule]
	\label{lem.leibniz}
	Let $v,w\in \GTM$ be a pair of vector fields over $M$, $f:M \longrightarrow \mathbb{R}$ a differentiable function, $\gamma:I\subset\mathbb{R} \longrightarrow M$ the integral curve of $v$ passing through $p=\gamma(t_0)$, $t_0\in I$. Then
		\beq
		\label{eq.leibnizD}
		\left.\frac{\D}{\d t} (f w) \, \right\vert_{t=t_0} = f(p) \left.\frac{\D}{\d t} w \,\right\vert_{t=t_0} + \left.\frac{\d}{\d t} [f \circ \gamma]\right\vert_{t=t_0} w\vert_p,
		\eeq
	or, equivalently,
		\beq
		\label{eq.leibniz}
		\nabla_v (f w)\vert_{p} = f(p) \nabla_v w \vert_p + v[f] w\vert_p.
		\eeq
	\end{lemma}
	\begin{proof}
	Straightforward calculation -- with $q\equiv\gamma(t)$ and using in the second equality the linearity \eqref{eq.linprop} of the parallel propagator,
		\begin{align}
		 \left.\frac{\D}{\d t} \left(f w \right)\right\vert_{t=t_0} 	& = \lim_{t\rightarrow t_0} \frac{\ipropagator[\gamma,t,t_0]\left[f(q) w\vert_{\gamma(t)} \right] - f(p) w\vert_{\gamma(t_0)}}{t-t_0}\nonumber\\
		 									& = \lim_{t\rightarrow t_0} \frac{f(q)\ipropagator[\gamma,t,t_0]\left[w\vert_{\gamma(t)} \right] - f(p) w\vert_{\gamma(t_0)}}{t-t_0}\nonumber\\
											& = \lim_{t\rightarrow t_0} \frac{f(q)\ipropagator[\gamma,t,t_0]\left[w\vert_{\gamma(t)} \right] - f(p) w\vert_{\gamma(t_0)} + f(q) w\vert_{\gamma(t_0)} - f(q) w\vert_{\gamma(t_0)} }{t-t_0}\nonumber\\
											& = \lim_{t\rightarrow t_0} \frac{f(q)\ipropagator[\gamma,t,t_0]\left[w\vert_{\gamma(t)} \right] -f(q) w\vert_{\gamma(t_0)}}{t-t_0} +\lim_{t\rightarrow t_0} \frac{f(q) w\vert_{\gamma(t_0)}-f(p) w\vert_{\gamma(t_0)}}{t-t_0}\nonumber\\
											& = f\left[\gamma(t_0) \right] \left.\frac{\D}{\d t} w\right\vert_{t=t_0} + \left.\frac{\d}{\d t} f\left[\gamma(t)\right]\right\vert_{t=t_0} w\vert_{\gamma(t_0)}.
		\end{align}
	Hence
	\beq
	\nabla_v \left(f w \right)\vert_p =\left. \frac{\D}{\d t} \left(f w \right)\right\vert_{t=t_0} = f(p) \nabla_v w\vert_p + v[f]\vert_{p} w\vert_p.
	\eeq
	\end{proof}

That settles the rescaling of the transported argument. The rescaling of the direction is a separate question, and it is answered by first asking what such a rescaling does to the curve along which the transport is carried out.

	\begin{prop}
	\label{prop.reparam}
	Let $v \in \GTM$ be a vector field on $M$, $f:M\longrightarrow \mathbb{R}$ a differentiable function, $\gamma:I\subset\mathbb{R}\longrightarrow M$ the integral curve of $v$ passing through $p=\gamma(t_0)$. Assume $f(p)\neq 0$. Let $\tilde v = f v$ be a new vector field in a neighbourhood around $p$. Then, there is a reparametrisation
	\beq
	\sigma: I \subset\mathbb{R} \longrightarrow \mathbb{R}, \qquad \sigma(t_0) = s_0, \qquad \left.\frac{\d}{\d t} \sigma\right\vert_{t=t_0}\neq 0,
	\eeq
	 such that the curve $\tilde\gamma = \gamma \circ \sigma^{-1}$ is the integral curve of $\tilde v$ passing through $p=\tilde\gamma(s_0)$. 
	\end{prop}
	\begin{proof}
	Consider the diagram
	\beq
	\label{diag.repar}
	\begin{tikzpicture}[]
	\matrix[matrix of math nodes,column sep={100pt,between origins},row
	sep={70pt,between origins},nodes={asymmetrical rectangle}] (s)
	{
	|[name=a1]| 					& |[name=a2]| M 		\\
	|[name=b1]| I_0\subset\mathbb{R} 	& |[name=b2]| I \subset \mathbb{R}.	 \\
	}
	;
	\draw[->] 	
			(b2) edge node[below] {\(\sigma\)} (b1)
			(b2) edge node[right] {\(\gamma\)} (a2)
			(b1) edge node[auto] {\(\tilde\gamma=\gamma\circ\sigma^{-1}\)} (a2)
	;
	\end{tikzpicture}
	\eeq
and note first that, $f$ being continuous and $f(p)\neq 0$, $f$ vanishes nowhere on a neighbourhood of $p$. On that neighbourhood define
		\beq
		\label{eq.sigmadef}
		\sigma(t) = \int_{t_0}^t \frac{\d \tau}{f\left[\gamma(\tau)\right]} + s_0,
		\eeq
	which satisfies $\sigma(t_0)=s_0$ and
		\beq
		\label{eq.reparam}
		f\left[\gamma(t) \right] = \left[\frac{\d }{\d t}\sigma \right]^{-1}_{t}, \qquad \text{in particular} \qquad f(p) = \left[\frac{\d }{\d t}\sigma \right]^{-1}_{t=t_0}.
		\eeq
	Its derivative therefore vanishes nowhere, so $\sigma$ is invertible on a neighbourhood $I_0$ of $s_0$ and the curve $\tilde\gamma = \gamma\circ\sigma^{-1}$ is defined there, with $p=\tilde\gamma(s_0)$. Its tangent vector is
		\beq
		\left.\frac{\d}{\d s} \tilde\gamma\right\vert_{s} = \left.\frac{\d}{\d t}\gamma \right\vert_{t} \left.\frac{\d }{\d s}\sigma^{-1} \right\vert_{s}=\left.\frac{\d}{\d t}\gamma \right\vert_{t} \left[\frac{\d }{\d t} \sigma \right]^{-1}_{t} = f\left[\gamma(t)\right] v\vert_{\gamma(t)},
		\eeq
	by the chain rule and the identification \eqref{eq.reparam}. Thus $\tilde\gamma$ is the integral curve of $\tilde v = f v$ passing through $p=\tilde\gamma(s_0)$, which is the reparametrisation sought.
	\end{proof}

A rescaling of the velocity is therefore no more than a change of clock along the same path, and it is because the propagator was required to be indifferent to that clock that the covariant derivative inherits the homogeneity we now record.

	\begin{lemma}[Homogeneity of the covariant derivative]
	\label{lem.chain}
	Let $v,w\in \GTM$ be a pair of vector fields, $f: M\longrightarrow \mathbb{R}$ a differentiable function, $\gamma:I\subset\mathbb{R}\longrightarrow M$ the integral curve of $v$ passing through $p=\gamma(t_0)$. Assume $f(p) \neq 0$. Then
		\beq
		\label{eq.chain}
		\nabla_{(fv)}w\vert_p = f(p) \nabla_v w \vert_p.
		\eeq
	\end{lemma}
	\begin{proof}
	\begin{align}
	\nabla_{(fv)}w\vert_p = \nabla_{\tilde v} w\vert_p 	& = \left.\frac{\d}{\d s} \ipropagator[\tilde\gamma,s,s_0]\left[w\vert_{\tilde\gamma(s)} \right]\right\vert_{s=s_0}\nonumber\\
											& = \left.\frac{\d}{\d t} \ipropagator[\gamma,t,t_0]\left[w\vert_{\gamma(t)} \right]\right\vert_{t=t_0} \left[\frac{\d}{\d t}\sigma\right]^{-1}_{t=t_0}\nonumber\\
											&=f\left(p \right) \nabla_v w\vert_p,
	\end{align}
where the second equality uses the parametrisation independence \eqref{eq.reparinv} of the parallel propagator -- the transports along $\tilde\gamma$ and $\gamma$ coincide -- together with the chain rule, and the last equality follows from the identification \eqref{eq.reparam} of Proposition~\ref{prop.reparam}.
	\end{proof}

Nothing in Definition~\ref{def.propagator} relates the transports along \emph{different} curves through the same event. Each curve carries its own propagator. Two further attributes of the system of propagators must therefore be declared. On the one hand, we assume, throughout, that the system is \emph{local}, that is, the derivative $\nabla_v w\vert_p$ depends on the vector field $v$ only through its value $v\vert_p$, so that transports along curves sharing their velocity at $p$ agree to first order. On the other hand, we also assume it \emph{smooth}, that is, the assignment $v\vert_p \longmapsto \nabla_v w\vert_p$ is smooth for every $w \in \GTM$.

	\begin{lemma}[Additivity of the covariant derivative]
	\label{lem.additive}
	Let $v,u,w,z \in \GTM$ be vector fields. Then
		\beq
		\label{eq.addw}
		\nabla_v \left(w + z\right)\vert_p = \nabla_v w\vert_p + \nabla_v z\vert_p,
		\eeq
	and
		\beq
		\label{eq.addv}
		\nabla_{(v+u)}\, w\vert_p = \nabla_v w\vert_p + \nabla_u w\vert_p.
		\eeq
	\end{lemma}
	\begin{proof}
	The additivity \eqref{eq.addw} in the transported argument is immediate. By the linearity \eqref{eq.linprop} of the parallel propagator, the transport of the sum $(w+z)\vert_{\gamma(t)}$ is the sum of the transports, and the derivative of a sum is the sum of the derivatives.

	For the additivity \eqref{eq.addv} in the direction, fix $w$ and $p$ and consider the map $h: T_p M \longrightarrow T_p M$ given by $h\left(v\vert_p\right) \equiv \nabla_v w\vert_p$, well defined by locality. The homogeneity \eqref{eq.chain} with a constant function $f=\alpha$ gives $h\left(\alpha\, v\vert_p\right) = \alpha\, h\left(v\vert_p\right)$ for every $\alpha \neq 0$, and $h$ is continuous, so letting $\alpha$ tend to zero in that identity places the origin in its kernel, $h(0)=0$. The smoothness assumed above of the system of propagators makes $h$ differentiable there, and therefore
		\beq
		\label{eq.homlin}
		h\left(v\vert_p \right) = \lim_{\epsilon \rightarrow 0^{+}} \frac{h\left(\epsilon\, v\vert_p \right)}{\epsilon} = \left.\frac{\d}{\d \epsilon} h\left(\epsilon\, v\vert_p\right)\right\vert_{\epsilon=0^+},
		\eeq
	the first equality by that homogeneity and the second by the very definition of the one-sided derivative at an origin where $h$ vanishes. The right hand side of \eqref{eq.homlin} is the differential of $h$ at the origin evaluated along $v\vert_p$ -- necessarily a linear function of $v\vert_p$. Additivity follows.
	\end{proof}

Lemmas~\ref{lem.covfunc}, \ref{lem.leibniz} and \ref{lem.additive} exhaust the properties that the covariant derivative inherits from the parallel propagator, that is, it reduces to the ordinary derivative on functions, obeys a Leibniz rule, and is additive and homogeneous of degree one in each of its two arguments. These are, verbatim, the defining properties by which an affine connection is usually introduced. Here they arise as consequences of a more primitive object, the propagator. Therefore, we can now produce a suitable definition.

	\begin{mydef}[Affine connection]
	\label{def.connection}
	 Let $v,w \in \GTM$ be a pair of vector fields on $M$ and $f:M \longrightarrow \mathbb{R}$ a differentiable function. An affine connection is a map
		\beq
		\nabla:\GTM \times \GTM \longrightarrow \GTM,
		\eeq
	sending $(v,w) \longmapsto \nabla_v w$, additive in each of its arguments [cf. Lemma~\ref{lem.additive}], such that the function rule~\eqref{eq.func}, the Leibniz rule~\eqref{eq.leibniz} and the homogeneity~\eqref{eq.chain} are satisfied.
	\end{mydef}

An affine connection provides us with the coordinate expression for the parallel propagator associated with it. Consider a set of coordinates together with its corresponding coordinate basis. At every point on the manifold, the covariant derivative of a basis vector along the integral curve of another is an element of the tangent space at that point, and thus, it can be expressed as the linear combination
		\beq
	\left.\nabla_{\basis[x^i]} \basis[x^j]\right\vert_p = \frac{\d}{\d t}\left[\ipropagator[\psi(x^i),t,t_0]\left.\basis[x^j] \right\vert_{\psi(x^i(t))} \right]_{t=t_0} = \left.\sum_{k=1}^n \Gamma^k_{ij}(p)\basis[x^k]\right\vert_p.
	\eeq
Here, $\psi\left(x^i\right) = \psi\left(0,\dots ,x^i(t),\dots,0\right)$ is the image of a curve along the $\hat e_{(i)}$ direction [cf. diagram~\eqref{diag.curve}, above]; and the $n^3$ functions $\Gamma^k_{ij}:M\longrightarrow \mathbb{R}$ are called the \emph{connection coefficients} and they completely determine the transport at each point in the manifold. Thus, for a pair of vector fields $v,w \in \GTM$ we have	
	\begin{align}
	\nabla_v w\vert_p	& = \left.\nabla_{\left(v^i \basis[x^i]\right)}\left( w^j\basis[x^j]\right)\right\vert_p\nonumber\\
					& = \left. v^i \nabla_{\basis[x^i]}\left(w^j \basis[x^j] \right)\right\vert_p\nonumber\\
					& = \left. v^i \left[w^j \nabla_{\basis[x^i]}\basis[x^j] + \left(\nabla_{\basis[x^i]} w^j\right) \basis[x^j] \right]\right\vert_p\nonumber\\
					& = \left[ v^i w^j \Gamma^k_{ij}\basis[x^k]+\left(\nabla_{v} w^j \right)\basis[x^j] \right]_p\nonumber\\
					& = \left[\left.\dt[t] w^k\left[\gamma(t)\right]\right\vert_{t=t_0} + \Gamma^k_{ij}(p) v^i(p) w^j(p) \right]\left.\basis[x^k]\right\vert_p.
	\end{align}
Here, we have used the Einstein summation convention, the additivity of Lemma~\ref{lem.additive} in both arguments, and the homogeneity~\eqref{eq.chain} and the Leibniz rule~\eqref{eq.leibniz} in the second and third equalities, respectively. Such a coordinate expression enables us to evaluate the parallel transport locally in terms of the connection coefficients. Thus, let us make the following

	\begin{mydef}[Parallel transport]
	\label{def.ptransport}
	Let the pair $(M,\nabla)$ denote a differentiable manifold equipped with an affine connection. Consider a pair of vector fields $v,w \in \GTM$ and $\gamma:I\subset\mathbb{R}\longrightarrow M$ the integral curve of $v$ passing through $p=\gamma(t_0)$. We say that $w$ is transported parallel to $v$ at a point $p\in M$ if 
		\beq
		\label{eq.parallelt}
		\left.\frac{\D}{\d t} w \,\right\vert_{t=t_0}= \left.\frac{\d}{\d t}\ipropagator[\gamma,t,t_0]\left[w\vert_{\gamma(t)} \right]\right\vert_{t=t_0} =\nabla_v w\vert_p =0.
		\eeq 
	If condition~\eqref{eq.parallelt} is satisfied at every point along the curves of $v$ we say $w$ propagates parallel to $v$. 
	\end{mydef}

Since the covariant derivative is defined in terms of the parallel propagator, the uniqueness of the parallel transport we have just defined is a consequence of the bijective property of the map $\propagator[\gamma,t_0,t]$. This can be realised through the following
	\begin{prop}
	\label{prop.uniquetransport}
	Let the pair $(M,\nabla)$ denote a differentiable manifold equipped with an affine connection. Consider a pair of vector fields $v,w \in \GTM$ such that $w$ propagates parallel to $v$. The parallel transport at each point along the integral curve $\gamma:I\subset\mathbb{R}\longrightarrow M$ of $v$ passing through $p \in M$ at $t=t_0 \in I\subset \mathbb{R}$ is unique.
	\end{prop}
	\begin{proof}
	Consider the coordinate expression for the covariant derivative evaluated at an arbitrary point along the curve $\gamma$
		\beq
		\nabla_v w\vert_q = \left[\dt[t] w^k\left[\gamma(t)\right] + \Gamma^k_{ij}(q) v^i(q) w^j(q) \right]\left.\basis[x^k]\right\vert_q =0.
		\eeq
	This gives us $n$ first order ordinary differential equations for the components of the vector field $w$ evaluated at a point $q=\gamma(t)$, that is,
		\beq
		\label{eq.unique}
		\dt[t] w^k\left(q\right) = - \Gamma^k_{ij}(q) v^i(q) w^j(q).
		\eeq
	Together with the initial condition at $w\vert_{\gamma(t_0)}$, we obtain that the solution to \eqref{eq.unique} is unique.
	\end{proof}

With the transport now known to be unique, we may safely consider the change of a vector field along its own integral curves -- in particular, of the velocity field along itself -- and define its acceleration.

	\begin{mydef}[Acceleration]
	\label{def.acceleration}
	Let us consider a vector field $v\in \GTM$ and $\gamma:I\subset \mathbb{R} \longrightarrow M$ the integral curve of $v$ passing through $p=\gamma(t_0)$ where $t_0\in I$. The acceleration of the curve at the point $p \in M$ is
	\beq
	a\vert_p \equiv \left.\frac{\D}{\d t} v\right\vert_{t=t_0} = \frac{\D}{\d t}\left[\frac{\d}{\d t}\gamma(t) \right]_{t=t_0}=\nabla_v v\vert_p.
	\eeq
	\end{mydef}
	
In this case, the derivative is not identically zero, as with the Lie derivative. Thus, the acceleration measures the failure of the velocity vector field to be propagated parallel along its own integral curves. That is, it is the failure of the integral curves of the vector field $v$ to be straight lines with respect to the transport.
The Lie derivative $L_v v = [v,v]$ vanishes identically [cf. Lemma~\ref{lem.liebracket}] because it compares the velocity field with itself along its own flow, whereas the covariant derivative $\nabla_v v$ compares the transported and the actual velocity at each point, using the independent structure of the propagator -- and only this comparison can fail. Exactly that failure measures the departure of a trajectory from straightness, and it has arisen at the cost of an additional structure for $M$. Moreover, there is an infinity of ways to equip spacetime with a parallel propagator, each defining its own accelerations.  

\subsection{The two derivatives}
\label{subsec.td}

The kinematics of the preceding subsections has furnished two ways of differentiating along a motion, and they come from different places. The Lie derivative was born of the flow [cf. the flow derivative \eqref{eq.lieflow}] and asks how a field is dragged by the motion itself. The covariant derivative was born of the propagator [cf. Definition~\ref{def.covd}] and asks how a field departs from the transport laid down on the manifold. Nothing so far compares them. The comparison is the subject of this subsection.

Two preparations are required. The first is the attribute of a connection that measures the failure of its coefficients to be symmetric.

	\begin{mydef}[Torsion]
	\label{def.torsion}
	The \emph{torsion} of an affine connection $\nabla$ is the map $T:\GTM\times\GTM\longrightarrow\GTM$,
		\beq
		\label{eq.torsion}
		T(u,v) \equiv \nabla_u v - \nabla_v u - [u,v].
		\eeq
	It is antisymmetric, $T(u,v) = -T(v,u)$, and tensorial in each argument; in a coordinate basis $T^k_{ij} = \Gamma^k_{ij} - \Gamma^k_{ji}$. A connection with $T\equiv 0$ is said to be \emph{torsion-free}, and its coefficients are then symmetric in their lower indices.
	\end{mydef}

The second preparation extends both derivatives from vector fields to the multilinear forms upon which they will have to act. The route in each case is the one already travelled, that is, carry the object back to the event at which the derivative is taken and then differentiate, the covariant derivative using the propagator to do the carrying and the Lie derivative using the flow.

	\begin{lemma}[Covariant derivative of a covariant tensor]
	\label{lem.nablatensor}
	Let $A$ be a $k$-linear form on $M$ and $v,u_{(1)},\dots,u_{(k)} \in \GTM$. The covariant derivative of $A$ along $v$ satisfies
		\beq
		\label{eq.nablatensor}
		\left(\nabla_v A\right)\left(u_{(1)},\dots,u_{(k)}\right) = v\left[A\left(u_{(1)},\dots,u_{(k)}\right)\right] - \sum_{i=1}^{k} A\left(u_{(1)},\dots,\nabla_v u_{(i)},\dots,u_{(k)}\right).
		\eeq
	\end{lemma}
	\begin{proof}
	Let $\{e_i\}$ be a basis of $T_p M$ and $\left\{\propagator[\gamma,t_0,t]\,e_i\right\}$ the frame obtained by transporting it parallel along the integral curve $\gamma$ of $v$, so that each argument reads $u_{(j)} = u^{i}_{(j)}(t)\,\propagator[\gamma,t_0,t]e_i$ in that frame. Differentiate the value $A\left(u_{(1)},\dots,u_{(k)}\right)$ along $v$ at $p$. The ordinary Leibniz rule distributes the derivative over the $k$ component functions and over $A$ itself, the last term being, by the very definition of the covariant derivative -- the rate of change under back-transport by the propagator -- the derivative of $A$ evaluated on the transported frame. By Definition~\ref{def.covd} the derivative of the $i$-th family of component functions reproduces $\nabla_v u_{(i)}$ at $p$, and rearranging yields \eqref{eq.nablatensor}. Nowhere does the argument use the rank, for the frame is parallel whatever the number of slots.
	\end{proof}

The rule~\eqref{eq.nablatensor} states that the propagator distributes over a multilinear pairing, and it is the only property of the extended covariant derivative we shall require.

The Lie derivative extends to the same objects by its own route. Dragging a $k$-linear form $A$ back along $\varphi_{-t}$ and differentiating at $t=0$ defines $L_v A$, exactly as the flow derivative \eqref{eq.lieflow} was defined on vector fields. Since the flow carries the pairing $A\left(u_{(1)},\dots,u_{(k)}\right)$ along with its arguments, the Leibniz rule applied to that scalar leaves
	\beq
	\label{eq.lietensor}
	\left(L_v A\right)\left(u_{(1)},\dots,u_{(k)}\right) = v\left[A\left(u_{(1)},\dots,u_{(k)}\right)\right] - \sum_{i=1}^{k} A\left(u_{(1)},\dots,\left[v,u_{(i)}\right],\dots,u_{(k)}\right),
	\eeq
each argument contributing the bracket of Lemma~\ref{lem.liebracket}. The two rules \eqref{eq.nablatensor} and \eqref{eq.lietensor} differ in exactly one respect, the covariant derivative of an argument standing where the Lie derivative has its commutator, and the comparison of the two derivatives is now a subtraction.

	\begin{prop}[Lie--covariant compatibility]
	\label{prop.liecovgen}
	Let $\alpha$ be a $1$-linear form and $v,u \in \GTM$. Then
		\beq
		\label{eq.liecovgen}
		\left(L_v \alpha\right)(u) = \left(\nabla_v \alpha\right)(u) + \alpha\left(\nabla_u v \right) + \alpha\left(T(v,u) \right).
		\eeq
	More generally, for a $k$-linear form $A$,
		\beq
		\label{eq.liecovk}
		\begin{split}
		\left(L_v A\right)\left(u_{(1)},\dots,u_{(k)}\right) & = \left(\nabla_v A\right)\left(u_{(1)},\dots,u_{(k)}\right)\\
		& \quad + \sum_{i=1}^{k} A\left(u_{(1)},\dots,\nabla_{u_{(i)}} v + T\left(v,u_{(i)}\right),\dots,u_{(k)}\right).
		\end{split}
		\eeq
	\end{prop}
	\begin{proof}
	The rule~\eqref{eq.lietensor} at $k=1$ gives $\left(L_v\alpha\right)(u) = v\left[\alpha(u)\right] - \alpha\left(\left[v,u\right]\right)$, and the rule~\eqref{eq.nablatensor} at the same $k$ gives $\left(\nabla_v\alpha\right)(u) = v\left[\alpha(u)\right] - \alpha\left(\nabla_v u\right)$. Subtracting, the terms in $v\left[\alpha(u)\right]$ cancel and
		\beq
		\label{eq.liecovproof}
		\left(L_v\alpha\right)(u) - \left(\nabla_v\alpha\right)(u) = \alpha\left(\nabla_v u - \left[v,u\right]\right) = \alpha\left(\nabla_u v + T(v,u)\right),
		\eeq
	the last equality being the definition \eqref{eq.torsion} of the torsion solved for $\nabla_v u$. For a $k$-linear form the same subtraction is performed slot by slot, the rule~\eqref{eq.lietensor} contributing one commutator and the rule~\eqref{eq.nablatensor} one covariant derivative in each argument, and \eqref{eq.liecovk} follows. Nowhere is any argument of $A$ exchanged with another, so no symmetry of the form is used.
	\end{proof}

The identity~\eqref{eq.liecovgen} is the precise sense in which the two derivatives of this section differ. They differ by the covariant derivative of the velocity field in the direction of the argument, and by the torsion. %
The identity holds for a multilinear form of any rank, and each rank furnishes its own comparison of the two rates of change that a motion admits.

\subsection{Inertial motion}
\label{subsec.IM}

\emph{Can uniformity be formulated in terms of the propagator alone?} In Newton's formulation, inertial motion is both \emph{rectilinear} and \emph{uniform}~\cite{newton1999principia}, which is the content of his First Law. The vanishing of the acceleration accounts for the former, that is, the integral curves of $v$ are straight with respect to the transport defined by $\nabla$. Kinematics, as developed thus far, is exclusively the theory of the parallel propagator, and we would keep it so. Such minimality, however, cannot be realised. Uniformity is a statement about the \emph{magnitude} of the velocity and -- as noted in the remarks following Proposition~\ref{prop.vecspace} -- none of the structures introduced so far assigns a magnitude to a vector. The formulation of inertia is thus the precise point at which Newtonian kinematics demands a second, independent attribute of the spacetime manifold. We introduce it here in its minimal form.

	\begin{mydef}[Bilinear structure]
	\label{def.metric}
	Let $M$ be a differentiable manifold. A \emph{bilinear structure} on $M$ is a smooth assignment of a bilinear form
		\beq
		\label{eq.metric}
		g\vert_p : T_p M \times T_p M \longrightarrow \mathbb{R}
		\eeq
	to each event $p\in M$. The scalar $g(v,v)\vert_p$ is called the \emph{squared magnitude} of the vector $v\vert_p$.
	\end{mydef}

Three remarks are in order. First, we demand of the assignment \eqref{eq.metric} neither symmetry nor non-degeneracy. A \emph{metric} structure proper would require, at the very least, symmetry -- with positive definiteness or indefinite signature as further refinements. Second, the kinematics of this section is altogether blind to the antisymmetric part of $g$. Writing $g = g_s + g_a$ for the symmetric and antisymmetric parts, the squared magnitude satisfies $g(v,v) = g_s(v,v)$ and, since differentiation preserves the symmetry type, every kinematical statement of this section constrains $g_s$ alone. Third, no relation between $g$ and $\nabla$ is assumed, that is, the bilinear structure and the parallel propagator are, for the time being, entirely independent attributes of the spacetime manifold.

Recall that the Lie derivative of the velocity field along its own flow is identically zero -- cf. the trivial condition~\eqref{eq.Xdot}. Thus, when the magnitude $g(v,v)$ is dragged along the flow of $v$, the whole of its change is carried by the Lie transport of $g$ itself, that is,
	\beq
	\label{eq.uniform}
	\frac{\d }{\d t} g(v,v) = L_v\left[g(v,v)\right] = \left(L_v g\right)(v,v) + g\left(L_v v, v \right) + g\left(v, L_v v \right) = \left(L_v g\right)(v,v).
	\eeq
This observation renders the following definition natural.

\begin{mydef}[Inertial motion]
\label{def.inertia}
Let $M$ be a differentiable manifold equipped with an affine connection $\nabla$ and a bilinear structure $g$ [cf. Definition~\ref{def.metric}], not necessarily related to one another. An \emph{inertial} motion in $(M,g,\nabla)$ is the flow [cf. Definition~\ref{def.flow}, the field $v$ here assumed complete] $\varphi:\mathbb{R}\times M \longrightarrow M$ infinitesimally generated by a vector field $v\in \GTM$ which is
	\begin{enumerate}
	\item \emph{rectilinear}, that is, its acceleration is identically zero,
		\beq
		\label{eq.rectilinear}
		\nabla_v v = 0,
		\eeq
	\item and \emph{uniform}, that is, the magnitude of the velocity is preserved by the Lie transport along its own flow,
		\beq
		\label{eq.uniformity}
		L_v\left[g(v,v)\right] = \left(L_v g\right)(v,v) = 0.
		\eeq
	\end{enumerate}
\end{mydef}

The uniformity condition~\eqref{eq.uniformity} is weaker than its appearance suggests, and the distinction governs everything that follows. It is a single scalar equation at each event, asserting the vanishing of $L_v g$ upon the one pair of arguments $\left(v\vert_p,v\vert_p\right)$, and it does not assert $L_v g = 0$. That stronger statement would make the flow of $v$ preserve the bilinear structure entire -- the Killing condition on $v$ -- and a manifold whose bilinear structure admits no such field at all may still admit inertial motions in abundance. Two gaps separate the two statements. The antisymmetric part of $L_v g$ annihilates the diagonal whatever $v$ may be, so the uniformity says nothing whatever about it, as the second of the three remarks following Definition~\ref{def.metric} has already anticipated. And we constrain the symmetric part only where we evaluate it, at a single vector per event, namely the velocity itself.

A flat example settles the matter. On $\mathbb{R}^2$ carrying the Euclidean structure and its flat connection, let $v = f(y)\,\basis[x]$ with $f$ positive and non-constant. Each integral curve is a straight line traversed at the constant speed $f(y)$, so that $\nabla_v v = 0$ and $g(v,v) = f(y)^2$ is constant along the flow, so that $\left(L_v g\right)(v,v) = 0$ and the motion is inertial. Yet $L_v g$ does not vanish, its only non-zero value on the coordinate basis being $\left(L_v g\right)\left(\basis[x],\basis[y]\right) = f'(y)$. Nor would a whole basis of vectors suffice in place of the velocity, for the diagonal of $L_v g$ vanishes on each of $\basis[x]$ and $\basis[y]$ separately while $\left(L_v g\right)(z,z) = 2 f'(y)$ on the diagonal direction $z = \basis[x] + \basis[y]$. The congruence is a shear, a family of parallel straight lines whose speeds differ from one line to the next, and no single particle of it can detect that difference. An inertial motion need not be an isometry of the structure which measures it.

Note that there are as many different types of inertial motion for a given manifold as pairs $(g,\nabla)$ it admits. In particular, the usual Euclidean picture corresponds to a flat connection compatible with $g$, in which case the integral curves of the self-parallel vector fields are straight lines traversed at constant speed.

Before asking when inertial motions exist, let us establish the symmetry principle underlying Newtonian kinematics
	\begin{theo}[Principle of Inertia]
	\label{theo.inertia}
	Let $M$ be a differentiable manifold equipped with an affine connection $\nabla$ and a bilinear structure $g$, and let $v\in \GTM$ be a vector field generating an inertial motion in $(M,g,\nabla)$ whose velocity is nowhere zero. Inertial motion is invariant under affine reparametrisations of the integral curves of $v$, and under these alone.
	\end{theo}
	\begin{proof}
	Let $\gamma$ be an integral curve of the vector field $v\in \GTM$. Then, the coordinates of $\gamma$ satisfy the system of second order differential equations defined by the vanishing of the components of its acceleration [cf. the rectilinear condition~\eqref{eq.rectilinear}],
		\beq
		\label{eq.geodesic}
		a\vert_p = \frac{\D}{\d t}\left[\frac{\d}{\d t}\gamma(t) \right]_{t=t_0} = \left[\frac{\d^2}{\d t^2} x^i (t) + \Gamma^i_{jk}\left[\gamma(t)\right] \dt [t] x^j(t) \dt[t] x^k(t)\right]_{t=t_0} \left.\basis[x^i]\right\vert_p = 0.
		\eeq
	Consider a reparametrisation of the form $\sigma:I\subset \mathbb{R}\longrightarrow \mathbb{R}$ [cf. diagram~\eqref{diag.repar}] with $s=\sigma(t)$ and $\sigma'(t)\neq 0$ for every $t \in I$, where the prime denotes the derivative of $\sigma$ with respect to its argument. The coordinates $y^i = y^i(s)$ of the reparametrised curve $\tilde\gamma = \gamma \circ \sigma^{-1}$ satisfy $y^i\left[\sigma(t)\right] = x^i(t)$. Differentiating this relation twice with respect to $t$ and solving for the derivatives of $y^i$ we obtain
		\beq
		\label{eq.chainrepar}
		\frac{\d}{\d s} y^i(s) = \frac{1}{\sigma'}\, \frac{\d}{\d t} x^i(t) \quad \text{and} \quad \frac{\d^2}{\d s^2} y^i(s) = \frac{1}{\left(\sigma'\right)^2}\, \frac{\d^2}{\d t^2} x^i(t) - \frac{\sigma''}{\left(\sigma'\right)^3}\, \frac{\d}{\d t} x^i(t).
		\eeq
	Substituting the derivatives \eqref{eq.chainrepar} into the acceleration of $\tilde\gamma$ yields the transformation law
		\begin{align}
		\label{eq.acctransf}
		\tilde a^i 	& = \frac{\d^2}{\d s^2} y^i + \Gamma^i_{jk} \frac{\d}{\d s} y^j\, \frac{\d}{\d s} y^k\nn\\
					& = \frac{1}{\left(\sigma'\right)^2}\left[\frac{\d^2}{\d t^2} x^i + \Gamma^i_{jk} \frac{\d}{\d t} x^j\, \frac{\d}{\d t} x^k\right] - \frac{\sigma''}{\left(\sigma'\right)^3}\, \frac{\d}{\d t} x^i = - \frac{\sigma''}{\left(\sigma'\right)^3}\, \frac{\d}{\d t} x^i,
		\end{align}
	where the last equality follows from the geodesic equation~\eqref{eq.geodesic}. Since the velocity of the motion is nowhere zero, the acceleration of the reparametrised curve vanishes if and only if $\sigma'' = 0$, that is, if and only if
		\beq
		\label{eq.affinerepar}
		\sigma(t) = \alpha t + \beta \quad \text{with} \quad \alpha,\beta \in \mathbb{R}, \quad \alpha \neq 0.
		\eeq
	Thus, the rectilinear condition~\eqref{eq.rectilinear} is preserved precisely by the affine reparametrisations of $\gamma$.

	Regarding uniformity, recall from the identification \eqref{eq.reparam} of Proposition~\ref{prop.reparam} that the tangent field of the reparametrised curve is $\tilde v = f v$ with $f = \left(\sigma'\right)^{-1}$ along $\gamma$, so that its squared magnitude is $g(\tilde v, \tilde v) = f^2 g(v,v)$. Being a scalar, its change along the flow of $\tilde v$ -- using the Leibniz rule in the second equality and the uniformity \eqref{eq.uniformity} of the original motion, $v[g(v,v)]=0$, in the third -- reads
		\beq
		\label{eq.uniftransf}
		L_{\tilde v}\left[g(\tilde v,\tilde v)\right] = f\, v\left[f^2 g(v,v)\right] = f^3\, v\left[g(v,v)\right] + 2 f^2\, v[f]\, g(v,v) = 2 f^2\, v[f]\, g(v,v).
		\eeq
	Whenever $g(v,v) \neq 0$, the right hand side vanishes if and only if $v[f] = 0$, that is, if and only if $\sigma'$ is constant along $\gamma$ -- once more, the affine reparametrisations \eqref{eq.affinerepar}. Along a null velocity the right hand side vanishes whatever the reparametrisation, so uniformity constrains the clock there not at all. The rectilinear condition has already confined the freedom to \eqref{eq.affinerepar} in either case, and inertial motion is therefore preserved by the affine reparametrisations of the integral curves of $v$ and by no others.
	\end{proof}

The Principle of Inertia thus identifies the affine reparametrisations as the symmetries of inertial motion, that is, the freedom left in the description of an inertial particle is precisely the choice of origin and unit for its time parameter. Note that in the indefinite case there may be motions with $g(v,v) = 0$ everywhere, for which the uniformity condition~\eqref{eq.uniformity} is trivially preserved by every reparametrisation; for these, the rectilinear condition alone singles out the affine class.

The symmetry carries a warning. Rectilinearity -- and with it inertia -- is a property not of the \emph{path} traced by the particle, but of the parametrised motion. A non-affine reparametrisation of an inertial motion traverses the very same events while failing to be rectilinear, and is therefore no longer inertial. Consequently, any demand concerning the \emph{existence} of inertial motions must be formulated modulo the affine symmetry, that is, the datum is an event together with a direction, and the motion is to be sought in some -- necessarily affine -- parametrisation.

The reader may wonder whether the two conditions in Definition~\ref{def.inertia} are independent. To settle the question we must first understand how the covariant derivative interacts with the bilinear structure. Following the same route as with vector fields -- back-transport by the parallel propagator, then differentiate -- the derivative of $g$ along $v$ inherits a Leibniz rule with respect to the natural pairing.

	\begin{lemma}[Covariant derivative of the bilinear structure]
	\label{lem.nablag}
	Let $g$ be a bilinear structure on $M$ and $u,v,w\in \GTM$. The covariant derivative of $g$ along $v$ satisfies
		\beq
		\label{eq.nablag}
		\left(\nabla_v g\right)(u,w) = v\left[g(u,w) \right] - g\left(\nabla_v u, w \right) - g\left(u, \nabla_v w \right).
		\eeq
	\end{lemma}
	\begin{proof}
	The bilinear structure is a $2$-linear form, so this is the case $k=2$ of the rule~\eqref{eq.nablatensor}, rearranged.
	\end{proof}

The vanishing of the derivative \eqref{eq.nablag} has a transparent meaning in terms of our primitive notion of transport, that is, the parallel propagator is an \emph{isometry} between tangent spaces, preserving the pairing infinitesimally,
	\beq
	\label{eq.isometry}
	v\left[g(u,w)\right] = g\left(\nabla_v u, w \right) + g\left(u, \nabla_v w\right).
	\eeq
In this case, we say the connection $\nabla$ is \emph{compatible} with the bilinear structure $g$.

Now, evaluating the definition \eqref{eq.nablag} on the triple $(v,v,v)$ and rearranging, the change of the magnitude along the flow reads
	\beq
	\label{eq.magrate}
	v\left[g(v,v)\right] = \left(\nabla_v g\right)(v,v) + g\left(\nabla_v v, v\right) + g\left(v, \nabla_v v\right).
	\eeq
Hence, along a rectilinear motion -- one satisfying the self-parallel condition~\eqref{eq.rectilinear} -- the rate of change of the magnitude reduces to $\left(\nabla_v g\right)(v,v)$, which need not vanish for a generic pair $(g,\nabla)$. For a compatible connection -- cf. the isometry property~\eqref{eq.isometry} -- uniformity is thus a consequence of straightness.

The converse fails, and the rate \eqref{eq.magrate} shows exactly how. For a compatible connection its first term drops and uniformity becomes
	\beq
	\label{eq.uniformorth}
	g\left(\nabla_v v, v\right) + g\left(v, \nabla_v v\right) = 2\, g_s\left(\nabla_v v, v\right) = 0,
	\eeq
which demands of the acceleration not that it vanish but only that it be orthogonal to the velocity in the symmetric part of the pairing. Rectilinearity is therefore sufficient for uniformity without being necessary, and the two conditions of Definition~\ref{def.inertia} are independent of one another.

A charged particle furnishes the physical instance. Let $g$ be symmetric, let $F$ be an antisymmetric bilinear structure on $M$, and let a motion have its acceleration fixed by
	\beq
	\label{eq.lorentz}
	g\left(\nabla_v v, u\right) = \kappa\, F(v,u),
	\eeq
for every $u\in \GTM$, with $\kappa$ a constant. Setting $u=v$ and invoking the antisymmetry of $F$ gives $g\left(\nabla_v v, v\right) = \kappa\, F(v,v) = 0$, so that the condition~\eqref{eq.uniformorth} holds while the acceleration is nowhere zero. This is the Lorentz force, with $F$ the electromagnetic field and $\kappa$ the ratio of the charge of the particle to its mass, and the motions it produces are the magnetic curves~\cite{lopezmonsalvo2025noether}. The field bends such a motion at every event, and it is nevertheless traversed at an unchanging magnitude of velocity. It satisfies the second condition of Definition~\ref{def.inertia} and fails the first. The antisymmetry of $F$ is the whole of the reason, and it is worth noticing that we have used no property of the electromagnetic field beyond that antisymmetry.

One purely algebraic fact intervenes twice, and we record it before it is needed. A multilinear form is not determined by its values on the diagonal, where all its arguments coincide; what those values do determine is its symmetric part, and they determine it completely.

	\begin{lemma}[Polarisation]
	\label{lem.polarisation}
	Let $V$ be a real vector space, $B:V^k\longrightarrow \mathbb{R}$ a $k$-linear form and $B_\Delta(z) \equiv B(z,\dots,z)$ its restriction to the diagonal. If $B_\Delta$ vanishes identically, so does the total symmetrisation of $B$, that is, the sum of $B$ over the $k!$ permutations of its arguments. In particular, a symmetric $k$-linear form vanishing on the diagonal is itself zero.
	\end{lemma}
	\begin{proof}
	Expand $B_\Delta\!\left(\sum_{i\in S} v_i\right)$ by multilinearity for each non-empty $S \subseteq \{1,\dots,k\}$ and form the alternating sum
		\beq
		\label{eq.polarisation}
		\sum_{\emptyset \neq S \subseteq \{1,\dots,k\}} (-1)^{k-\vert S \vert}\, B_\Delta\!\left(\sum_{i\in S} v_i\right) = \sum_{\rho} B\left(v_{\rho(1)},\dots,v_{\rho(k)}\right),
		\eeq
	the sum on the right running over the $k!$ permutations $\rho$ of $\{1,\dots,k\}$. The identity is due to Mazur and Orlicz~\cite{mazur1934grundlegende}, and Thomas gives a short proof of it by iterated difference operators~\cite{thomas2014polarization}; both state it for a symmetric form, where the right hand side reads $k!\,B$, and the general case follows because the restriction to the diagonal sees only the total symmetrisation. Indeed, a term of the expansion in which some $v_i$ is absent occurs in the subsets $S$ containing the arguments it does use and in all their supersets missing $i$, with alternating signs, and cancels; the terms surviving are precisely those using each argument exactly once, which is the right hand side. Every summand on the left vanishes by hypothesis, hence the total symmetrisation of $B$ vanishes. If $B$ is symmetric, each of the $k!$ summands on the right equals $B(v_1,\dots,v_k)$ and the conclusion is $B = 0$.
	\end{proof}

The identity~\eqref{eq.polarisation} is a statement about multilinear algebra alone, owing nothing to the manifold or to the transport. We use it for a trilinear form assembled from the derivative of $g$, and once more for a bilinear one.

Remarkably, compatibility is sufficient but not necessary, and the precise demand that inertia places upon the pair $(g,\nabla)$ is the content of the following

	\begin{prop}[Inertia and compatibility]
	\label{prop.killing}
	Let $M$ be a differentiable manifold equipped with an affine connection $\nabla$ and a bilinear structure $g$. The following statements are equivalent.
	\begin{enumerate}
	\item The totally symmetric part of the derivative \eqref{eq.nablag} vanishes, that is,
		\begin{align}
		\label{eq.killing}
		&\left(\nabla_u g\right)(v,w) + \left(\nabla_v g\right)(w,u) + \left(\nabla_w g\right)(u,v)\nn\\
		&\qquad + \left(\nabla_u g\right)(w,v) + \left(\nabla_v g\right)(u,w) + \left(\nabla_w g\right)(v,u) = 0,
		\end{align}
	for every $u,v,w \in \GTM$.
	\item Every rectilinear motion in $(M,g,\nabla)$ is uniform.
	\item Through every event $p\in M$ and every direction at $p$ there passes an inertial motion, unique up to the affine freedom of Theorem~\ref{theo.inertia}.
	\end{enumerate}
	\end{prop}
	\begin{proof}
	Denote by $S(u,v,w)$ the left hand side of the condition~\eqref{eq.killing} -- the sum of $\left(\nabla_u g\right)(v,w)$ over the six permutations of the arguments -- which is totally symmetric by construction, whatever the symmetries of $g$. On the diagonal it reduces to
		\beq
		\label{eq.Sdiag}
		S(v,v,v) = 6 \left(\nabla_v g \right)(v,v).
		\eeq

	\emph{(1) implies (2).} Along any rectilinear motion, the self-parallel condition~\eqref{eq.rectilinear} together with the rate \eqref{eq.magrate} give
		\beq
		\label{eq.rectuniform}
		v\left[g(v,v)\right] = \left(\nabla_v g\right)(v,v) = \tfrac{1}{6}\, S(v,v,v) = 0,
		\eeq
	and the motion is uniform.

	\emph{(2) implies (3).} Fix an event $p \in M$ and a vector $v\vert_p \in T_p M$ along the prescribed direction. The rectilinear condition is, in coordinates, the second order system \eqref{eq.geodesic}; by the standard existence theorem for such systems, its solutions realise the initial pair $\left(p, v\vert_p\right)$ and assemble smoothly, in a neighbourhood of $p$, into a vector field all of whose integral curves are rectilinear. The flow of this field is rectilinear by construction and uniform by the hypothesis (2), hence an inertial motion through the prescribed event and direction. As for uniqueness, any two rectilinear motions through $p$ with proportional velocities at $p$ are affine reparametrisations of one another, by the uniqueness of solutions of the system \eqref{eq.geodesic} together with the transformation law \eqref{eq.acctransf}.

	\emph{(3) implies (1).} Fix an event $p$ and any $v\vert_p \neq 0$. By hypothesis, some inertial motion passes through $p$ along the direction of $v\vert_p$; its generating field $u$ satisfies $u\vert_p = \lambda\, v\vert_p$ with $\lambda \neq 0$. Being rectilinear and uniform, the rate \eqref{eq.magrate} along this motion yields $\left(\nabla_u g\right)(u,u)\vert_p = 0$; since this expression is pointwise and cubic in $u\vert_p$,
		\beq
		\label{eq.killingdiag}
		\left(\nabla_v g\right)(v,v)\vert_p = \lambda^{-3} \left(\nabla_u g\right)(u,u)\vert_p = 0.
		\eeq
	By the diagonal evaluation \eqref{eq.Sdiag}, $S(v,v,v)\vert_p = 0$ for every $v\vert_p \in T_p M$, since the case $v\vert_p = 0$ is trivial. Being totally symmetric and vanishing on the diagonal, $S$ vanishes identically by Lemma~\ref{lem.polarisation}, which is the condition~\eqref{eq.killing}.
	\end{proof}

Proposition~\ref{prop.killing} identifies the precise demand that the Principle of Inertia places upon the pair $(g,\nabla)$ and, in accordance with the second remark following Definition~\ref{def.metric}, it constrains the symmetric part of $g$ alone. When the bilinear structure is symmetric, the condition~\eqref{eq.killing} collapses to twice its cyclic part,
	\beq
	\label{eq.killingcyc}
	\left(\nabla_u g\right)(v,w) + \left(\nabla_v g\right)(w,u) + \left(\nabla_w g\right)(u,v) = 0,
	\eeq
the statement that $g$ is a \emph{Killing tensor} of the connection, that is, the squared magnitude is a first integral of every rectilinear motion. This is strictly weaker than compatibility. Indeed, let $g$ be symmetric, let $\nabla'$ be any connection compatible with $g$, let $Z\in \GTM$ be an arbitrary vector field, and set $\theta \equiv g(Z,\cdot)$. The deformed connection
	\beq
	\label{eq.thetafamily}
	\nabla_u w = \nabla'_u w + \theta(u)\, w + \theta(w)\, u - 2\, g(u,w)\, Z
	\eeq
is torsion-free whenever $\nabla'$ is, satisfies the Killing condition~\eqref{eq.killing} for every choice of $Z$, and yet a direct computation gives
	\beq
	\label{eq.nablagtheta}
	\left(\nabla_u g\right)(v,w) = g(u,v)\,\theta(w) + g(u,w)\,\theta(v) - 2\, g(v,w)\,\theta(u),
	\eeq
which is non-zero whenever $\theta$ is. Every rectilinear motion of the family \eqref{eq.thetafamily} is therefore uniform, while its propagator fails to preserve the pairing. \emph{Inertia alone cannot tell these connections from a compatible one}.

Full compatibility acquires, in turn, its own kinematical characterisation. Let us demand that inertial motion preserve the magnitude not only of its own velocity but of \emph{every} vector it carries along -- that is, that observers moving inertially may transport \emph{standards of magnitude} unchanged. For a vector field $w$ propagated parallel along a rectilinear motion, $\nabla_v w = 0$, the rate of change of its squared magnitude is $v\left[g(w,w)\right] = \left(\nabla_v g\right)(w,w)$, and the demand reads $\left(\nabla_v g\right)(w,w) = 0$ for arbitrary and independent $v, w \in \GTM$. Lemma~\ref{lem.polarisation}, applied at each fixed $v$ to the bilinear form $\left(\nabla_v g\right)(\cdot,\cdot)$ in the transported argument, now forces the vanishing of the symmetric part of $\nabla_v g$ in its last two slots -- the compatibility $\nabla g_s = 0$ of everything kinematics can see. The full isometry property~\eqref{eq.isometry} is a genuinely stronger assumption only as regards the antisymmetric part of $g$, to which every statement of this section is blind. In short, uniformity of the trajectories demands the Killing condition~\eqref{eq.killing}; transportability of magnitudes demands the compatibility of the symmetric part.

The whole of the kinematics developed in this section can now be gathered in a single commutative diagram, in the same spirit as the diagrams that have guided us throughout. Let $v \in \GTM$ generate the flow $\varphi$ [cf. Definition~\ref{def.flow}], let $p = \gamma(t_0)$ be an event on one of its integral curves, and consider
	\beq
	\label{diag.inertia}
	\begin{tikzpicture}[]
	\matrix[matrix of math nodes,column sep={100pt,between origins},row
	sep={70pt,between origins},nodes={asymmetrical rectangle}] (s)
	{
	|[name=a1]| M				& |[name=a2]| 				& |[name=a3]| M							\\
	|[name=b1]| T_{p} M			& |[name=b2]| 				& |[name=b3]| T_{\varphi_t(p)} M		\\
	|[name=c1]| 				& |[name=c2]| \mathbb{R}.	& |[name=c3]| 							\\
	}
	;
	\draw[->]
			(a1) edge node[above] {\(\varphi_t\)} (a3)
			(a1) edge node[left] {\(v\)} (b1)
			(a3) edge node[right] {\(v\)} (b3)
			(b1) edge node[above] {\(\propagator[\gamma,t_0,t_0+t]\)} (b3)
			(b1) edge node[below] {\(g\left(\,\cdot\,,\cdot\,\right)\ \) } (c2)
			(b3) edge node[below] {\(\ \ g\left(\,\cdot\,,\cdot\,\right)\)} (c2)
	;
	\end{tikzpicture}
	\eeq
Here, the vertical arrows assign to each event the velocity of the motion, the upper horizontal arrow advances the events along the flow, the middle horizontal arrow is the parallel propagator along the integral curve, and the two lower legs evaluate the squared magnitude \eqref{eq.metric} on the corresponding tangent spaces. The flow $\varphi$ is an inertial motion in $(M,g,\nabla)$ if and only if the diagram~\eqref{diag.inertia} commutes for every $t$. Commutativity of the upper square is the rectilinear condition~\eqref{eq.rectilinear}, that is, the transport carries the velocity into the velocity,
	\beq
	\label{eq.diaginertia}
	\propagator[\gamma,t_0,t_0+t]\left[v\vert_{p} \right] = v\vert_{\varphi_t(p)}.
	\eeq
Granted it, the agreement of the two descents from the upper row to $\mathbb{R}$ is the uniformity condition~\eqref{eq.uniformity}, that is, the squared magnitude of the velocity is the same whether it is read at departure or at arrival. The Principle of Inertia (Theorem~\ref{theo.inertia}) completes the picture. The reparametrisations of the integral curves preserving the commutativity of the diagram -- the propagator being indifferent to them by the parametrisation independence \eqref{eq.reparinv} -- are precisely the affine ones.

Before turning to the symmetries, let us take stock of the assumptions actually in force. The kinematics of this section rests on exactly two independent structures, namely, a parallel propagator -- equivalently, an affine connection $\nabla$ -- and a bilinear structure $g$. Of $g$ we have required neither symmetry, nor non-degeneracy, nor any signature; of $\nabla$, no torsion condition whatsoever; and of the pair, no compatibility beyond what inertia itself imposes -- the vanishing \eqref{eq.killing} of the symmetrised derivative of $g$, strengthened to the compatibility of its symmetric part if magnitudes are to be transportable. Symmetry, non-degeneracy, signature, full compatibility and torsion-freeness are none of them demanded by the description of motion, and none of them is assumed here.

\subsection{The symmetries of inertial motion}

The Principle of Inertia was established before the question of existence was raised, and it is the first of the symmetries of inertial motion. It concerns the parameter -- an inertial motion remains inertial under the affine reparametrisations of its integral curves, and under no others (Theorem~\ref{theo.inertia}). The freedoms that remain are of an entirely different nature. They act not upon the clock that times the motion but upon the two structures out of which the motion was built, and each exposes a datum of those structures that inertial motion cannot resolve. They are two faces of a single principle, and we establish them in turn before gathering them into it.

The first face concerns the bilinear form. Of the two structures from which inertial motion was built, $g$ enters only up to its overall scale, for multiplying it by a positive constant alters nothing that has been said. The freedom deserves to be recorded precisely.

	\begin{prop}[Homothety invariance of inertial motion]
	\label{prop.homothety}
	Let $\lambda > 0$ be a constant. The inertial motions of $(M,g,\nabla)$ and of $(M,\lambda g,\nabla)$ coincide; moreover the compatibility of $\nabla$ with $g$ -- the isometry property~\eqref{eq.isometry} -- and the Killing condition~\eqref{eq.killing} are unchanged by the rescaling $g \mapsto \lambda g$. The inertial structure depends on the bilinear form only through its homothety class. Conversely, let $f:M\longrightarrow \mathbb{R}_{>0}$ be smooth and let the inertial motions of $(M,f g,\nabla)$ coincide with those of $(M,g,\nabla)$. Then $f$ is constant along every inertial motion of non-null velocity and, if $M$ is connected, the symmetric part of $g$ vanishes at no event, and inertial motions pass through every event in every direction [cf. Proposition~\ref{prop.killing}], $f$ is constant.
	\end{prop}
	\begin{proof}
	The rectilinear condition~\eqref{eq.rectilinear} does not involve $g$, and the uniformity \eqref{eq.uniformity}, compatibility \eqref{eq.isometry} and Killing condition~\eqref{eq.killing} are each homogeneous of the first degree in $g$, hence insensitive to a constant factor. Every condition defining an inertial motion is blind to the scale of $g$.

	For the converse, rectilinearity is again untouched, while the uniformity of the rescaled structure along a motion uniform for the original one reads
		\beq
		\label{eq.uniflocal}
		L_v\left[f\, g(v,v)\right] = f\, L_v\left[g(v,v)\right] + v[f]\, g(v,v) = v[f]\, g(v,v),
		\eeq
	which vanishes if and only if $v[f] = 0$ wherever $g(v,v) \neq 0$. Granted inertial motions through every event and direction, $\d f\vert_p$ annihilates every non-null vector at every $p$. Wherever the symmetric part of $g\vert_p$ does not itself vanish those vectors span $T_p M$. Indeed, let $v\vert_p$ be non-null and $u\vert_p$ null. The squared magnitude along $v\vert_p + t\, u\vert_p$ is the affine function $g(v,v) + 2t\, g_{(\rm s)}(v,u)$ of the real parameter $t$, which fails to vanish for all but at most one value of it, so that family contains two non-null vectors whose difference is a non-zero multiple of $u\vert_p$. Hence $\d f = 0$ and $f$ is constant on a connected $M$.
	\end{proof}

The uniqueness we recover here is the elementary, single-metric shadow of a harder question. \emph{Given only $\nabla$, how many bilinear structures are compatible with it in the sense of the isometry property~\eqref{eq.isometry}?} The holonomy representation of the connection governs the answer, which collapses to a constant scale only when that representation is irreducible~\cite{schmidt1973conditions}. We bypass that global question here, since $g$ is given from the outset and only its own ray is at stake.

This freedom is invisible to the whole of the kinematics. Whatever physical attribute the scale of the bilinear form may be called upon to carry, inertial motion cannot resolve it, for every representative of a single homothety class prescribes the very same inertial trajectories, and no observation of a particle in inertial motion can tell one representative from another.

The connection carries the companion freedom, the second face of the same principle. Where the bilinear form entered the inertial structure only up to its scale, the connection enters only through part of itself; the remainder is a tensor that no inertial motion can detect. That tensor is the torsion of Definition~\ref{def.torsion}, which entered above only as a guarantee that the deformed connections \eqref{eq.thetafamily} transported in no exotic manner, and which we now examine in its own right.

	\begin{prop}[Torsion invariance of inertial motion]
	\label{prop.torsion}
	Let $\nabla$ and $\bar\nabla$ be affine connections on $M$ with the same symmetric part, $\bar\Gamma^k_{(ij)} = \Gamma^k_{(ij)}$ -- equivalently, differing only in their torsion \eqref{eq.torsion}. Then the inertial motions of $(M,g,\nabla)$ and of $(M,g,\bar\nabla)$ coincide. The inertial structure depends on the connection only through its symmetric part.
	\end{prop}
	\begin{proof}
	In the rectilinear condition~\eqref{eq.rectilinear} the connection coefficients are contracted against the symmetric product $v^i v^j$ of the velocity components, so only the symmetric part $\Gamma^k_{(ij)}$ enters the acceleration $\nabla_v v$. Connections sharing that symmetric part assign the same acceleration to every field, $\nabla_v v = \bar\nabla_v v$, hence the same rectilinear motions, while the uniformity condition~\eqref{eq.uniformity} does not involve the connection at all. Every condition defining an inertial motion is blind to the torsion of $\nabla$.
	\end{proof}

The torsion is to the connection what the scale was to the bilinear form, namely, a part of the structure that inertial motion cannot resolve. Here it enters solely as a symmetry of inertial motion, and nothing in the description of motion constrains it.

The two freedoms are one statement, and we record it as such.

	\begin{theo}[Inertial-structure equivalence]
	\label{theo.ise}
	The inertial motions of $(M,g,\nabla)$ depend on the pair $(g,\nabla)$ only through the homothety class of $g$ and the symmetric part of $\nabla$. Consequently no observation of a particle in inertial motion resolves either the scale of the bilinear form or the torsion of the connection.
	\end{theo}
	\begin{proof}
	The two data are independent attributes of the pair, and each is erased in turn by the homothety invariance of Proposition~\ref{prop.homothety} and the torsion invariance of Proposition~\ref{prop.torsion}.
	\end{proof}

Theorem~\ref{theo.ise} asserts that blindness, and not its converse, for two pairs may share their inertial motions without differing by a rescaling and a torsion. On a reducible manifold the bilinear structures $g_1 \oplus g_2$ and $g_1 \oplus c\, g_2$ admit, for every constant $c>0$, one and the same torsion-free compatible connection, and neither is a rescaling of the other. That is the reducible case of the remark following Proposition~\ref{prop.homothety}, and the reason we state the theorem in one direction only.

Each of the structures we gave the manifold thus carries a datum that the description of motion leaves unfixed, the scale in the case of the bilinear form and the torsion in the case of the connection, and inertial motion is blind to both. Whatever physical attributes we call upon these two data to carry, bodies differing in them alone traverse the very same inertial trajectories. Such an indifference -- obtained from the mere blindness of the inertial structure, and before any notion of gravity has been introduced -- is the geometric precursor of the \emph{universality of free fall}, and we shall refer to Theorem~\ref{theo.ise} as \emph{inertial-structure equivalence}. The question of which features of a geometry are fixed by observation, and which remain free, is the subject of the axiomatic programme of Ehlers, Pirani and Schild~\cite{ehlers1972geometry}, where light propagation determines the conformal structure of spacetime and free fall its projective one. Theorem~\ref{theo.ise} answers the same question for the two data that the present construction leaves unfixed. The equivalence principle, in its usual non-relativistic form, is a statement about gravitation, namely that the coupling to the Newtonian potential is measured by the same mass that measures the inertia. A geometric reading trades the potential for a curved and time-dependent spatial metric, which forces the equality of the two masses by that very requirement~\cite{kapustin2021nonrelativistic}. We have introduced no gravitational coupling here, and none is needed. What Theorem~\ref{theo.ise} asserts is that inertial motion cannot resolve two data of the structures from which it was built, so that bodies differing in those data alone traverse the same trajectories. That a gravitational mass should afterwards prove equal to the inertial one is a further and separate claim, of which the present blindness is the geometric precursor.

Inertial motion was built from a parameter and two structures, and to each answers a symmetry, a transformation carrying every inertial motion to an inertial motion whatever the manifold on which it takes place. The Principle of Inertia (Theorem~\ref{theo.inertia}) acts on the parameter, namely, the affine reparametrisations of the worldline. Inertial-structure equivalence (Theorem~\ref{theo.ise}) acts on the structures, in its two faces -- the constant rescalings of the bilinear form (Proposition~\ref{prop.homothety}) and the alteration of the torsion of the connection (Proposition~\ref{prop.torsion}). An inertial particle is thus indifferent to the clock that times it, to the scale that weighs it and to the twist of the transport that carries it, and the last two indifferences are one and the same.

\section{Dynamics}
\label{sec.dyn}

Kinematics described the motion of a single particle independently of its causes, and everything it built -- the tangent vector, the parallel propagator, the connection, the bilinear form -- lived on the tangent bundle of the spacetime manifold alone. Dynamics is where the causes of motion enter, and in this section we develop, in the same constructive spirit, the complementary \emph{dual} framework in which they find their natural expression.

The objects of dynamics inhabit the dual bundle. We denote by $T^* M$ the \emph{cotangent bundle}, the collection of the duals $T^*_p M$ of the tangent spaces, whose elements -- \emph{1-forms} or co-vectors -- act linearly on tangent vectors to return a real number. As with the vector fields of the previous section, we write $\GTsM$ for its smooth sections.

Dynamics needs one further operation, the pairing of a co-vector with the velocity of a motion, dual to the evaluation of a function on a vector. This pairing is the \emph{interior product}.

	\begin{mydef}[Interior product]
	\label{def.interior}
	Let $\omega$ be a $k$-form and $v\in \GTM$ a vector field. The interior product $\iota_v \omega$ is the $(k-1)$-form obtained by inserting $v$ into the first slot of $\omega$,
		\beq
		\label{eq.interior}
		\iota_v \omega\left[u_{(1)},\hdots,u_{(k-1)}\right] = \omega \left[v, u_{(1)},\hdots, u_{(k-1)}\right],
		\eeq
	where $u_{(i)} \in \GTM$.
	\end{mydef}

For a 1-form $\alpha$ the interior product reduces to the pairing $\iota_v \alpha = \alpha(v)$, a scalar.

With the interior product at our disposal, together with the exterior derivative, we may ask how a 1-form $\alpha \in \GTsM$ changes along the flow of a vector field $v$. \emph{Cartan's identity} furnishes the answer, and we adopt it as the definition of the Lie derivative of a 1-form,
	\beq
	\label{eq.momch}
	\frac{\d }{\d t} \alpha \equiv L_v \alpha = \iota_v \d \alpha + \d \left(\iota_v \alpha\right).
	\eeq

The first term on the right hand side of~\eqref{eq.momch} acts on a vector field $u$ to give
	\beq
	 \left[\iota_v \d \alpha\right] (u) = v\left[\alpha(u) \right] - u \left[\alpha(v) \right] - \alpha\left( [v,u]\right),
	\eeq
while the second is
	\beq
	\left[\d \iota_v  \alpha\right] \left(u \right) = \d \left[\alpha(v) \right] (u) = u \left[\alpha(v) \right],
	\eeq
where the last equality follows from the definition of the differential of a function, that is, $\d f (u) = u (f)$.
Thus the change of a 1-form $\alpha$ along a curve $\gamma$, acting on an arbitrary vector field $u$, reads
	\begin{align}
	\label{eq.dp}
	\left[L_v \alpha \right](u)	& = v\left[\alpha(u) \right] - \alpha\left([v,u] \right)\nonumber\\
				& = L_v \left[\alpha(u) \right] - \alpha \left[L_v u \right]\nonumber\\
				& = \left\{L_v,\alpha \right\}(u).
	\end{align}
Here, the bracket denotes the \emph{commutator} of the Lie derivative and the 1-form $\alpha$ acting on a vector field $u$. Its content appears in the square
	\beq
	\label{diag.dp}
	\begin{tikzpicture}[]
	\matrix[matrix of math nodes,column sep={130pt,between origins},row
	sep={60pt,between origins},nodes={asymmetrical rectangle}] (s)
	{
	|[name=a1]| \GTM	& |[name=a2]| C^\infty(M)	\\
	|[name=b1]| \GTM	& |[name=b2]| C^\infty(M).	\\
	}
	;
	\draw[->]
			(a1) edge node[above] {\(\alpha\)} (a2)
			(a1) edge node[left]  {\(L_v = \left[v,\,\cdot\,\right]\)} (b1)
			(a2) edge node[right] {\(L_v = v\left[\,\cdot\,\right]\)} (b2)
			(b1) edge node[below] {\(\alpha\)} (b2)
	;
	\end{tikzpicture}
	\eeq
Its horizontal arrows contract a vector field with $\alpha$, while its vertical arrows are the Lie derivative -- acting as the bracket $\left[v,\,\cdot\,\right]$ on vector fields and as $v\left[\,\cdot\,\right]$ on functions. By the commutator structure \eqref{eq.dp}, the two paths from $\GTM$ to $C^\infty(M)$ differ by $\left(L_v\alpha\right)(u)$, so the square commutes if and only if $L_v\alpha = 0$.

\subsection{Momentum and force}
Velocity -- an element of the tangent bundle -- is the whole of the kinematic content of a motion. Dynamics turns instead on the \emph{dual} objects. The change of a 1-form along the flow is what equation~\eqref{eq.dp} measures. To carry the kinematic velocity into this dual arena we require a map from the tangent to the cotangent bundle; the bilinear structure of Definition~\ref{def.metric} -- introduced in the previous section merely to select the uniform motions -- furnishes precisely one, and we read the passage it effects as a \emph{constitutive relation}. Every bilinear structure induces the \emph{musical} map
	\beq
	\label{eq.gflat}
	g^\flat: 	 T_p M	 \longrightarrow 	 T^*_p M \quad  \text{with} \quad u \longmapsto 	u^\flat = g^\flat (u) \equiv g\left(u, \cdot \right),
	\eeq
whose coordinate expression is $u^\flat = g_{ij} u^j\, \d x^i$, where $g_{ij}$ denote the components of $g$ in the coordinate basis. We recover the squared magnitude of a vector through the natural pairing,
	\beq
	\label{eq.flatpairing}
	u^\flat(u) = g_{ij}\, u^i u^j = g(u,u).
	\eeq
When the map \eqref{eq.gflat} is a linear isomorphism at every event -- in which case we say the bilinear structure is \emph{non-degenerate} -- it possesses an inverse
	\beq
	\label{eq.gsharp}
	g_\sharp \equiv \left(g^\flat\right)^{-1}:T^*_p M \longrightarrow T_p M \quad  \text{with} \quad \omega \longmapsto \omega_\sharp = g_\sharp (\omega) = g^{ij} \omega_j \frac{\partial}{\partial x^i},
	\eeq
where $g^{ij}$ are the components of the matrix inverse of $\left(g_{ij}\right)$. A symmetric and non-degenerate bilinear structure is called a \emph{metric}; if the metric is, moreover, positive definite, the pair $(M,g)$ is a \emph{Riemannian} manifold, and otherwise we say $(M,g)$ is \emph{indefinite}.

More than a geometric convenience, a bundle morphism of this kind -- from the tangent to the cotangent bundle -- is an attribute by which one particle differs from another. A particle has no extension, so one is told from another by the \emph{constitutive data} it carries, the attributes that decide how we convert the kinematic objects of the previous section into dual ones. That a body should be without extension does not deprive it of structure, and we shall not presume at the outset how many such data there are. The morphism before us is the first of them, its \emph{constitutive morphism}. Which co-vector a given velocity receives depends upon it alone, and the whole of its freedom is an overall scale relative to the ambient map \eqref{eq.gflat}.

That scale is exactly the datum the kinematics of the previous section cannot resolve. There we let the bilinear structure enter only through the squared magnitude $g(v,v)$ that selects the uniform motions and through the connection compatible with it, and a constant rescaling $g \mapsto \lambda g$ leaves both unchanged. Kinematics therefore fixes the musical map \eqref{eq.gflat} only up to its \emph{homothety class}
	\beq
	\label{eq.homothety}
	\left[g^\flat\right] \equiv \left\{\lambda\, g^\flat \; : \; \lambda \in \mathbb{R}_{>0}\right\},
	\eeq
and a particle attribute is the selection of a single representative within it.

	\begin{mydef}[Inertial mass]
	\label{def.mass}
	Let $g^\flat$ be the spacetime musical map \eqref{eq.gflat} and $\left[g^\flat\right]$ its homothety class \eqref{eq.homothety}. A \emph{particle} carries, as the first of its constitutive data, a choice of representative $G^\flat \in \left[g^\flat\right]$, its \emph{constitutive morphism}. The \emph{inertial mass} of the particle is the positive number $m$ that selects the representative,
		\beq
		\label{eq.mass}
		G^\flat = m\, g^\flat.
		\eeq
	The mass so defined is a pure ratio and carries no dimension of its own, since the dimension of mass resides in the representative $g^\flat$ against which it is measured.
	\end{mydef}

	\begin{mydef}[Momentum]
	\label{def.mom}
	Let $(M,g)$ be a differentiable manifold equipped with a bilinear structure [cf. Definition~\ref{def.metric}], and let $\gamma:I\subset\mathbb{R}\longrightarrow M$ be a curve with velocity $v\in T_p M$ at each event $p$. The \emph{momentum} of a particle of constitutive morphism $G^\flat$ [cf. Definition~\ref{def.mass}] is the constitutive dual of its velocity,
		\beq
		\label{eq.momentum}
		\pi \equiv  G^\flat(v) = m\, g^\flat (v) = m\, v^\flat.
		\eeq
	\end{mydef}

A rescaling $G^\flat \mapsto \lambda G^\flat$ passes to another representative of $\left[g^\flat\right]$, namely, a particle of mass $\lambda m$, kinematically indistinguishable from the first -- inertial motion, self-parallel curves and the uniform standard all belong to the homothety class -- yet dynamically distinct. None of this presupposes symmetry or non-degeneracy of the bilinear structure; the momentum calls upon the map \eqref{eq.gflat} alone.

A ray has no distinguished element. Nothing in the geometry singles out $g^\flat$ from among the $\lambda g^\flat$, so the number $m$ of Definition~\ref{def.mass} is not attached to the particle alone. It is the ratio of the particle's constitutive morphism to a representative chosen by decree, and choosing that representative is the act of adopting a standard of mass. The freedom is a group,
	\beq
	\label{eq.gauge}
	g^\flat \longmapsto \lambda\, g^\flat, \qquad m_i \longmapsto \lambda^{-1} m_i, \qquad \lambda \in \mathbb{R}_{>0},
	\eeq
under which the constitutive morphism $G^\flat_i = m_i\, g^\flat$ of every particle survives unchanged, so that choosing a representative within the class is a gauge fixing in the sense that nothing observable resolves what the choice fixes~\cite{weatherall2016understanding}. The theory therefore asserts the collection of morphisms themselves, and with them the ratios $m_i/m_j$; an absolute mass is not among its objects. We thus separate three offices. Kinematics sees the class alone. Dynamics sees the representatives. The choice of a fiducial representative is the business of neither, and it is all that metrology contributes.

The constant $\lambda$ in the transformation \eqref{eq.gauge} is not a technical convenience. Were we to let the rescaling vary from event to event, Proposition~\ref{prop.homothety} would fail in its converse direction, since the uniformity~\eqref{eq.uniflocal} of the rescaled structure acquires the term $v[f]\,g(v,v)$, and the inertial motions would no longer be the same motions. A standard of mass that varied from place to place would therefore bend the trajectories of free particles. In this sense a unit of mass must be one and the same everywhere as a consequence of the invariance of inertial motion, rather than by a convention adopted for convenience, and it is falsifiable by the same measurements that test the universality of free fall.

The reading extends to the representative itself. If $m$ is a ratio to a chosen element of the class, then that element is the constitutive morphism of the body which the choice declares to be of unit mass, and $g^\flat$ is not the metric of the spacetime manifold in any absolute sense. What $M$ carries intrinsically is the class $\left[g^\flat\right]$; what carries the normalisation is a body. The reading is relationalist about scale, and it declines the substantivalist alternative on which the manifold would carry a magnitude of its own~\cite{nerlich1994shape}. We record the reading rather than argue for it, since it concerns the conventions by which we measure the theory and not the geometry upon which we build it.

The momentum is conserved along the motion precisely when $L_v \pi = 0$. This is not the general case; the obstruction $L_v \pi$ measures the departure of the momentum from conservation, and Newton named that failure the \emph{force}~\cite{newton1999principia}, in the fourth of the Definitions and in the Second Law.

The derivative we use to measure that failure is not a matter of taste, and our reason for the Lie derivative is the one that has governed every choice so far. At this point of the construction the bilinear form and the connection remain independent attributes of the manifold, and the Second Law, Theorem~\ref{theo.nsl}, finally relates them. To measure the change of the momentum with $\nabla$ would be to carry the connection into the very definition of the force, and so to presuppose the coupling that the Second Law exists to establish. The Lie derivative asks only for what is already in hand, namely the flow of the motion itself.

The minimality has a price. We assumed the propagator local when we built the covariant derivative, so that $\nabla_v w$ at an event depends upon the velocity only through its value there. The Lie derivative enjoys no such property, since we define it by dragging along the flow, so that $L_v \pi$ at an event requires the velocity field in a neighbourhood of it. A force in the present sense is therefore an attribute of the family of motions filling a region and not of a single worldline in isolation. This is a further respect in which dynamics departs from kinematics, for the rectilinear condition \eqref{eq.rectilinear} and the scalar constrained by the uniformity condition \eqref{eq.uniformity} each demand the velocity only along the curve they govern, so that a motion may be inertial, or fail to be, with no reference to its neighbours. Thus, we define
	\begin{mydef}[Force]
	\label{def.force}
	Let $(M,g)$, $\gamma$, $v$ and $\pi$ be as in Definition~\ref{def.mom}. Then, the action
		\beq
		\label{eq.force}
		F(u) \equiv \left[\frac{\d}{\d t} \pi\right] (u) = \left\{L_v,\pi \right\}(u)
		\eeq
	on an arbitrary vector field $u\in \GTM$ defines a 1-form $F\in \GTsM$ called the force.	
	\end{mydef}	

Note that the definition of force does not require, or presuppose, the idea of inertial motion (cf. Definition~\ref{def.inertia}).

We single out the motions upon which the force does not act.

	\begin{mydef}[Force-free motion]
	\label{def.forcefree}
	Let $(M,g)$, $\gamma$, $v$ and $F$ be as in Definition~\ref{def.force}. The motion is \emph{force-free} if the force vanishes on every vector field,
		\beq
		\label{eq.forcefreedef}
		F(u) = 0 \qquad \text{for every} \qquad u \in \GTM.
		\eeq
	\end{mydef}

The mass now discloses its dynamical role. A force is the stimulus that alters the inertial state of a particle -- its velocity -- and since the momentum~\eqref{eq.momentum} carries the mass as a constant factor, that stimulus is $F = L_v \pi = m\, L_v\!\left(v^\flat\right)$. At a prescribed kinematic change of the velocity, the force a particle demands is proportional to its mass. The scale we introduced in Definition~\ref{def.mass} as the bare label individuating one particle among others is thus the particle's \emph{inertial response}.

\subsection{Newton's Second Law}

The idea of linking inertial motion to force-free motion leads precisely to Newton's Second Law. Let us take stock of how far kinematics has constrained the pair $(g,\nabla)$, and declare the first refinement it could not detect. From this point onwards we require the bilinear structure to be \emph{symmetric} -- an assumption invisible to the whole of Section 2. The existence of inertial motions through every event demands that $g$ be a Killing tensor of the connection -- the cyclic condition~\eqref{eq.killingcyc} of Proposition~\ref{prop.killing} -- while the transportability of standards of magnitude upgrades this demand to full compatibility -- the isometry property~\eqref{eq.isometry}. The dynamics adds the final requirement, the \emph{torsion-free} condition
	\beq
	\label{eq.torsionfree}
	\nabla_v u - \nabla_u v = \left[v,u \right],
	\eeq
for every pair of vector fields $v,u\in \GTM$. It is a remarkable fact -- the fundamental theorem of Riemannian geometry -- that once we impose these demands no freedom remains. For a metric -- a symmetric, non-degenerate $g$ -- there exists exactly one torsion-free connection compatible with $g$, the \emph{Levi-Civita} connection. The well-posedness of inertia, the transportability of magnitudes and the identification of force with the change of momentum fix the symmetric part of the transport between them; setting the antisymmetric part to zero is a further step.

The bridge between the Lie and the covariant derivative is the content of the following lemma. The statement is that of Theorem 1.17 in Chapter IV of Arnold and Khesin~\cite{arnold1998topological}, where it accounts for the form of the Euler equation of an ideal fluid, and it belongs to the standard apparatus of geometric mechanics~\cite{marsden1999introduction,khesin2009geometry}. Their proof proceeds through test fields commuting with $v$ and delegates the general case to a remark, whereas we prove it for an arbitrary test field from the outset. The lemma is in any case the specialisation to $\alpha = \iota_v g$ of the Lie--covariant compatibility of Proposition~\ref{prop.liecovgen}, which owes nothing to the metric.

	\begin{lemma}[Lie--covariant compatibility]
	\label{lem.liecov}
	Let $\nabla$ be the Levi-Civita connection of $g$, that is, the unique affine connection which is torsion-free -- cf. the exchange condition~\eqref{eq.torsionfree} -- and compatible with $g$ -- cf. the isometry property~\eqref{eq.isometry}. Then, for every vector field $v \in \GTM$,
		\beq
		\label{eq.liecov}
		L_v \left(v^\flat\right) = \left(\nabla_v v\right)^\flat + \frac{1}{2}\, \d \left[ g(v,v)\right].
		\eeq
	\end{lemma}
	\begin{proof}
	Let $w\in \GTM$ be an arbitrary vector field. The isometry property~\eqref{eq.isometry} evaluated on the triple $(v,w,v)$ gives the rate of change of the pairing along the flow,
		\beq
		\label{eq.pairingrate}
		v\left[g(w,v) \right] = g\left(\nabla_v w, v \right) + g\left(w, \nabla_v v \right).
		\eeq
	The torsion-free condition~\eqref{eq.torsionfree} exchanges the derivative in the first term at the cost of a commutator,
		\beq
		\label{eq.exchange}
		g\left(\nabla_v w, v \right) = g\left(\nabla_w v, v \right) + g\left(\left[v,w\right], v \right),
		\eeq
	while the isometry property, now evaluated on the triple $(w,v,v)$, identifies the exchanged term as half the differential of the squared magnitude,
		\beq
		\label{eq.gradspeed}
		g\left(\nabla_w v, v \right) = \frac{1}{2}\, w\left[g(v,v) \right] = \frac{1}{2}\, \d\left[g(v,v)\right](w),
		\eeq
	where the last equality is the definition of the differential of a function, $\d f(w) = w(f)$. Assembling the three relations \eqref{eq.pairingrate}--\eqref{eq.gradspeed} we obtain
		\beq
		\label{eq.rateassembled}
		v\left[g(w,v) \right] = \frac{1}{2}\, \d \left[g(v,v) \right](w) + g\left(\left[v,w\right], v\right) + g\left(w, \nabla_v v\right).
		\eeq
	On the other hand, the action \eqref{eq.dp} of the Lie derivative on the 1-form $v^\flat$ reads
		\beq
		\label{eq.lieflat}
		\left(L_v v^\flat\right) (w) = v\left[v^\flat(w) \right] - v^\flat\left(\left[v,w\right] \right) = v\left[ g(v,w)\right] - g\left(v, \left[v,w\right]\right).
		\eeq
	Substituting the assembled rate \eqref{eq.rateassembled} into the evaluation \eqref{eq.lieflat}, the two commutator terms cancel by virtue of the symmetry of $g$, leaving
		\beq
		\label{eq.liecovw}
		\left(L_v v^\flat \right)(w) = \left[\left(\nabla_v v \right)^\flat + \frac{1}{2}\, \d\left[g(v,v) \right] \right](w).
		\eeq
	Since $w$ is arbitrary, the 1-form identity~\eqref{eq.liecov} follows.
	\end{proof}

Two features of the proof are worth noting. First, at no stage do we require the covariant derivative to act on the 1-form $v^\flat$ itself. The entire argument takes place at the level of vector fields and the scalar pairing, in harmony with our construction of $\nabla$ from the parallel propagator. Second, the cancellation of the commutator terms between the exchange \eqref{eq.exchange} and the evaluation \eqref{eq.lieflat} is precisely what spares us the subtleties of the original argument. Arnold and Khesin work with \emph{comoving} test fields -- those dragged along the flow without change, $\dot w = L_v w = \left[v,w\right] = 0$, cf. the time derivative \eqref{eq.Xdot} -- for which both commutators are absent from the outset, and must then extend the identity to arbitrary $w$ by rectifying the flow around every regular point of $v$ and inspecting the zeros of $v$ separately. The comoving fields nonetheless retain their physical appeal, for they represent the measuring devices of observers carried by the flow and, for them, the compatibility \eqref{eq.liecov} states how the change of the momentum co-vector is perceived from within the motion itself. In this guise the same lemma accounts for the shape of the Euler equation of an ideal fluid [Corollary 1.19 in Chapter IV of~\cite{arnold1998topological}]. A shorter route to the same identity is available, namely to evaluate the kinematic compatibility \eqref{eq.liecovgen} on the 1-form $v^\flat$, whose three terms are then the lowered acceleration, half the differential of the squared magnitude and a torsion contribution that vanishes here.

\emph{How essential is the torsion-free condition to the compatibility~\eqref{eq.liecov}?} Not essential at all, and the identity that replaces it when we drop the condition is explicit enough to be recorded in its own right.

	\begin{prop}[Lie--covariant compatibility with torsion]
	\label{prop.liecovtorsion}
	Let $\bar\nabla$ be an affine connection compatible with $g$ -- cf. the isometry property~\eqref{eq.isometry} -- carrying an arbitrary torsion $T$ [cf. Definition~\ref{def.torsion}]. Then, for every vector field $v\in \GTM$,
		\beq
		\label{eq.liecovtorsion}
		L_v\left(v^\flat\right) = \left(\bar\nabla_v v \right)^\flat + \frac{1}{2}\, \d\left[g(v,v) \right] + g\left(T(v, \cdot\,), v\right).
		\eeq
	The compatibility \eqref{eq.liecov} of Lemma~\ref{lem.liecov} is the particular case of vanishing torsion.
	\end{prop}
	\begin{proof}
	The argument of Lemma~\ref{lem.liecov} is repeated verbatim save at a single step. The torsion-free exchange \eqref{eq.torsionfree} is no longer available, and the definition \eqref{eq.torsion} of the torsion supplies in its place $\bar\nabla_v w = \bar\nabla_w v + \left[v,w\right] + T(v,w)$, so that the exchange \eqref{eq.exchange} acquires one further term,
		\beq
		\label{eq.exchangetorsion}
		g\left(\bar\nabla_v w, v \right) = g\left(\bar\nabla_w v, v \right) + g\left(\left[v,w\right], v \right) + g\left(T(v,w), v\right).
		\eeq
	Neither the identification \eqref{eq.gradspeed} of the exchanged term with half the differential of the squared magnitude nor the evaluation \eqref{eq.lieflat} of the Lie derivative involves the connection, so the added term is carried untouched through the assembly, and the cancellation of the two commutators proceeds as before by the symmetry of $g$. In place of \eqref{eq.liecovw} we are left with $\left(L_v v^\flat\right)(w) = \left(\bar\nabla_v v\right)^\flat(w) + \tfrac{1}{2}\,\d\left[g(v,v)\right](w) + g\left(T(v,w),v\right)$, and $w$ is arbitrary.
	\end{proof}

The compatibility \eqref{eq.liecov} therefore survives some torsions and not others, and the ones it survives form a class that the polarisation of Lemma~\ref{lem.polarisation} identifies exactly.

	\begin{corollary}[Admissible torsions]
	\label{cor.antisym}
	The compatibility \eqref{eq.liecov} holds for $\bar\nabla$ if and only if
		\beq
		\label{eq.admissible}
		g\left(T(v,w),v \right) = 0 \quad \text{for every} \quad v,w \in \GTM,
		\eeq
	and the condition~\eqref{eq.admissible} holds if and only if the lowered torsion $g\left(T(u,w),z\right)$ is \emph{totally antisymmetric}.
	\end{corollary}
	\begin{proof}
	The first equivalence is immediate from the compatibility with torsion \eqref{eq.liecovtorsion}, whose last term is the only one in which the torsion appears. For the second, write $\Theta(u,w,z) \equiv g\left(T(u,w),z\right)$, antisymmetric in its first two arguments by Definition~\ref{def.torsion}. Fix $w$ and regard $B_w(u,z) \equiv \Theta(u,w,z)$ as a bilinear form. The condition~\eqref{eq.admissible} states that $B_w$ vanishes on the diagonal, so by Lemma~\ref{lem.polarisation} its symmetrisation vanishes,
		\beq
		\label{eq.thetaskew}
		\Theta(u,w,z) = -\,\Theta(z,w,u),
		\eeq
	and $\Theta$ is antisymmetric under the exchange of its first and third arguments as well as of its first two. A trilinear form antisymmetric under both exchanges is totally antisymmetric, for the two transpositions generate the whole permutation group of three letters. The converse is immediate, since a totally antisymmetric form vanishes whenever two of its arguments coincide.
	\end{proof}

Such torsion is moreover invisible to the trajectories -- it is the freedom of Proposition~\ref{prop.torsion}, whose rectilinear motions coincide with those of the Levi-Civita connection. That invisibility is narrower than it may appear, and the narrowing is exactly what leaves a torsion behind. Two connections share their geodesics as unparametrised curves whenever they are projectively equivalent, a freedom which strictly contains the antisymmetric one; they share them as affinely parametrised motions precisely when their difference tensor is antisymmetric. What excludes the wider freedom here is the Principle of Inertia itself, since the affine reparametrisations are the only ones under which an inertial motion remains inertial [Theorem~\ref{theo.inertia}], so that the parameter is pinned to within the very freedom the Principle allows. The distinction is not idle. In the axiomatic programme of Ehlers, Pirani and Schild the projective structure of free fall is combined with the conformal structure of light propagation to yield a torsion-free connection, and Wheeler has lately argued that the torsion-freeness obtained there follows from restricting the reparametrisations admitted, as an arbitrary connection affords greater freedom~\cite{wheeler2025geometry}. The two constructions agree upon the underdetermination and part company upon what it leaves behind. There the class of reparametrisations is enlarged, and the torsion the enlargement generates is of the trace type, built from a gradient and the identity, so that its torsion $1$-form does not vanish. Here the class is pinned to the affine reparametrisations by the Principle of Inertia, and the torsion which survives is the totally antisymmetric one of Corollary~\ref{cor.antisym}. The two residues are complementary, for the trace type is exactly what the admissibility~\eqref{eq.admissible} of Corollary~\ref{cor.antisym} excludes, while the totally antisymmetric part is untouched by a change of parameter. The parallel with Proposition~\ref{prop.killing} is worth noting. Each demand we have imposed erases a symmetric part of the corresponding tensor -- inertia, the symmetrised derivative \eqref{eq.killing} of the metric; the compatibility \eqref{eq.liecov}, all but the totally antisymmetric part of the torsion -- and torsion-freeness proper is the canonical choice within the residual freedom, the only one requiring no additional tensor field on the spacetime manifold.

\begin{theo}[Newton's Second Law]
\label{theo.nsl}
Let $(M,g,\nabla)$ be a Riemannian (resp. indefinite) manifold, where $\nabla$ is the Levi-Civita connection of $g$. Let $\gamma:I \subset \mathbb{R}\longrightarrow M$ be a path whose tangent vector field  $v$ describes the motion of a particle of constant mass $m$. Then
	\beq
	\label{eq.nsl}
	F - m\, a^\flat = \frac{1}{2} \d \left(\iota_v \pi \right).
	\eeq
In particular, if $\gamma$ is parametrised by its arc-length we have
	\beq
	\label{eq.nslarc}
	F = m\, a^\flat.
	\eeq

\end{theo}

\begin{proof}
Since the mass is constant, the momentum of Definition~\ref{def.mom} satisfies $L_v \pi = m\, L_v \left(v^\flat\right)$. The force \eqref{eq.force} is then given directly by the compatibility Lemma~\ref{lem.liecov},
	\beq
	\label{eq.nslproof1}
	F = L_v \pi = m \left(\nabla_v v \right)^\flat + \frac{m}{2}\, \d\left[g(v,v) \right] = m\, a^\flat + \frac{1}{2}\, \d\left(\iota_v \pi \right),
	\eeq
where we have recognised the lowered acceleration $a^\flat = g^\flat\left(\nabla_v v \right)$ and the pairing $\iota_v \pi = \pi(v) = m\, g(v,v)$. This is the identity~\eqref{eq.nsl}. Finally, when the integral curves of $v$ are parametrised by arc-length the speed is constant, $g(v,v) = 1$, hence $\iota_v \pi = m$ and the exact term in \eqref{eq.nsl} vanishes identically, leaving $F = m\, a^\flat$.
\end{proof}

The arc-length parametrisation brings the Second Law to its sharpest form, and with it the identification of free and inertial motion.

	\begin{corollary}[Force-free and inertial motion]
	\label{cor.forcefree}
	For a motion parametrised by arc length, $g(v,v) = 1$, the force is the lowered acceleration, $F = m\, a^\flat$; consequently the motion is force-free if and only if it is inertial,
		\beq
		\label{eq.forcefree}
		L_v \pi = 0 \quad\Longleftrightarrow\quad \nabla_v v = 0.
		\eeq
	\end{corollary}
	\begin{proof}
	By the Second Law \eqref{eq.nsl} the exact term $\tfrac12\,\d\left(\iota_v \pi\right)$ vanishes when the speed is constant, leaving $F = L_v \pi = m\, a^\flat$. Since $g^\flat$ is injective, the lowered acceleration $a^\flat = g^\flat\left(\nabla_v v\right)$ vanishes precisely when $\nabla_v v$ does, hence \eqref{eq.forcefree}.
	\end{proof}

Away from arc length the force and the lowered acceleration part company by the exact term, and force-free motion -- the vanishing of $L_v \pi$ -- no longer coincides with inertia.

The Second Law is thus a statement about the compatibility of the two derivatives at our disposal. The force -- the dynamical rate of change of the momentum, measured by the Lie derivative -- differs from the mass times the lowered acceleration -- the kinematical rate of change of the velocity, measured by the covariant derivative -- by an exact term which a suitable choice of clock removes. This is the precise realisation of the claim we made in the Introduction, that is, the link between force and acceleration is nothing but the relation between the Lie and the covariant derivative for a torsion-free connection compatible with the metric.

In local coordinates the relation~\eqref{eq.nslarc} is a set of $n$ second order ordinary differential equations for the curve coordinates $x^i = x^i(t)$, and it fully characterises the fundamental problem of Newtonian mechanics, that is, given the forces, determine the motion. The exact term in the identity~\eqref{eq.nsl} has a transparent kinematic meaning. Since $\iota_v \pi = m\, g(v,v)$, it vanishes precisely for those parametrisations rendering the speed constant -- the uniform motions singled out by the Principle of Inertia (Theorem~\ref{theo.inertia}) -- and, in particular, for the arc-length parametrisation of a Riemannian $(M,g)$.

The same contraction gives the term a physical reading. It is the differential $\d\left[\tfrac12\, m\, g(v,v)\right]$ of a single scalar, and in the reading of the construction in which $g$ measures spatial magnitudes and the curve parameter is the Newtonian time, that scalar is the kinetic energy of the particle. The term vanishes under arc length precisely because the kinetic energy is then constant. We record the identification in the constitutive register, since it depends upon that reading of $g$ and upon nothing established here.

We established the Second Law under the torsion-free condition~\eqref{eq.torsionfree}, and yet the compatibility on which it rests survives the relaxation of that condition in the explicit form \eqref{eq.liecovtorsion}. The law itself therefore generalises, at the cost of a single additional term assembled from the torsion and the velocity alone.

	\begin{mydef}[Torsion 1-form]
	\label{def.torsionform}
	Let $\bar\nabla$ be an affine connection on $(M,g)$ with torsion $T$ [cf. Definition~\ref{def.torsion}] and let $v$ be the velocity of a motion. The \emph{torsion 1-form} of the motion is the element $\tau_v \in \GTsM$ whose action on an arbitrary vector field $w$ is
		\beq
		\label{eq.torsionform}
		\tau_v(w) \equiv g\left(T(v,w),v \right).
		\eeq
	\end{mydef}

Since we build it from a tensor, the torsion 1-form at an event depends on the velocity there and on nothing else about the motion, and it is homogeneous of the second degree in that velocity.

	\begin{theo}[Generalised Second Law]
	\label{theo.nslgen}
	Let $(M,g)$ be a metric manifold and $\bar\nabla$ a connection compatible with $g$ -- cf. the isometry property~\eqref{eq.isometry} -- carrying an arbitrary torsion $T$. Let $\gamma:I \subset \mathbb{R}\longrightarrow M$ be a path whose tangent vector field $v$ describes the motion of a particle of constant mass $m$, with acceleration $a = \bar\nabla_v v$. Then
		\beq
		\label{eq.nslgen}
		F - m\, a^\flat = \frac{1}{2} \d \left(\iota_v \pi \right) + m\, \tau_v.
		\eeq
	In particular, if $\gamma$ is parametrised by its arc-length we have
		\beq
		\label{eq.nslgenarc}
		F = m\left(a^\flat + \tau_v \right).
		\eeq
	\end{theo}
	\begin{proof}
	Since the mass is constant, the momentum of Definition~\ref{def.mom} satisfies $L_v \pi = m\, L_v\left(v^\flat\right)$, and the compatibility \eqref{eq.liecovtorsion} delivers the force directly,
		\beq
		\label{eq.nslgenproof}
		F = L_v \pi = m \left(\bar\nabla_v v \right)^\flat + \frac{m}{2}\, \d\left[g(v,v) \right] + m\, g\left(T(v,\cdot\,),v\right) = m\, a^\flat + \frac{1}{2}\, \d\left(\iota_v \pi \right) + m\, \tau_v,
		\eeq
	where we have recognised the lowered acceleration $a^\flat = g^\flat\left(\bar\nabla_v v\right)$, the pairing $\iota_v \pi = \pi(v) = m\, g(v,v)$ and the torsion 1-form \eqref{eq.torsionform}. This is the identity~\eqref{eq.nslgen}. When the integral curves of $v$ are parametrised by arc-length the speed is constant, $g(v,v) = 1$, hence $\iota_v \pi = m$ and the exact term vanishes identically.
	\end{proof}

The identity~\eqref{eq.nslgen} contains the Second Law \eqref{eq.nsl} as the particular case $T=0$. The torsion enters it through a single term, and that term is of a different nature from the exact one beside it. The exact term disappears with a suitable choice of clock, whereas the torsion 1-form is tensorial in the velocity and survives every parametrisation, arc-length included. This term, and this term alone, therefore decides whether the dynamical and the kinematic notions of a free particle still agree.

	\begin{corollary}[Force-free and inertial motion with torsion]
	\label{cor.forcefreegen}
	For motions parametrised by arc length, the force-free motions of $(M,g,\bar\nabla)$ coincide with its inertial motions if and only if the torsion 1-form \eqref{eq.torsionform} vanishes for every velocity, that is, if and only if $\bar\nabla$ satisfies the admissibility condition~\eqref{eq.admissible} of Corollary~\ref{cor.antisym}.
	\end{corollary}
	\begin{proof}
	Suppose the admissibility \eqref{eq.admissible} holds. The generalised law \eqref{eq.nslgenarc} then reduces to $F = m\, a^\flat$ and, since the musical map $g^\flat$ is injective, $L_v \pi = 0$ if and only if $\bar\nabla_v v = 0$, which is the statement of Corollary~\ref{cor.forcefree} for the connection $\bar\nabla$.

	Conversely, suppose the two notions coincide and fix an event $p\in M$ together with a direction at $p$. Compatibility implies the Killing condition~\eqref{eq.killing}, so through that event and direction there passes an inertial motion [cf. Proposition~\ref{prop.killing}], unique up to the affine freedom of Theorem~\ref{theo.inertia} and therefore normalisable to unit speed. Being inertial it is force-free by hypothesis, and the generalised law \eqref{eq.nslgenarc} leaves $\tau_v\vert_p = 0$. Since $\tau_v$ is homogeneous of the second degree in the velocity, it vanishes along every non-null direction at $p$. It is quadratic in $v\vert_p$, hence continuous, and the null cone is the zero set of a non-trivial quadratic form and so has dense complement, whence $\tau_v$ vanishes for every $v\vert_p \in T_p M$; as the event was arbitrary, the admissibility \eqref{eq.admissible} follows.
	\end{proof}

Corollary~\ref{cor.forcefreegen} recasts the torsion-free assumption as a physical demand instead of a technical convenience. Newton's Second Law does not require the torsion to vanish. It requires only the condition~\eqref{eq.admissible}, and that condition is precisely the one under which a particle subject to no force moves inertially. What the demand erases is the part of the torsion that would drive a free particle off its inertial trajectory; what it leaves untouched is, by Corollary~\ref{cor.antisym}, a totally antisymmetric object.

The condition is not new as geometry, and the credit for it belongs elsewhere. A metric connection preserves the geodesics of its metric precisely when its torsion is a $3$-form, which is Corollary 2.1 of Agricola and Friedrich~\cite{agricola2004holonomy} and is stated as a characterisation at survey level~\cite{agricola2006srni}, proved there for its own sake and accompanied by the explicit form $\nabla_X Y = \nabla^g_X Y + \frac{1}{2}T(X,Y,\cdot\,)$ of the connections concerned, the same family that Proposition~\ref{prop.theta} reaches from the opposite direction and the same one whose uniqueness and structure are worked out in the string-theoretic literature~\cite{friedrich2002parallel}. The same condition organises the metric-affine literature~\cite{hehl1995metric,hehl2007cartan}. Corollary~\ref{cor.antisym} thus contributes where the condition comes from rather than the geometry itself, for the condition arrives here from the Second Law together with the demand that force-free and inertial motion agree, and with no field equation and no gravitational coupling in sight. Einstein--Cartan theory reaches a torsion by the opposite route, making the spin of matter its source through a field equation~\cite{sciama1964physical,kibble1961lorentz,hehl1976general}. That equation is algebraic, fixing the torsion pointwise as a multiple of the spin density, so that the torsion there is a field upon the spacetime manifold which does not propagate and vanishes wherever the matter does. The torsion of this construction is a constitutive datum carried by the particle, and we obtain the spin from it rather than the reverse. The geometry has been available since Cartan~\cite{cartan1923varietes}, and that three derivations so differently motivated should converge upon the same algebraic condition is a check upon this one instead of a borrowing from them. In the following subsection we take up the antisymmetric residue and find a second constitutive relation in it.

\subsection{The second constitutive relation}
\label{sec.spin}

Inertial-structure equivalence (Theorem~\ref{theo.ise}) exhibited two data of the pair $(g,\nabla)$ that no free trajectory resolves, the scale of the bilinear form and the torsion of the connection. Dynamics has already disposed of the first. The constitutive morphism of Definition~\ref{def.mass} selected the scale and disclosed it as the inertial mass, and momentum and force descended from that single choice. We have so far set the torsion to zero. We now ask what dynamics makes of it, and find that the Second Law has answered already.

Corollary~\ref{cor.forcefreegen} has already fixed which torsions the dynamics tolerates, namely those whose torsion 1-form vanishes, and Corollary~\ref{cor.antisym} identified them as the connections whose lowered torsion $g\left(T(u,w),z\right)$ is totally antisymmetric. A totally antisymmetric trilinear form is a $3$-form, and we record that observation as the second constitutive datum a particle may carry.

	\begin{mydef}[Constitutive torsion]
	\label{def.consttorsion}
	Let $(M,g)$ be a metric manifold and $\nabla$ the Levi-Civita connection of $g$. An affine connection $\bar\nabla$ sharing the symmetric part of $\nabla$ [cf. Proposition~\ref{prop.torsion}] is \emph{admissible} if its torsion 1-form vanishes, that is, if it satisfies the condition~\eqref{eq.admissible} of Corollary~\ref{cor.forcefreegen}, and the $3$-form
		\beq
		\label{eq.theta}
		\Theta\left(u,w,z\right) \equiv g\left(T(u,w),z \right)
		\eeq
	is its lowered torsion. The \emph{constitutive torsion} of a particle is a $3$-form of this kind carried along its worldline, that is, a smooth assignment of an element of $\Lambda^3 T^*_{\gamma(t)}M$ to each event of the motion.
	\end{mydef}

Definition~\ref{def.consttorsion} is a constitutive choice of exactly the kind that Definition~\ref{def.mass} made for the bilinear form, and we intend it as nothing else. It selects a representative within a freedom that the description of motion has left open, and it does so particle by particle. It is emphatically not the statement that a spin generates a torsion. The familiar proportionality between a torsion and a spin density, which a reader arriving from Einstein--Cartan theory will expect, has no counterpart here. We assert no such relation and introduce no coupling constant of any strength.

	\begin{prop}[The residual freedom is a $3$-form]
	\label{prop.theta}
	The assignment $T \longmapsto \Theta$ of \eqref{eq.theta} is a bijection between the torsions of admissible connections and the $3$-forms on $M$.
	\end{prop}
	\begin{proof}
	That $\Theta$ is totally antisymmetric is the content of Corollary~\ref{cor.antisym}, the admissibility \eqref{eq.admissible} being precisely its hypothesis. Conversely, given a $3$-form $\Theta$ the non-degeneracy of $g$ determines a unique $T$ through the definition \eqref{eq.theta}, namely $T(u,w) = g_\sharp\left[\Theta(u,w,\cdot\,)\right]$ with $g_\sharp$ the inverse musical map \eqref{eq.gsharp}; it is antisymmetric in $(u,w)$ and satisfies the admissibility \eqref{eq.admissible}, since $g\left(T(v,w),v\right) = \Theta(v,w,v) = 0$. That this $T$ is the torsion of an admissible connection, and not merely a tensor of the right symmetry type, is settled by exhibiting one, the coefficients
		\beq
		\label{eq.barconnection}
		\bar\Gamma^k_{ij} = \Gamma^k_{(ij)} + \tfrac{1}{2}\, T^k_{ij}
		\eeq
	defining a connection whose symmetric part is that of $\nabla$ and whose torsion, by the coordinate expression of Definition~\ref{def.torsion}, is $T$. That connection is compatible with $g$, and the total antisymmetry of $\Theta$ is exactly what makes it so. The added term contributes to the derivative of the bilinear structure the two lowered pieces $-\tfrac{1}{2}\,\Theta_{ijk}$ and $-\tfrac{1}{2}\,\Theta_{ikj}$, which cancel one another, so the isometry property~\eqref{eq.isometry} survives the addition and the connection is admissible. The two assignments are mutually inverse by construction.
	\end{proof}

Kinematics fixes the connection only up to its symmetric part, and Proposition~\ref{prop.theta} identifies the residual freedom as a $3$-form. A particle is the selection of a representative within that freedom, and contracting the velocity into the object so selected is what produced the momentum from $G^\flat$.

The two constitutive relations part company here, and the difference is structural, as opposed to a matter of taste. The homothety class \eqref{eq.homothety} is a ray upon which $\mathbb{R}_{>0}$ acts simply transitively, so it possesses no distinguished element and the label $m$ that selects one is a pure ratio. The residual freedom in the connection is a vector space, which possesses an origin. Its label therefore ranges over the whole of $\mathbb{R}$ and not merely over its positive part, so that the second constitutive datum admits both a sign and an absolute zero, neither of which the first admits. Torsion-freeness is that zero, and it is canonical twice over, being the origin of the linear structure and the only choice demanding no further tensor upon $M$.

That the residual freedom is a vector space rather than a ray decides where the datum must live. A $3$-form fixed once and for all upon $M$ would deliver, at an event and for a given velocity, a spin determined by that velocity and a single number. The following proposition shows that a spin obeying the supplementary condition has more freedom than that, and the discrepancy is why the constitutive torsion must follow the path instead of furnishing the manifold.

	\begin{prop}[Freedom of the spin]
	\label{prop.spindof}
	Let $p\in M$, let $v\vert_p$ be non-null and let $\iota_v : \Lambda^3 T^*_p M \longrightarrow \Lambda^2 T^*_p M$ be the contraction. Its image is exactly the space of $2$-forms annihilating the velocity and its kernel has dimension $\binom{n-1}{3}$, so that
		\beq
		\label{eq.spindof}
		\dim\left\{S \in \Lambda^2 T^*_p M \; : \; \iota_v S = 0 \right\} = \binom{n-1}{2}.
		\eeq
	In four dimensions the constitutive torsion has four components at each event, the spin three, and the contraction one.
	\end{prop}
	\begin{proof}
	Decompose $T_p M = \mathbb{R}v \oplus v^\perp$, which the non-degeneracy of $g$ and $g(v,v)\neq 0$ permit. A $2$-form annihilating $v$ is determined by its restriction to $v^\perp$, hence the dimension \eqref{eq.spindof}. The contraction kills those $3$-forms whose arguments may all be taken in $v^\perp$, a space of dimension $\binom{n-1}{3}$, and $\binom{n}{3} - \binom{n-1}{3} = \binom{n-1}{2}$ exhausts the target, so the image is the whole of it.
	\end{proof}

A spin may therefore point anywhere in the orthogonal complement of the velocity, independently of where the particle is going, and this is what a spin must do. A constitutive torsion fixed upon the manifold could not supply it, since one number and a velocity would determine the answer. A constitutive torsion carried along the worldline supplies exactly it, its $\binom{n}{3}$ components at each event yielding the $\binom{n-1}{2}$ of the spin with $\binom{n-1}{3}$ to spare. The second constitutive datum is thus an attribute following the path, where the first was a scale selected once.

The counting is not the only argument against furnishing the manifold, and the second is the more telling. A $3$-form fixed upon $M$ would have to be transportable if it were to serve as a standard, and the natural demand is that it be parallel. Riemannian manifolds admitting a metric connection with non-vanishing parallel skew-symmetric torsion are a classified and very special family, the naturally reductive homogeneous spaces, the nearly K\"ahler and nearly parallel $G_2$ manifolds, the Sasakian and $3$-Sasakian manifolds and the twistor spaces over positive quaternion-K\"ahler manifolds among them~\cite{cleyton2021metric}. To furnish the manifold with the second constitutive datum would therefore be to legislate the geometry of the world in order to grant one particle a spin. To let the particle carry it is to legislate nothing at all.

	\begin{mydef}[Spin]
	\label{def.spin}
	Let $\gamma$ be a motion with velocity $v$ and let $\Theta$ be the constitutive torsion \eqref{eq.theta} of the particle traversing it. The \emph{spin} of the particle is the $2$-form
		\beq
		\label{eq.spin}
		S \equiv \iota_v \Theta .
		\eeq
	\end{mydef}

	\begin{prop}[Supplementary condition]
	\label{prop.frenkel}
	Let $\Theta$ be a $3$-form, let $v$ be the velocity of a motion and set $S = \iota_v\Theta$. Then $S$ annihilates that velocity,
		\beq
		\label{eq.frenkel}
		\iota_v S = 0.
		\eeq
	Equivalently, when $\Theta$ is the lowered torsion of an admissible connection, the endomorphism $A_v(u) \equiv T(v,u)$ is skew with respect to $g$ and has the velocity in its kernel.
	\end{prop}
	\begin{proof}
	Directly from the interior product of Definition~\ref{def.interior}, $\iota_v S = \iota_v \iota_v \Theta$, whose value on a vector field $u$ is $\Theta(v,v,u) = 0$ by the antisymmetry of $\Theta$ in its first two arguments. For the second statement, the pairing $g\left(A_v u, w\right) = \Theta(v,u,w) = S(u,w)$ is antisymmetric in $u$ and $w$, hence $g\left(A_v u, w \right) = - g\left(u, A_v w\right)$; and $g\left(A_v v, w\right) = \Theta(v,v,w) = 0$ for every $w$, hence $A_v v = 0$ by the non-degeneracy of $g$.
	\end{proof}

The spin is thus the lowered form of an infinitesimal rotation carried along the worldline, and we do not impose the condition~\eqref{eq.frenkel} upon it, for it follows from the vanishing of the twofold interior product. That the velocity should lie in the kernel of a particle's own rotation is the geometric content of the supplementary conditions that classical spinning-particle mechanics must postulate.

Neither the algebra of Proposition~\ref{prop.frenkel} nor the conclusion drawn from it is new, and we concede the point at once and precisely. Relativistic spin hydrodynamics constructs the local spin density by contracting a totally antisymmetric spin tensor with the fluid velocity, and derives the condition~\eqref{eq.frenkel} from the twofold contraction by the argument used above~\cite{fang2026relativistic}, a construction prepared by earlier treatments of antisymmetric spin tensors~\cite{hongo2021relativistic,cao2022gyrohydrodynamics}. The spin current of a Dirac field is totally antisymmetric for a related reason, since it is the only irreducible part of a torsion which spin-half matter can excite, so that the Cartan equation of Einstein--Cartan theory ties a totally antisymmetric torsion algebraically to it~\cite{hehl1976general,hehl2007cartan}. A $3$-form in the company of a spin is thus thoroughly familiar, and we do not claim the identity~\eqref{eq.frenkel} as a discovery.

What we propose is the assignment rather than the algebra. In every one of those treatments the totally antisymmetric object is a property of matter, a spin current assembled from a field, and the geometry either receives it as a source or furnishes a background against which it moves. Here the $3$-form is none of these. It is the residual freedom in the connection which inertial motion cannot resolve [Theorem~\ref{theo.ise}], narrowed by the Second Law to precisely the totally antisymmetric case [Corollary~\ref{cor.antisym}], and then adopted as a constitutive datum of the particle. To the best of our knowledge that identification, of the geometric residue with the object whose contraction against the velocity is the spin, has not been proposed elsewhere. That the contraction should satisfy the supplementary condition identically is then a corollary of an algebra others have already exhibited, and we take the agreement of the two accounts as corroboration, as opposed to a claim. The condition~\eqref{eq.frenkel} which follows from the identification has a name and a long history, both worth recording, since what is new here is the status of the condition rather than the condition itself. Contracting the spin with the velocity of the body is the \emph{Frenkel--Mathisson--Pirani} condition, one of a family among which the relativistic mechanics of spinning particles must choose, and the Corinaldesi--Papapetrou and Tulczyjew--Dixon conditions are the principal alternatives; each choice selects a different representative worldline within the body, and none of them is forced by the geometry~\cite{costa2015center}. In the constitutive theory of a spinning fluid the same condition is listed frankly among the postulates~\cite{herrmann2000constitutive}. Here there is nothing to choose and nothing to postulate, since the spin is the contraction~\eqref{eq.spin} of a $3$-form with the very velocity against which we test the condition, so that \eqref{eq.frenkel} is an identity of the interior product and could not fail.

The degeneracy which attends the condition in its relativistic setting is absent here, and for the same reason. There the condition does not determine the worldline upon which one imposes it, since an entire disc of centroids satisfies it and generates the helical motions~\cite{costa2015center}, because a centroid is a moment of an extended distribution of matter~\cite{dixon1970dynamics} and every observer computes a different one, which is why a supplementary condition has to be imposed there at all~\cite{steinhoff2011canonical}. In the present construction the worldline is the primitive from which we build everything else and the spin is an attribute carried along it, so there is no centroid in question and nothing for the condition to fail to determine.

To ask how the spin changes along the motion we require the Lie derivative of a $2$-form. We adopt as its definition the Leibniz rule with respect to the contraction, in the same spirit in which we adopted Cartan's identity~\eqref{eq.momch} for a $1$-form,
	\beq
	\label{eq.lie2form}
	\left(L_v S\right)(u,w) = v\left[S(u,w) \right] - S\left(\left[v,u\right],w \right) - S\left(u,\left[v,w\right] \right),
	\eeq
which reduces to the commutator structure~\eqref{eq.dp} when we suppress one of the two arguments. The failure of the spin to be conserved along the flow is then the exact analogue of the force.

	\begin{mydef}[Torque]
	\label{def.torque}
	Let $S$ be the spin \eqref{eq.spin} of a particle along a motion with velocity $v$. The \emph{torque} is the $2$-form
		\beq
		\label{eq.torque}
		N \equiv \frac{\d}{\d t} S = L_v S .
		\eeq
	The motion is \emph{torque-free} if $N = 0$.
	\end{mydef}

\emph{How is the spin carried along the motion?} For the momentum a hypothesis settled the question, that of a constant mass, which in the present language reads $\nabla_v G^\flat = 0$, so that the constitutive morphism is parallel along the worldline. The constitutive torsion, following the path as opposed to furnishing the manifold, requires a hypothesis of the same kind, and a demand the first relation was too poor to make distinguishes the two candidates. A magnitude is available here, and it ought not to change. The transport that answers to that demand is not new. Fermi introduced it and Walker extended it~\cite{fermi1922sopra,walker1932relative}, it is standard in the relativistic literature~\cite{synge1960relativity}, and we postulate it here for the constitutive $3$-form rather than derive it, as one usually does, for a relativistic spin already defined.

	\begin{mydef}[Fermi--Walker transport]
	\label{def.fw}
	Let $\gamma$ be a motion of unit speed, $g(v,v) = \epsilon = \pm 1$, with acceleration $a = \nabla_v v$. The \emph{Fermi--Walker} derivative along $\gamma$ acts on a vector field by
		\beq
		\label{eq.fw}
		\D_{\rm FW} X \equiv \nabla_v X - \Omega(X), \qquad \Omega(X) \equiv \epsilon\left[g(v,X)\, a - g(a,X)\, v \right],
		\eeq
	and upon a covariant tensor as a derivation. A field with $\D_{\rm FW} = 0$ is said to be Fermi--Walker transported.
	\end{mydef}

The operator $\Omega$ is skew with respect to $g$, so Fermi--Walker transport preserves every pairing, and it carries the velocity into itself, $\D_{\rm FW}v = \nabla_v v - \Omega(v) = a - a = 0$. Along an inertial motion the acceleration vanishes with $\Omega$ and the derivative reduces to $\nabla_v$. These three properties characterise it, and they are exactly what a transport of the spin must possess.

Before adopting it we should say what kind of transport it is, for Definition~\ref{def.propagator} answers the question and the answer occupies a slot that definition left empty. There we required a propagator to be a family of bijections between tangent spaces, and we called it \emph{linear} if it respected the linear structure, \emph{parametrisation independent} if it belonged to the path rather than to the clock, and \emph{parallel} if it was both. Nothing in the manuscript has yet been linear without being parallel.

	\begin{prop}[The Fermi--Walker propagator]
	\label{prop.fwprop}
	The solutions of $\D_{\rm FW}X = 0$ along a motion of unit speed define a propagator $P^{\rm FW}_\gamma\left(t_0,t\right)$ in the sense of Definition~\ref{def.propagator}. It is linear and it is an isometry, but it is not parametrisation independent, and it is therefore not parallel. It differs from the parallel propagator by the rotation generated by $\Omega$, which annihilates every vector orthogonal to both the velocity and the acceleration and preserves the plane they span. It is moreover the unique linear isometric propagator which carries the velocity into the velocity and whose generator annihilates that orthogonal complement, at every event where the acceleration is non-null or vanishes.
	\end{prop}
	\begin{proof}
	The condition $\D_{\rm FW}X = 0$ reads $\nabla_v X = \Omega(X)$, a linear ordinary differential equation along $\gamma$, whose solutions exist, are unique from each initial value and depend linearly upon it; the composition and inverse properties of Definition~\ref{def.propagator} are those of its solution operator, and the linearity \eqref{eq.linprop} is that of the equation. That the transport is an isometry follows from the skewness of $\Omega$, since along a transported pair the compatibility \eqref{eq.isometry} gives $v\left[g(X,Y)\right] = g\left(\Omega X, Y\right) + g\left(X,\Omega Y\right) = 0$. Parametrisation independence fails because $\Omega$ is assembled from the acceleration, and under $v \mapsto f v$ the acceleration reads $\nabla_{fv}(f v) = f\, v[f]\, v + f^2\, \nabla_v v$, so that the generator, and with it the transport, is altered; the condition~\eqref{eq.reparinv} is not met.

	For the decomposition, $\Omega(X) = \epsilon\left[g(v,X)a - g(a,X)v\right]$ vanishes whenever $X$ is orthogonal to both $v$ and $a$ and carries the plane $\mathrm{span}\{v,a\}$ into itself, so it generates a rotation in that plane and the solution of $\nabla_v X = \Omega(X)$ is the parallel transported field acted upon by the flow of $\Omega$. As for uniqueness, let $\nabla_v X = A(X)$ define another such propagator. Isometry makes $A$ skew, and a skew endomorphism annihilating $\mathrm{span}\{v,a\}^\perp$ has its image in $\mathrm{span}\{v,a\}$. Carrying the velocity into the velocity forces $A(v) = a$. Writing $A(a) = \alpha v + \beta a$, the unit speed gives $g(v,a) = 0$ and skewness gives $g(A(a),a) = 0$, which together leave $\beta\, g(a,a) = 0$. Where the acceleration is non-null this forces $\beta = 0$, while $g(A(a),v) = -g(a,A(v)) = -g(a,a)$ gives $\alpha = -\epsilon\, g(a,a)$. These are the values taken by $\Omega$, so $A = \Omega$. Where the acceleration vanishes, $A(v) = 0$ and skewness carries the image of $A$ into the line spanned by the velocity while making it orthogonal to that line, so $A$ vanishes with $\Omega$.
	\end{proof}

The failure recorded in Proposition~\ref{prop.fwprop} is not an accident of the example, and it costs the Fermi--Walker propagator its connection. Lemma~\ref{lem.chain} established the homogeneity $\nabla_{(fv)}w = f\,\nabla_v w$, and its proof used the parametrisation independence \eqref{eq.reparinv} at the very step where we identified the transports along $\gamma$ and its reparametrisation. Withdraw that attribute and homogeneity goes with it, and with homogeneity goes Definition~\ref{def.connection}. The Fermi--Walker derivative is thus a connection along the curve rather than an affine connection upon $M$, knowing how to transport vectors upon the worldline that generated it and nowhere else.

The distinction has a familiar physical reading. A gyroscope carried by an inertial observer keeps its axis fixed, and the transport describing it is the parallel propagator. A gyroscope carried by an accelerated observer cannot do so, since its axis must stay orthogonal to a velocity which is itself turning, and the transport describing that is the Fermi--Walker propagator, which is the standard transport law of a non-rotating gyroscope along an accelerated worldline~\cite{misner1973gravitation}. The rotation by which the two differ is the precession the accelerated observer measures, accumulated around a closed accelerated worldline in flat spacetime as the precession of Thomas. A spin is a gyroscope, and the particle carries it instead of the manifold, so it answers to the propagator belonging to its worldline rather than to the one belonging to the spacetime.

That the two constitutive objects should answer to different transports is less of a departure than it appears. Fermi--Walker transport preserves the metric, so $\D_{\rm FW}G^\flat = v[m]\, g^\flat = \nabla_v G^\flat$, and upon the first constitutive object the two transports agree exactly; the constancy of the mass is equally a statement of parallelism and of Fermi--Walker transport, and we never had to choose between them there. Only the second constitutive object, carrying a direction as well as a magnitude, can tell them apart, and that ability is why we must choose.

One feature of the Fermi--Walker propagator is worth recalling against the assumptions taken in the course of the kinematics. Its generator \eqref{eq.fw} is assembled from the bilinear form, from the velocity of the motion and from the acceleration which the connection assigns to that velocity, and every one of these was in hand already. A second connection upon $M$ would have been a second commitment about the world; the Fermi--Walker propagator commits us to nothing beyond the observer who carries it. It is observer dependent and structure free, and the two are not the same thing.

We may take stock of the construction as a whole upon the same principle. The spacetime manifold carries exactly what the kinematics declared, a propagator and a bilinear form, refined by the demands the dynamics has since imposed upon them and by nothing else. A particle carries two data and no more, a representative of the homothety class \eqref{eq.homothety} and a constitutive torsion along its worldline. Everything else -- the momentum, the force, the spin, the torque, and the very transport by which the spin is carried -- we manufacture out of those, and the manufacture calls for no further furniture. The minimality with which the section opened is therefore intact at its close, which is the only warrant a construction of this kind can give for the objects it introduces.

	\begin{mydef}[Transport of the constitutive torsion]
	\label{def.spinlaw}
	The constitutive torsion of a particle is Fermi--Walker transported along its worldline,
		\beq
		\label{eq.spinlaw}
		\D_{\rm FW}\Theta = 0 .
		\eeq
	\end{mydef}

The law~\eqref{eq.spinlaw} is a constitutive postulate and we label it as such. Parallel transport, $\nabla_v\Theta = 0$, is the other candidate and agrees with it upon every inertial motion; we reject it because it fails to preserve the magnitude of the spin under acceleration, as the discussion at the close of this section records.

The postulate is not an eccentric one, and the relativistic mechanics of spinning bodies had reached it long before. Under the Frenkel--Mathisson--Pirani condition, and in the absence of torque, the spin vector of a gyroscope is Fermi--Walker transported along its worldline~\cite{costa2015center}, so the pairing of the velocity condition with this transport is precisely the pairing that theory arrives at from its own balance laws. A third direction reaches it as well. A model of a spinning particle in Brans--Dicke spacetime, built upon a metric-compatible connection with torsion, yields a spin which is Fermi--Walker transported by the Levi-Civita connection along the particle's worldline~\cite{burton2008spinning}. That model is instructive here for the way it differs. Its torsion is the trace part, $T^a = e^a\wedge\d\varphi/\varphi$ for a scalar $\varphi$, whose torsion $1$-form does not vanish, so it is exactly the class which the admissibility \eqref{eq.admissible} excludes; and consistently with that exclusion, the free trajectories of the model do depart from the geodesics of its metric. The same transport law thus appears in the torsion class the Second Law forbids and in the one it spares, which is some evidence that the transport is answering to the accelerated worldline as opposed to the torsion carried along it.

	\begin{prop}[Transport of the spin]
	\label{prop.spintransport}
	Let $\Theta$ be a $3$-form defined along a motion of velocity $v$ with acceleration $a = \nabla_v v$, and set $S = \iota_v\Theta$. Then
		\beq
		\label{eq.spintransport}
		\nabla_v S = \iota_a \Theta + \iota_v \left(\nabla_v \Theta\right).
		\eeq
	\end{prop}
	\begin{proof}
	Evaluate the rule~\eqref{eq.nablatensor} on the $3$-form $\Theta$ with the velocity in its first slot,
		\beq
		\label{eq.spintransportproof}
		\left(\nabla_v\Theta\right)(v,u,w) = v\left[\Theta(v,u,w)\right] - \Theta\left(\nabla_v v,u,w\right) - \Theta\left(v,\nabla_v u,w\right) - \Theta\left(v,u,\nabla_v w\right),
		\eeq
	and on the $2$-form $S$, whose value is $S(u,w) = \Theta(v,u,w)$, to obtain $\left(\nabla_v S\right)(u,w) = v\left[\Theta(v,u,w)\right] - \Theta\left(v,\nabla_v u,w\right) - \Theta\left(v,u,\nabla_v w\right)$. The two differ by the single term $\Theta\left(\nabla_v v,u,w\right) = \left(\iota_a\Theta\right)(u,w)$, which is the transport identity~\eqref{eq.spintransport}.
	\end{proof}

The covariant rate of change of the spin is therefore the acceleration contracted into the constitutive torsion, corrected by whatever the transport law does to that torsion along the motion. We obtain the torque on comparing this with the Lie derivative, which is the whole content of the kinematic compatibility one degree higher.

	\begin{theo}[Torque equation of motion]
	\label{theo.torque}
	Let $\Theta$ be a $3$-form defined along a motion of velocity $v$, and set $S = \iota_v\Theta$ and $N = L_v S$. Then
		\beq
		\label{eq.torqueeq}
		N - \iota_a \Theta = \iota_v\left(\nabla_v\Theta\right) + S\left(\nabla_{\cdot\,} v + T(v,\cdot\,),\, \cdot\,\right) + S\left(\,\cdot\,, \nabla_{\cdot\,} v + T(v,\cdot\,)\right).
		\eeq
	\end{theo}
	\begin{proof}
	The compatibility \eqref{eq.liecovk} at $k=2$, applied to the $2$-form $S$, gives $\left(L_v S\right)(u,w) = \left(\nabla_v S\right)(u,w) + S\left(\nabla_u v + T(v,u),w\right) + S\left(u,\nabla_w v + T(v,w)\right)$. Substituting the transport~\eqref{eq.spintransport} of $S$ for the first term on the right and rearranging yields \eqref{eq.torqueeq}.
	\end{proof}

The identity~\eqref{eq.torqueeq} stands to the spin as the Second Law \eqref{eq.nslgen} stands to the momentum, and the correspondence is term by term. The torque takes the place of the force, the contraction $\iota_a\Theta$ of the acceleration into the constitutive object takes the place of $m\,a^\flat$, and the remaining terms are built from the covariant derivative of the velocity field and from the torsion, exactly as the exact term and the torsion 1-form were. There is one difference of substance. A choice of clock removes the exact term of the Second Law, whereas nothing in the parametrisation removes the terms on the right of \eqref{eq.torqueeq}, which are properties of the congruence carrying the motion.

	\begin{corollary}[Free motion carries its spin unchanged]
	\label{cor.spinparallel}
	Under the transport \eqref{eq.spinlaw} the spin is itself Fermi--Walker transported, $\D_{\rm FW}S = 0$. Consequently its magnitude and the supplementary condition~\eqref{eq.frenkel} are preserved along every motion, and along an inertial motion the spin is propagated parallel to the velocity,
		\beq
		\label{eq.spinfree}
		\nabla_v v = 0 \quad \Longrightarrow \quad \nabla_v S = 0 .
		\eeq
	\end{corollary}
	\begin{proof}
	The velocity is Fermi--Walker transported and so, by hypothesis, is the constitutive torsion; the derivative acting as a derivation, the contraction $S = \iota_v\Theta$ is transported with them. Fermi--Walker transport preserves every pairing, hence the magnitude of $S$, and it preserves $\iota_v S$, which vanishes identically in any case by Proposition~\ref{prop.frenkel}. Along an inertial motion the acceleration vanishes, the operator $\Omega$ of \eqref{eq.fw} with it, and $\D_{\rm FW}$ reduces to $\nabla_v$.
	\end{proof}

Corollary~\ref{cor.spinparallel} is the counterpart, for the second constitutive relation, of the identification of force-free with inertial motion. A particle left to itself does not merely travel a self-parallel worldline; it carries its spin along that worldline by the same parallel transport, and the two statements are one, since both are the vanishing of $\nabla_v$ upon the objects the particle carries. Away from inertial motion the two part company, since the momentum answers to $\nabla_v$ and the spin to $\D_{\rm FW}$, and the difference between the two derivatives is the operator $\Omega$ of \eqref{eq.fw}, built from the acceleration alone.

We may now summarise the construction compactly, and with it the constitutive data of a particle are complete. A particle is a worldline together with two constitutive choices, one within each of the two classes that inertial-structure equivalence declared invisible to free motion. The particle has no extension and yet, in this precise sense, it has an internal structure, since the second datum supplies the orientation which the first could not. We choose the bilinear form within its homothety class and the connection within the class of connections sharing a symmetric part, and contracting the velocity into each chosen object yields the two dual observables of the theory.

The momentum is the contraction of the velocity into $G^\flat$ and the spin the contraction of the velocity into $\Theta$; the force is the failure of the first to be conserved along the flow and the torque the failure of the second. Two constitutive objects, two dual observables, two failures, and one mechanism common to both.

We must now keep two registers firmly apart. Everything stated above is geometry, proved from the metric and the admissible connection alone, save the transport \eqref{eq.spinlaw} of the constitutive torsion, which is a postulate and was declared as one where it was made. What follows is constitutive throughout, and we declare it postulate and conjecture rather than consequence.

We have not supplied the scale of the second constitutive object. The mass of Definition~\ref{def.mass} is a pure ratio, since the dimension of mass resides in the representative $g^\flat$ against which we measure it, as the requirement that $\pi = \iota_v G^\flat$ carry units of momentum demands. The same requirement applied to $S = \iota_v \Theta$ fixes the dimension of $\Theta$, for if the spin is to carry units of action then $\Theta$ carries units of action divided by velocity, that is, of mass times length.

Granted the identification, inertial-structure equivalence acquires its full physical reading. Free motion is independent of the mass of the body that executes it, and independent of its spin. The first half becomes the universality of free fall once a gravitational coupling exists. The second is guaranteed by Theorem~\ref{theo.ise} rather than predicted, and a spin-polarised torsion balance~\cite{heckel2008preferred}, which compares the free fall of a polarised body with that of an unpolarised one, therefore tests the construction only in the sense that a departure would refute it. The Mathisson--Papapetrou force~\cite{mathisson1937neue,papapetrou1951spinning} is the standard prediction for such a departure of a spinning body from the geodesic, and the totally antisymmetric case at issue here is the one surveyed in~\cite{hehl2007cartan}. Data of that class have moreover been turned into direct bounds on the components of the torsion itself~\cite{kostelecky2008torsion}. Those bounds constrain a field upon the manifold, whereas $\Theta$ is carried along a worldline, so they reach the constitutive object of Definition~\ref{def.consttorsion} only through the transport law. A second route is interferometric rather than mechanical. A flux of torsion through a region registers as a topological phase between two coherent beams which enclose it, and that again asks for a field upon the manifold, whereas the phase of an accelerated beam is governed by Fermi--Walker transport~\cite{anandan1994topological} and does bear upon the law adopted here. One further circumstance is worth recording, although it belongs to a theory beyond this one. At the level of the action a torsion couples minimally to Dirac fields through the axial-vector current, so that only its totally antisymmetric part is felt by fundamental matter at all~\cite{valle2022nieh}. The part which the Second Law spares is thus the part which can be probed, and the algebraic selection of Corollary~\ref{cor.antisym} agrees with the coupling that quantum field theory reports. The construction of this section predicts no such departure at the kinematic level, and any deviation would be attributable to a coupling instead of to the inertial structure.

The rejection of parallel transport deserves its reason in full, for the two candidate laws differ only away from inertial motion and the difference is measurable. Suppose $\nabla_v\Theta = 0$, take the manifold four-dimensional and oriented, and let $\Theta$ be the dual of a parallel vector field $Z$. A short computation then gives $\lVert S\rVert^2 = g(v,v)\,\lVert Z\rVert^2 - g(v,Z)^2$, so that along a unit-speed motion the combination $\lVert S\rVert^2 + g(v,Z)^2$ is constant and the magnitude of the spin varies precisely as the worldline tips with respect to $Z$, and it is conserved if and only if $g(a,Z) = 0$. A parallel $3$-form is a preferred direction upon the spacetime manifold and the variation is a preferred-frame effect, of exactly the kind the experiments cited above were built to bound and which those bounds exclude at the levels they reach. The transport~\eqref{eq.spinlaw} suffers none of this, since it preserves the magnitude by construction, and we adopted it upon that ground.

\section{Closing remarks}
\label{sec.closing}

We asked where the spacetime of Newtonian dynamics comes from. It comes from the description of motion, one structure at a time. We have assembled it from a differentiable manifold, a law for transporting vectors along its curves and a bilinear form, admitting each at the point where the description could not proceed without it. That discipline did more than organise the results; it produced them. The usual course posits the connection alongside the bilinear form; we derived it from a transport, and the Principle of Inertia emerged as a theorem rather than an independent postulate. We admitted the bilinear form without symmetry, without non-degeneracy and without any relation to the connection, and could then isolate the exact demand that inertia places upon the pair. We let the Second Law run with an arbitrary torsion, and could then exhibit torsion-freeness as the canonical representative of a physical demand, as opposed to a convenience.

The geometric register is self-contained, and we prove every statement in it. Inertial motion is invariant under affine reparametrisations of the worldline and under these alone (Theorem~\ref{theo.inertia}). The existence of inertial motions through every event and every direction is equivalent to the bilinear form being a Killing tensor of the connection, a condition strictly weaker than compatibility (Proposition~\ref{prop.killing}). The inertial motions see the pair $(g,\nabla)$ only through the homothety class of $g$ and the symmetric part of $\nabla$ (Theorem~\ref{theo.ise}). The force differs from the mass times the lowered acceleration by an exact term which a suitable choice of clock removes (Theorem~\ref{theo.nsl}), and the identity survives an arbitrary torsion at the cost of one further term quadratic in the velocity (Theorem~\ref{theo.nslgen}). Requiring the force-free particle of the dynamics to be the inertial particle of the kinematics spares exactly the torsions whose lowered form is totally antisymmetric (Corollary~\ref{cor.forcefreegen}, characterised in Corollary~\ref{cor.antisym}). That residual freedom is precisely a $3$-form (Proposition~\ref{prop.theta}), and its contraction with the velocity annihilates that velocity as an identity of the interior product (Proposition~\ref{prop.frenkel}). The torque obeys the commutator structure of the force one degree higher (Theorem~\ref{theo.torque}). We required no physical hypothesis for any of them.

Newton's First Law is three of these results at once. Definition~\ref{def.inertia} declares the privileged class of motions, and Proposition~\ref{prop.killing} supplies its existence through every event and every direction. That a body free of force keeps to that class is Corollary~\ref{cor.forcefree}, the agreement of the dynamical notion with the kinematic one. That \emph{uniform} means anything at all is Theorem~\ref{theo.inertia}, which leaves the inertial motion no freedom in its parameter beyond the choice of origin and unit. Newton stated the first two and could leave the third unsaid, since absolute time supplied it~\cite{newton1999principia}. The order is inverted here as well. The \emph{Principia} puts the First Law before the Second, whereas we reach the perseverance of a force-free body as a corollary of the Second.

We make three constitutive postulates, and we declared each one where it was made. A particle carries a representative of the homothety class of the bilinear form, and its inertial mass is the ratio to that representative (Definition~\ref{def.mass}). It carries a $3$-form along its worldline, the constitutive torsion, which the manifold does not furnish (Definition~\ref{def.consttorsion}). And that $3$-form obeys Fermi--Walker transport (Definition~\ref{def.spinlaw}).

One thing is conjectured. We propose the identification of the $2$-form $\iota_v\Theta$ with a physical spin as a candidate rather than a result, and we have not supplied its scale. One route to it already exists. The transport law is falsifiable, since the parallel alternative would make the magnitude of the spin vary with the orientation of the worldline, a preferred-frame effect which the existing torsion bounds exclude. The blindness of inertial motion to the spin is not a second route, since Theorem~\ref{theo.ise} guarantees it. The bounds and interferometric phases cited in the literature constrain a torsion field upon the manifold, whereas $\Theta$ is a $3$-form carried along a worldline. The transport law predicts a turning of the spin against a parallel-transported frame at a rate the acceleration alone fixes. That rate becomes a laboratory number only once the arc-length parameter is identified with the reading of a clock, and we do not identify it here. We set out that route and its sources in Section~\ref{sec.spin}.

The originality we claim is narrower than the material may suggest, and we make each concession in the body at the point where it falls due. The characterisation of the metric connections which preserve the geodesics of their metric by the total antisymmetry of their torsion belongs to Agricola and Friedrich~\cite{agricola2004holonomy,agricola2006srni}. The construction of a spin density by contracting a totally antisymmetric tensor with a velocity, and the supplementary condition which follows from that contraction alone, belong to relativistic spin hydrodynamics~\cite{fang2026relativistic}. Fermi--Walker transport belongs to Fermi and Walker~\cite{fermi1922sopra,walker1932relative}, and the identity relating the Lie and the covariant derivative of the metric dual of a velocity belongs to geometric hydrodynamics~\cite{arnold1998topological}. What we claim is the assignment rather than the algebra. The $3$-form is the residual freedom in the connection which inertial motion cannot resolve. We reach it from the Second Law, and the particle carries it along its worldline.

Two matters remain open, and they mark the boundary of what the present framework can reach. We select the transport~\eqref{eq.spinlaw} of the constitutive torsion by the conservation of a magnitude rather than derive it, so the spin, unlike the momentum, owes its equation of motion to a postulate the construction does not compel. We have assigned no meaning to the exterior derivative $\d\Theta$, and that is a consequence rather than an omission. Definition~\ref{def.consttorsion} placed the constitutive torsion upon the worldline, so that $\d\Theta$ is not defined at all. A field equation for the torsion cannot be posed until one passes from the single $3$-form carried by one particle to a continuum of them. That passage is the construction of a spin density from a totally antisymmetric spin tensor, which relativistic spin hydrodynamics already performs~\cite{fang2026relativistic}. The contact interaction between two spins would have to appear there, at the level of a continuum rather than of a single particle.

Three directions lead out of those boundaries. Whether the construction survives the degeneration of the bilinear form to a Galilean pair was posed in the Introduction, and it remains the first question we would put to it. What the arc-length parameter measures is the second, for the one quantitative prediction of the construction waits upon it. The third is the passage to a continuum just named. It is also the passage to the relativistic setting, where the spin density and its supplementary condition are already in place. We suspect that it and the Galilean question are one problem approached from two sides.

\bibliographystyle{unsrt}
\bibliography{force_bib}

\end{document}